%% file: sc20-fp.tex
\documentclass[conference]{IEEEtran}

\def\BibTeX{{\rm B\kern-.05em{\sc i\kern-.025em b}\kern-.08emT\kern-.1667em\lower.7ex\hbox{E}\kern-.125emX}}

\include{mac-includes}

\usepackage{cuted}

\usepackage{marginnote}

\usepackage[binary-units=true]{siunitx}
\usepackage{xspace}
\usepackage{verbatim}

\usepackage{fontawesome}
\usepackage{pifont}
\usepackage{textcomp}

\DeclareMathOperator{\rotl}{rotl}

\usepackage{xcolor}
\usepackage{soul}
\usepackage{dblfloatfix}

\newtheorem{theorem}{Theorem}

\newcommand{\BigO}{\mathcal{O}}
\DeclareMathOperator{\rank}{rank}

\usepackage{mathrsfs}

\usepackage{xcolor}
\definecolor{darkgrey}{RGB}{70,70,70}
\definecolor{lightgrey}{RGB}{200,200,200}

\usepackage{colortbl}

\usepackage{amsmath}

\DeclareMathOperator*{\argmin}{arg\,min}

\usepackage{graphicx}
\usepackage{balance}

\usepackage[switch]{lineno}

\usepackage{soul}

\usepackage{fontawesome}
\usepackage{pifont}
\usepackage{textcomp}

\usepackage{makecell}

\usepackage{tikz}
\usetikzlibrary{tikzmark}

\usepackage{xpatch}
\expandafter\xpatchcmd
\csname pgfk@/tikz/every picture/.@cmd\endcsname
{\thepage}{\arabic{page}}{}{}

\usepackage{varwidth}

\newif\iftr
\newif\ifsc

\trtrue
\scfalse

\definecolor{lyellow}{RGB}{255,255,0}
\definecolor{llyellow}{RGB}{250, 250, 0}
\definecolor{lgreen}{RGB}{144,238,144}

\let\oldhl\hl
\renewcommand{\hl}[1]{{\sethlcolor{llyellow}\oldhl{#1}}}

\DeclareRobustCommand{\hll}[1]{{\sethlcolor{llyellow}\hl{#1}}}

\begin{document}


\title{FatPaths: Routing in Supercomputers \\ and Data Centers when Shortest Paths Fall Short}

\author{\IEEEauthorblockN{Maciej Besta$^1$, Marcel Schneider$^1$, Karolina Cynk$^2$, Marek Konieczny$^2$,\\ Erik Henriksson$^1$, Salvatore Di Girolamo$^1$, Ankit Singla$^1$, Torsten Hoefler$^1$}
\IEEEauthorblockA{
\textit{$^1$Department of Computer Science, ETH Zurich}\\
\textit{$^2$Faculty of Computer Science, Electronics and Telecommunications, AGH-UST
\iftr
\vspace{0.5em}
\fi
} \\
}
}

\maketitle

\thispagestyle{plain}
\pagestyle{plain}

\begin{abstract}
We introduce FatPaths: a simple, generic, and robust routing architecture that
enables state-of-the-art low-diameter topologies such as Slim Fly to achieve
unprecedented performance. FatPaths targets Ethernet stacks in both HPC
supercomputers as well as cloud data centers and clusters.
FatPaths exposes and exploits the rich (``fat'') diversity of both minimal and
non-minimal paths for high-performance multi-pathing. Moreover, FatPaths
features a redesigned ``purified'' transport layer that removes virtually all
TCP performance issues (e.g., the slow start), and uses flowlet switching, a
technique used to prevent packet reordering in TCP networks, to enable very
simple and effective load balancing.
Our design enables recent low-diameter topologies to outperform powerful Clos
designs, achieving 15\% higher net throughput at 2$\times$ lower latency for
comparable cost.
FatPaths will significantly accelerate Ethernet clusters that form more than
50\% of the Top500 list and it may become a standard routing scheme for modern
topologies.
\end{abstract}

\iftr
\vspace{1em}
\fi

\iftr\vspace{1.5em}\fi
{\noindent\small\bf This is the full version of a paper published at\\
 ACM/IEEE Supercomputing'20 under the same title}
\iftr\vspace{1em}\fi

\input{body.tex}

\bibliographystyle{IEEEtran}


\bibliography{../../bibl_conf}

\iftr
\appendices

\input{appendix.tex}

\fi


\end{document}

%% file: mac-includes.tex
\newif\ifall
\allfalse

\newif\ifsq     
\sqfalse

\newcommand{\vspaceSQ}[1]{\ifsq\vspace{#1}\fi}

\usepackage{etex}

\usepackage{makecell,cellspace}
\usepackage{graphicx}

\usepackage[10pt]{moresize}

\ifsq
\usepackage[font={scriptsize}]{caption}
\usepackage[font={scriptsize}]{subcaption}
\else
\usepackage[font={small}]{caption}
\usepackage[font={small}]{subcaption}
\fi

\usepackage{amsthm}
\usepackage{float}

\usepackage{graphicx}
\usepackage{booktabs}
\usepackage{epstopdf}
\usepackage{url}
\usepackage{placeins}
\usepackage{amsmath,amsfonts}
\usepackage{mathtools,mathrsfs}
\usepackage{cite}
\usepackage{multirow}
\usepackage{rotating}
\usepackage{pbox}
\usepackage[normalem]{ulem}
\usepackage{inconsolata}
\usepackage{listings}

\ifsq
\newcommand{\subparagraph}{}
\usepackage[compact]{titlesec}
\titlespacing*{\section}{0pt}{6pt}{2pt}
\titlespacing*{\subsection}{0pt}{5pt}{1pt}
\titlespacing*{\subsubsection}{0pt}{5pt}{1pt}
\else
\newcommand{\subparagraph}{}
\usepackage[compact]{titlesec}
\titlespacing*{\section}{0pt}{9pt}{5pt}
\titlespacing*{\subsection}{0pt}{9pt}{4pt}
\titlespacing*{\subsubsection}{0pt}{9pt}{4pt}
\fi

\usepackage{cleveref}
\usepackage[utf8]{inputenc}
\crefname{section}{§}{§§}
\Crefname{section}{§}{§§}

\usepackage{xcolor}
\definecolor{darkgrey}{RGB}{70,70,70}
\definecolor{lightgrey}{RGB}{200,200,200}

\usepackage{enumitem}

\ifsq
\else
\fi

\usepackage[linesnumbered,ruled]{algorithm2e}
\usepackage{multicol}
\SetKwComment{Comm}{$\triangleright$\ }{}
\SetAlFnt{\scriptsize}
\SetAlCapFnt{\scriptsize}
\SetAlCapNameFnt{\scriptsize}
\SetKwInOut{Input}{Input}
\SetKwInOut{Output}{Output}

\newcommand{\maciej}[1]{\textcolor{blue}{[Maciej: #1]}}

\newcommand{\macb}[1]{\textbf{{#1}}}

\usepackage{scalerel,stackengine}
\stackMath
\newcommand\rwh[1]{%
\savestack{\tmpbox}{\stretchto{%
  \scaleto{%
        \scalerel*[\widthof{\ensuremath{#1}}]{\kern-.6pt\bigwedge\kern-.6pt}%
                  {\rule[-\textheight/2]{1ex}{\textheight}}
                              }{\textheight}%
}{0.5ex}}%
\stackon[1pt]{#1}{\tmpbox}%
}

\usepackage[hang,flushmargin]{footmisc}

\renewcommand{\epsilon}{\ensuremath\varepsilon}

\renewcommand{\phi}{\ensuremath{\varphi}}

\usepackage{algcompatible}

\algblockdefx{FORALLP}{ENDFORALLP}[1]%
{\textbf{for all }#1 \textbf{do in parallel}}%
{\textbf{end for}}

\usepackage{tabulary}

\makeatletter
\newcommand{\removelatexerror}{\let\@latex@error\@gobble}
\makeatother



%% file: body.tex
\section{Introduction}
\label{sec:intro}
\ifsq
\vspace{-0.25em}
\fi

\iftr
\begin{figure*}[t]
\vspaceSQ{-1em}
\centering
\includegraphics[width=0.28\textwidth]{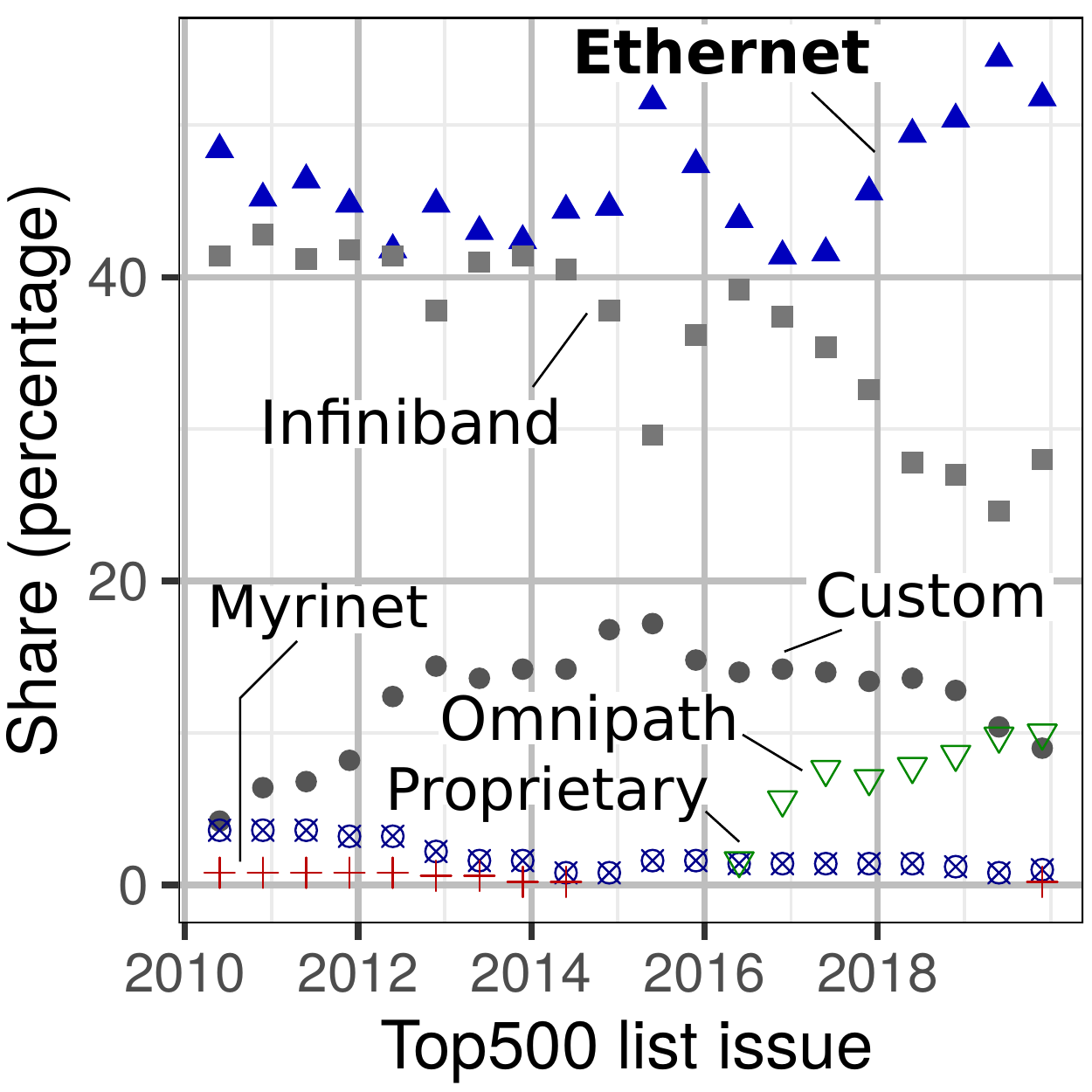}
%
%
\includegraphics[width=0.335\textwidth]{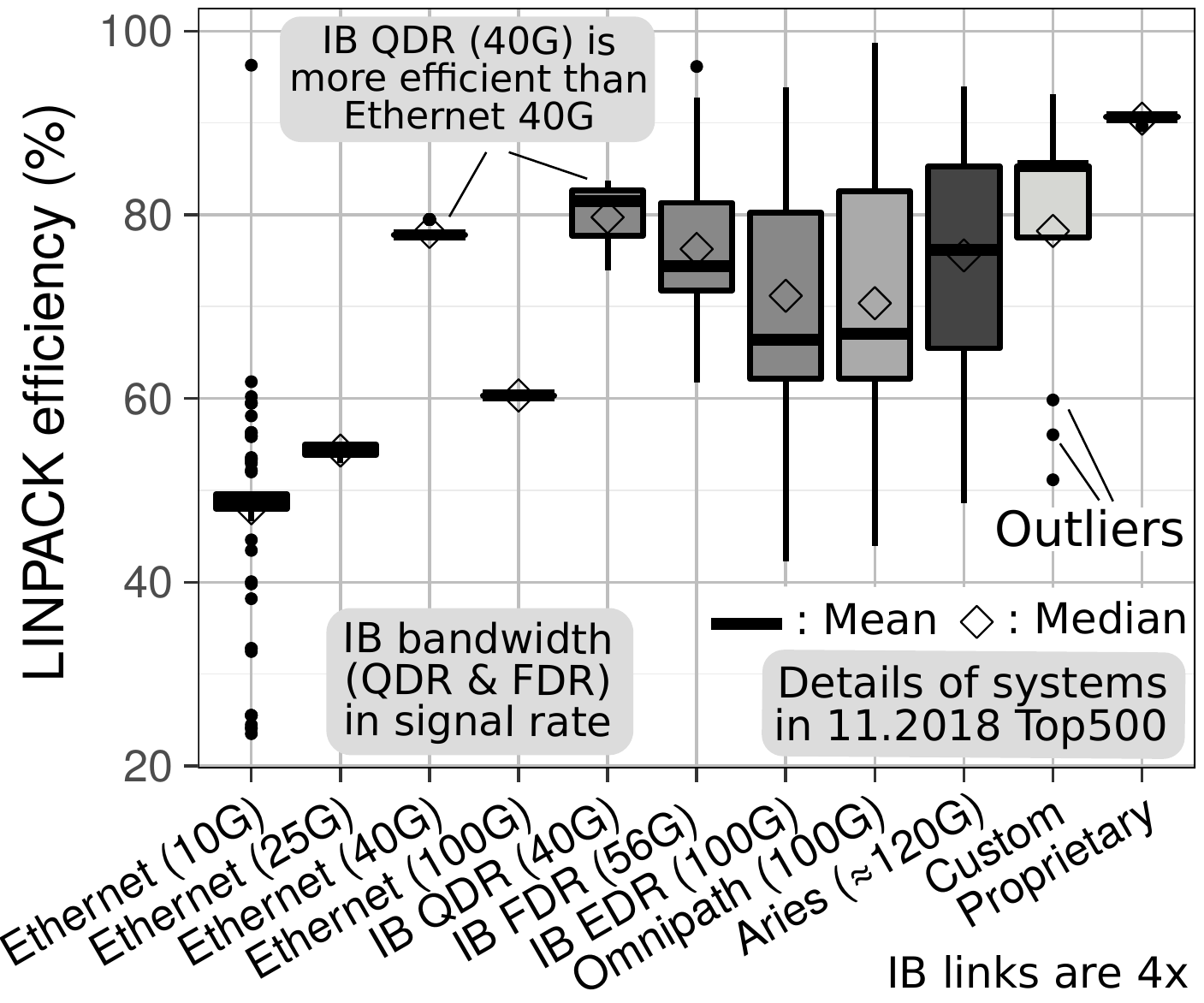}
\includegraphics[width=0.335\textwidth]{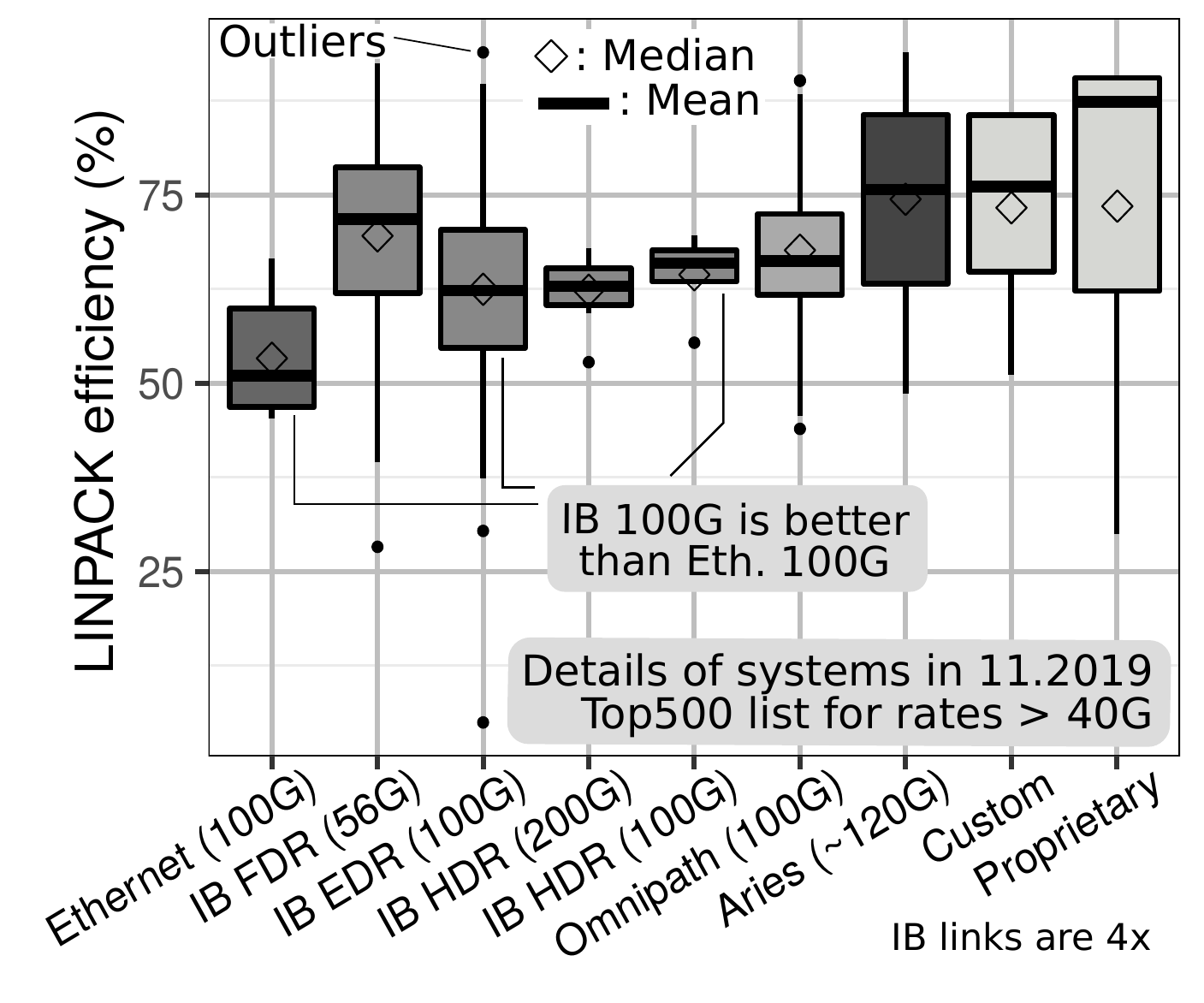}
\vspaceSQ{-0.5em}
\caption{\textmd{The percentage of Ethernet systems in the Top500 list (on the
left) and the LINPACK efficiency of Top500 systems with various networks
in November 2018 and 2019 Top500 issues (on
the right).
}}
\vspaceSQ{-1em}
\label{fig:motivation}
\end{figure*}
\fi

\iftr
\begin{figure*}[b]
\vspaceSQ{-1em}
\centering
\includegraphics[width=0.6\textwidth]{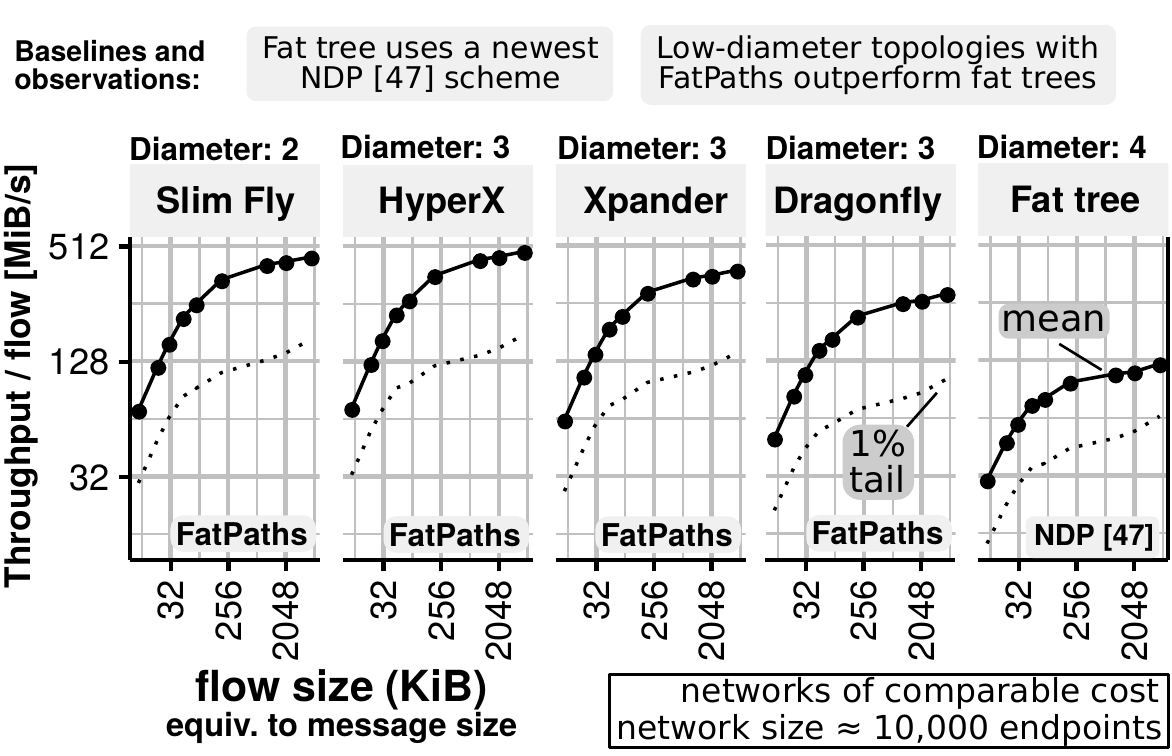}
\vspaceSQ{-0.5em}
\caption{\textmd{Example \textbf{performance advantages of low-diameter
topologies that use FatPaths over fat trees} equipped with NDP (very recent
routing architecture by Handley et al.~\cite{handley2017re}). Evaluation
methodology is discussed in detail in~\cref{sec:eval}.}}
%
%
\vspace{-0.5em}
\label{fig:ndp_results_motiv}
\end{figure*}
\fi

\ifsq
\enlargethispage{\baselineskip}
\fi

Ethernet continues to be important in the HPC landscape. While the most
powerful Top500 systems use vendor-specific or Infiniband (IB) interconnects,
more than half of the Top500 (cf.~the Nov.~2019 issue)
machines~\cite{dongarra1997top500} are based on Ethernet, see
Figure~\ref{fig:motivation} (the left plot). For example, Mellanox connects
156 Ethernet systems (25 GiB and faster), which is 20\% more than in Nov.~2018.
The Green500 list is similar~\cite{feng2007green500}. 
The importance of Ethernet is increased by the {``convergence of HPC and
Big Data''}, with cloud providers and data center operators aggressively aiming
for high-bandwidth and low-latency fabric~\cite{valadarsky2015, handley2017re,
  vanini2017letflow}. An example is the growing popularity of RDMA over
  Converged Ethernet (RoCE)~\cite{infiniband2014rocev2} that facilitates
  deploying Remote Direct Memory Access (RDMA)~\cite{fompi-paper} applications
  and protocols -- traditionally associated with HPC and IB interconnects -- on
  top of Ethernet.

\ifall
An example is Mellanox, with its Ethernet sales for the 3rd quarter of 2017 being
higher than those for Infiniband~\cite{mellanox-sales}.
\fi

Yet, Ethernet systems are scarce in the \emph{highest 100} positions of Top500. For
example, in November 2019, only \emph{six} such systems were among the highest
100. Ethernet systems are also less efficient than Infiniband, custom,
OmniPath, and proprietary systems, see Figure~\ref{fig:motivation} (on the
right). This is also the case for systems with similar sizes, injection
bandwidth, and topologies, indicating overheads caused by {routing}.
Thus, enhancing routing in HPC Ethernet clusters would improve the overall
performance of $\approx$50\% of Top500 systems and accelerate cloud
infrastructure, mostly based on Ethernet~\cite{azodolmolky2013cloud}.

%
Clos dominates the landscape of cloud data centers and
supercomputers~\cite{niranjan2009portland, handley2017re, valadarsky2015}.
Yet, many recent \emph{{low-diameter topologies}} claim to outperform Clos in
the cost-performance tradeoff. For example, Slim Fly can be $\approx$$2\times$
more cost- and power-efficient while offering $\approx$25\% lower latency.
Similar numbers were reported for Jellyfish~\cite{singla2012jellyfish} and
Xpander~\cite{valadarsky2015}. Thus, \emph{modern low-diameter networks may
significantly enhance compute capabilities of Ethernet clusters.}

However, to the best of our knowledge, {no high-performance routing
architecture has been proposed for low-diameter networks} based on Ethernet
stacks. The key issue here is that {traditional routing schemes} (e.g.,
Equal-Cost Multipath (ECMP)~\cite{hopps2000analysis}) {cannot be directly used
in networks such as Slim Fly}, because (as we will show) there is almost always
{only one} shortest path between endpoint pairs (i.e., \emph{shortest
paths fall short}). Restricting traffic to these paths does not utilize such
topologies' path diversity, and it is unclear how to split traffic across
non-shortest paths of \emph{unequal} lengths. 

\ifsq\enlargethispage{\baselineskip}\fi

To answer this, {we propose {FatPaths}, the first high-performance, simple, and
robust {routing architecture} for low-diameter networks}, to accelerate {both}
HPC systems and cloud infrastructure that use Ethernet. 
FatPaths uses our \emph{key research outcome}: {{although low-diameter networks
fall short of shortest paths, they have enough ``almost'' shortest paths}}.
This insight comes from our in-depth analysis of path diversity in five
low-diameter topologies (contribution~\textbf{\#1}). 
Then, in our \emph{key design outcome}, we show how to encode this {rich
diversity of non-minimal paths in low-diameter networks} in commodity hardware
(HW) using a novel routing approach called {{layered routing}}
(contribution~\textbf{\#2}). 
Here, we divide network links into subsets called {layers}. One layer contains
at least one ``almost'' shortest path between any two endpoints. Non-minimal
multipath routing is then enabled by using more than one layer.

For higher performance in TCP environments, FatPaths uses a {{high-performance
transport design}}. 
We seamlessly combine layered routing with several {TCP
enhancements}~\cite{handley2017re} 
\iftr
(such as lossless metadata exchange -- packet headers always reach their
destinations, shallow buffers, and priority queues for retransmitted packets to
avoid head-of-line congestion~\cite{handley2017re})
\fi
and we use {flowlet switching}~\cite{kandula2007dynamic}, a scheme that enables
{very simple} but {powerful} load balancing by sending batches of packets over
multiple layers (contribution~\textbf{\#3}).

\ifall\m{?}
This exposed path diversity enables FatPaths to
benefit from multipath routing in low-diameter networks.
\fi

We exhaustively compare FatPaths to other routing schemes in
Table~\ref{tab:intro} (contribution~\textbf{\#4}). FatPaths is {the only
scheme} that simultaneously (1) enables multi-pathing using both (2) shortest
and (3) non-shortest paths, (4) explicitly considers disjoint paths, (5) offers
adaptive load balancing, and (6) is applicable across topologies.
Table~\ref{tab:intro} focuses on path diversity, because, as topologies lower
their diameter and reduce link count, path diversity -- key to high-performance
routing -- {becomes a scarce resource demanding careful examination}.

\ifall

For wide applicability, we also integrate FatPaths with with Data Center TCP
(DCTCP)~\cite{alizadeh2011data}. Next, we discuss integration with
Multipath TCP (MPTCP)~\cite{raiciu2011improving},
RDMA~\cite{fompi-paper} (iWARP~\cite{iwarp}, RoCE~\cite{infiniband2014rocev2}),
and Infiniband~\cite{pfister2001introduction}.
%
%
Finally, we discuss how to deploy FatPaths' non-minimal routing using a simple
ECMP hardware~\cite{hopps2000analysis}. This effectively shows
\emph{{how to use ECMP also with non-minimal paths}}, which is of
separate interest.

\fi

\iftr
We discuss integration of FatPaths not only with Ethernet and TCP,
but also with Data Center TCP (DCTCP)~\cite{alizadeh2011data} and Multipath TCP
(MPTCP)~\cite{raiciu2011improving}. Next, we
discuss how FatPaths could enhance Infiniband~\cite{pfister2001introduction}
and sketch implementing RDMA~\cite{fompi-paper}
(iWARP~\cite{iwarp}, RoCE~\cite{infiniband2014rocev2}) on top of FatPaths.
%
%
%
\fi

\begin{table}[t]
%
%
\setlength{\tabcolsep}{2pt}
\ifsq
\renewcommand{\arraystretch}{0.5}
\else
\renewcommand{\arraystretch}{1.2}
\fi
\centering
\ssmall
\sf
\begin{tabular}{lllllllll}
\toprule
\multirow{2}{*}{ \makecell[c]{\textbf{Routing Scheme}\\\textbf{(Name, Abbreviation, Reference)}}} & \multirow{2}{*}{\makecell[c]{\textbf{Stack}\\\textbf{Layer}}} & \multicolumn{7}{c}{Supported path diversity aspect} \\
\cmidrule(lr){3-9}
 & & \textbf{SP} & \textbf{NP} & \textbf{SM} & \textbf{MP} & \textbf{DP} & \textbf{ALB} & \textbf{AT} \\
\midrule
\multicolumn{9}{c}{\textbf{(1) SIMPLE ROUTING PROTOCOLS (often used as building blocks):}} \\
\midrule
Valiant load balancing (VLB)~\cite{valiant1982scheme} & L2--L3 & \faThumbsDown & \faThumbsOUp & \faThumbsDown & \faThumbsDown & \faThumbsDown & \faThumbsDown & \faThumbsOUp  \\
\makecell[l]{Simple Spanning Tree (ST)~\cite{perlman1985algorithm}} & L2 & \faThumbsDown$^S$ & \faThumbsDown$^S$ & \faThumbsDown & \faThumbsDown & \faThumbsDown & \faThumbsDown & \faThumbsOUp  \\ 
\makecell[l]{Simple routing, e.g., OSPF~\cite{moy1997ospf, rekhter2005border, malkin1994rip, oran1990osi}} & L2, L3 & \faThumbsOUp & \faThumbsDown & \faThumbsDown & \faThumbsDown & \faThumbsDown & \faThumbsDown & \faThumbsOUp \\
UGAL~\cite{kim2008technology} & L2--L3 & \faThumbsOUp & \faThumbsOUp & \faThumbsDown & \faThumbsDown & \faThumbsDown & \faThumbsOUp & \faThumbsOUp \\ 
ECMP~\cite{hopps2000analysis}, OMP~\cite{villamizar1999ospf}, Pkt.~Spraying~\cite{dixit2013impact} & L2, L3 & \faThumbsOUp & \faThumbsDown & \faThumbsDown & \faThumbsOUp & \faThumbsDown & \faThumbsDown & \faThumbsOUp  \\ 
%
\midrule
\multicolumn{9}{c}{\textbf{(2) ROUTING ARCHITECTURES:}}
\\
\midrule
DCell~\cite{guo2008dcell} & L2--L3 & \faThumbsDown & \faThumbsOUp & \faThumbsDown & \faThumbsDown & \faThumbsDown & \faThumbsDown & \faThumbsDown \\
Monsoon~\cite{greenberg2008towards} & L2, L3 & \faThumbsUp & \faThumbsUp & \faThumbsDown & \faThumbsUp & \faThumbsDown & \faThumbsDown & \faThumbsDown \\
PortLand~\cite{niranjan2009portland} & L2 & \faThumbsOUp & \faThumbsDown & \faThumbsDown & \faThumbsOUp & \faThumbsDown & \faThumbsDown & \faThumbsDown \\ 
DRILL~\cite{ghorbani2017drill}, LocalFlow~\cite{sen2013localflow}, DRB~\cite{cao2013per} & L2 & \faThumbsOUp & \faThumbsDown & \faThumbsDown & \faThumbsOUp & \faThumbsDown & \faThumbsOUp & \faThumbsDown \\ 
VL2~\cite{greenberg2009vl2} & L3 & \faThumbsOUp & \faThumbsDown & \faThumbsDown & \faThumbsOUp & \faThumbsDown & \faThumbsUp & \faThumbsDown \\ 
Architecture by Al-Fares et al.~\cite{alfares2008scalable} & L2--L3 & \faThumbsOUp & \faThumbsDown & \faThumbsDown & \faThumbsOUp & \faThumbsOUp & \faThumbsOUp & \faThumbsDown \\
BCube~\cite{guo2009bcube} & L2--L3 & \faThumbsOUp & \faThumbsDown & \faThumbsDown & \faThumbsOUp & \faThumbsOUp & \faThumbsDown & \faThumbsDown \\
%
%
SEATTLE~\cite{kim2008floodless}, others$^*$~\cite{lui2002star, rodeheffer2000smartbridge, perlman2004rbridges, garcia2003lsom} & L2 & \faThumbsOUp & \faThumbsDown & \faThumbsDown & \faThumbsDown & \faThumbsDown & \faThumbsDown & \faThumbsOUp \\
VIRO~\cite{jain2011viro} & L2--L3 & \faThumbsDown$^S$ & \faThumbsDown$^S$ & \faThumbsDown & \faThumbsDown & \faThumbsDown & \faThumbsDown & \faThumbsOUp \\
Ethernet on Air~\cite{sampath2010ethernet} & L2 & \faThumbsDown$^S$ & \faThumbsDown$^S$ & \faThumbsDown & \faThumbsUp$^R$ & \faThumbsDown & \faThumbsDown & \faThumbsOUp \\
PAST~\cite{stephens2012past} & L2 & \faThumbsUp$^S$ & \faThumbsUp$^S$ & \faThumbsDown & \faThumbsDown & \faThumbsOUp & \faThumbsDown & \faThumbsOUp \\ 
\makecell[l]{MLAG, MC-LAG, others~\cite{subramanian2014multi}} & L2 & \faThumbsUp & \faThumbsUp & \faThumbsDown & \faThumbsUp$^R$ & \faThumbsDown & \faThumbsDown & \faThumbsOUp \\ 
MOOSE~\cite{scott2009addressing} & L2 & \faThumbsOUp & \faThumbsDown & \faThumbsDown & \faThumbsUp & \faThumbsDown & \faThumbsDown & \faThumbsOUp \\
MPA~\cite{narvaez1999efficient} & L3 & \faThumbsOUp & \faThumbsOUp & \faThumbsDown & \faThumbsOUp & \faThumbsDown & \faThumbsDown & \faThumbsOUp \\
AMP~\cite{gojmerac2003adaptive} & L3 & \faThumbsOUp & \faThumbsDown & \faThumbsDown & \faThumbsOUp & \faThumbsDown & \faThumbsOUp & \faThumbsOUp \\
MSTP~\cite{de2006improving}, GOE~\cite{iwata2004global}, Viking~\cite{sharma2004viking} & L2 & \faThumbsDown$^S$ & \faThumbsDown$^S$ & \faThumbsDown & \faThumbsOUp & \faThumbsDown & \faThumbsDown & \faThumbsOUp \\ 
\makecell[l]{SPB~\cite{allan2010shortest}, TRILL~\cite{touch2009transparent}, Shadow MACs~\cite{agarwal2014shadow}} & L2 & \faThumbsOUp & \faThumbsDown$^R$ & \faThumbsDown & \faThumbsOUp & \faThumbsDown & \faThumbsDown & \faThumbsOUp \\ 
%
%
%
SPAIN~\cite{mudigonda2010spain} & L2 & \faThumbsUp$^S$ & \faThumbsUp$^S$ & \faThumbsUp$^S$ & \faThumbsOUp & \faThumbsOUp & \faThumbsDown & \faThumbsOUp \\ 
\iftr
\midrule
\multicolumn{9}{c}{\textbf{(3) Schemes for exposing/encoding paths (can be combined with FatPaths):}}
\\
\midrule

XPath~\cite{hu2016explicit} & L3 & \faThumbsOUp & \faThumbsUp & \faThumbsUp & \faThumbsOUp & \faThumbsOUp & \faThumbsUp & \faThumbsOUp \\
Source routing for flexible DC fabric~\cite{jyothi2015towards} & L3 & \faThumbsOUp & \faThumbsUp$^R$ & \faThumbsUp$^R$ & \faThumbsDown & \faThumbsDown & \faThumbsDown & \faThumbsUp$^{\text{\textdagger}}$ \\
\fi
\midrule
%
%
%
\textbf{(3) FatPaths~[This work]} & L2--L3 & \faThumbsOUp & \faThumbsOUp & \faThumbsOUp & \faThumbsOUp & \faThumbsOUp & \faThumbsOUp & \faThumbsOUp \\
\midrule
%
%
%
\bottomrule
\end{tabular}
\vspace{-0.5em}
\caption{
\textmd{
\textbf{Support for path diversity in routing schemes}.
\iftr
``Stack Layer'' indicates the location of each routing scheme in the TCP/IP stack.
\fi
%
%
\textbf{SP}, \textbf{NP}: support for arbitrary \textbf{shortest} and \textbf{non-minimal} paths, respectively.
%
%
%
\textbf{SM}: A given scheme \textbf{simultaneously} enables minimal and non-minimal paths.
\textbf{MP}: support for \textbf{multi-pathing} (between two hosts).
\textbf{DP}: support for \textbf{disjoint} paths.
\textbf{ALB}: support for \textbf{adaptive load balancing}.
\textbf{AT}: compatible with an \textbf{arbitrary topology}.
%
%
\faThumbsOUp, \faThumbsUp, \faThumbsDown: A given scheme, respectively, offers
a given feature, offers it in a limited way (e.g.,
Monsoon~\cite{greenberg2008towards} uses multi-pathing (ECMP) only between
border and access routers), and does not offer it.
$^R$A given feature is offered \emph{only} for resilience (\emph{not}
performance).
\ifsq
$^S$Minimal or non-minimal paths are offered \emph{only} within spanning trees.
\else
$^S$Minimal or non-minimal paths are offered \emph{only} within one or more
spanning trees that do not necessarily cover \emph{all} physical paths.
\fi
}}
\vspaceSQ{-1.5em}
\label{tab:intro}
\end{table}

We conduct extensive, large-scale packet-level simulations
(contribution~\textbf{\#5}), and a comprehensive theoretical analysis
(contribution~\textbf{\#6}). {We simulate topologies with up to
\emph{$\approx$1 million endpoints}}. We motivate FatPaths in
Figure~\ref{fig:ndp_results_motiv}. Slim Fly and Xpander equipped with FatPaths
ensure $\approx$15\% higher throughput and $\approx$2$\times$ lower latency
than similar-cost fat trees.
%

\ifsq\enlargethispage{\baselineskip}\fi

We stress that FatPaths outperforms the bleeding-edge Clos proposals based on
per-packet load balancing (these schemes account for packet-reordering) and
novel transport mechanisms, that {achieve $3$-$4\times$ smaller tail
flow\footnote{\scriptsize In performance analyses, we use the term ``flow'',
which is equivalent to a ``message''.} completion time (FCT) than Clos based on
ECMP}~\cite{handley2017re, ghorbani2017drill}.
Consequently, our work illustrates that \emph{low-diameter networks can
continue to claim an improvement in the cost-performance tradeoff against the
new, superior Clos baselines} (contribution~\textbf{\#7}). 

\ifall

Finally, we analyze how to deploy FatPaths' non-minimal routing using a simple
ECMP hardware~\cite{hopps2000analysis}. This effectively shows
\emph{\textbf{how to use ECMP also with non-minimal paths}}, which is -- due to
the prevalence of the ECMP hardware in data centers -- of separate interest
(\textbf{contribution~\#8}).

\fi

\iftr
\begin{table*}[b]
\vspaceSQ{-0.5em}
\sffamily
\centering
\footnotesize
\begin{tabular}{@{}lll@{}}
\toprule
%
%
\iftr
\multirow{8}{*}{\begin{turn}{90}\shortstack{\textbf{Network }\\\textbf{structure}}\end{turn}} 
\fi
\ifsc
\multirow{5}{*}{\begin{turn}{90}\shortstack{\textbf{Network }\\\textbf{model}}\end{turn}} 
\fi
                   & $V, E$ & Sets of vertices/edges (routers/links, $V=\{0,\dots,N_r-1\}$).\\
\ifsc
                   & $N, N_r$& \#endpoints and \#routers in the network ($N_r = |V|$).\\
                   & $p, k'$& \#endpoints attached to a router, \#channels to other routers.\\
\fi
\iftr
                   & $N$& The number of endpoints in the network.\\
                   & $N_r$& The number of routers in the network ($N_r = |V|$).\\
                   & $p$& The number of endpoints attached to a router.\\
                   & $k'$& The number of channels from a router to other routers.\\
\fi
                   \iftr
                   & $k$&\emph{Router radix} ($k = k' + p$).\\
                   \fi
                   & $D, d$&Network diameter and the average path length.\\
                   \ifall
                  & & $q, a, h, \ell, L, S, K$: input parameters of various topologies \\ 
\fi
%
\midrule
\iftr
\multirow{7}{*}{\begin{turn}{90}\shortstack{\textbf{Diversity of}\\\textbf{paths (\cref{sec:paths})}}\end{turn}} 
\fi
\ifsc
\multirow{4}{*}{\begin{turn}{90}\shortstack{\textbf{Paths}\\\textbf{(\cref{sec:paths})}}\end{turn}} 
\fi
                   & $x \in V$ & Different routers used in~\cref{sec:paths} ($x \in \{s,t,a,b,c,d\}$).\\
\iftr
                   & $X \subset V$ & Different router sets used in~\cref{sec:paths} ($X \in \{A,B\}$).\\
\fi
                   & $c_l(A,B)$ & \emph{Count of (at most $l$-hop) disjoint paths} between router sets $A$, $B$.\\
                   & $c_\text{min}(s,t), l_\text{min}(s,t)$ & \emph{Diversity} and \emph{lenghts of minimal paths} between routers $s$ and $t$.\\
\iftr
                   & $l_\text{min}(s,t)$ & \emph{Lengths of minimal paths} between routers $s$ and $t$.\\
\fi
                   & $I_{ac,bd}$ & \emph{Path interference} between pairs of routers $a,b$ and $c,d$.\\
\midrule
\multirow{1}{*}{\begin{turn}{90} \shortstack{\hspace{-0.4em}\textbf{ Layers}\\\textbf{(\cref{sec:routing})}}\end{turn}} 
                   & $n$ & The total number of layers in FatPaths routing.\\
                   & $\sigma_i$ & A layer, defined by its forwarding function, $i\in \{1,\dots,n\}$.\\
                   & $\rho$ & Fraction of edges used in routing.
                   \vspace{0.1em}\\ 

%
%
                   \bottomrule
\end{tabular}
\vspaceSQ{-0.5em}
\caption{\textmd{The \textbf{most important symbols} used in this work.}}
\vspaceSQ{-1.5em}
\label{tab:symbols}
\end{table*}
\fi

\iftr

\noindent
\vspace{1em}

Towards the above goals, we contribute:

\vspace{0.5em}

\begin{itemize}[noitemsep, leftmargin=0.75em]
\item A {high-performance, simple, and robust \textbf{routing architecture, FatPaths}},
that enables modern low-diameter topologies such as Slim Fly to
achieve unprecedented performance.
\item A novel routing approach called \textbf{layered routing} that is a key
ingredient of FatPaths and facilitates using diversity of non-minimal paths in
modern low-diameter networks.
\item The first \textbf{detailed analysis of path diversity} in five modern low-diameter
network topologies, and the identification of the diversity of non-minimal
paths as {a key resource} for their high performance.
\item A \textbf{novel path diversity metric}, Path Interference, that captures
bandwidth loss between specific pairs of routers. 

%
\item A comprehensive {analysis of existing routing schemes} in terms of their
support for path diversity.
\item A \textbf{theoretical analysis} showing FatPaths' advantages.
%
%
%
%
\item Extensive, large-scale packet-level simulations (up to 
$\approx$\textbf{one million endpoints}) to demonstrate the advantages of
low-diameter network topologies equipped with FatPaths over very recent Clos
designs, achieving 15\% higher net throughput at 2$\times$ lower latency for
comparable cost.
\end{itemize}

\fi

%
\section{Notation, Background, Concepts}
\label{sec:back}


We first outline basic concepts; 
Table~\ref{tab:symbols} shows the notation. 
%


\subsection{{Network Model}}
We model an interconnection network as an undirected graph $G = (V,E)$; $V$ and
$E$ are sets of {routers\footnote{\scriptsize We abstract away HW details and
use a term ``router'' for L2 switches and L3 routers.}} ($|V| = N_r$) and
full-duplex inter-router physical links. Endpoints are \emph{not} modeled
explicitly. There are $N$ endpoints, $p$~endpoints are attached to
each router (\emph{concentration}) and $k'$ channels from each router to other
routers (\emph{network radix}). The total router \emph{radix} is $k = p+k'$.
The diameter is $D$ while the average path length is $d$.

\ifsq\enlargethispage{\baselineskip}\fi

\subsection{{Topologies and Fair Topology Setup}}
\label{sec:back_topos}
We consider all recent low-diameter networks: {Slim Fly}
(SF)~\cite{besta2014slim} (a variant with $D=2$), {Dragonfly}
(DF)~\cite{kim2008technology} (the ``balanced'' variant with $D=3$),
{Jellyfish} (JF)~\cite{singla2012jellyfish} (with $D=3$), {Xpander}
(XP)~\cite{valadarsky2015} (with $D \le 3$), {HyperX} (Hamming graph)
(HX)~\cite{ahn2009hyperx} that generalizes {Flattened Butterflies}
(FBF)~\cite{kim2007flattened} with $D=3$.  We also
use established three-stage {fat trees} (FT3)~\cite{leiserson1996cm5} that are
a variant of {Clos}~\cite{clos1953study}. 
Note that we do not detail the considered topologies. This is because our
design does \emph{not} rely on any specifics of these
networks (i.e., FatPaths can be used in \emph{any} topology, but from
performance perspective, it is most beneficial for low-diameter networks).
\iftr
The considered topologies are illustrated in Table~\ref{tab:parameters}.
\fi

\iftr

\begin{table}[h!]
%
\centering
\scriptsize
\sf
\setlength{\tabcolsep}{1pt}
\begin{tabular}{llcll}
\toprule
\textbf{Topology} & \textbf{Structure remarks} & $D$ & \textbf{Variant} & \textbf{Deployed?} \\
\midrule
%
%
%
\makecell[l]{Slim Fly (SF)~\cite{besta2014slim}} & \makecell[l]{Consists of groups}  & 2 & \makecell[l]{MMS} & unknown \\
\makecell[l]{Dragonfly (DF)~\cite{kim2008technology}} & \makecell[l]{Consists of groups} & 3 & \makecell[l]{``balanced''} & \makecell[l]{PERCS~\cite{arimilli2010percs},\\Cascade~\cite{faanes2012cray}} \\
%
%
\makecell[l]{HyperX\\ (HX, HX3)~\cite{ahn2009hyperx}} & \makecell[l]{Consists of groups} & 3 & \makecell[l]{``regular'' (cube)} & unknown \\
%
%
%
\makecell[l]{Xpander (XP)~\cite{valadarsky2015}} & \makecell[l]{Consists of metanodes}  & $\le$3 & randomized & unknown \\
\makecell[l]{Jellyfish (JF)~\cite{singla2012jellyfish}} & \makecell[l]{Random network}  & $\le$3 & \makecell[l]{``homogeneous''} & unknown \\
\makecell[l]{Fat tree (FT)~\cite{leiserson1996cm5}} & \makecell[l]{Endpoints form pods} & 4 & \makecell[l]{3 router layers} & \makecell[l]{Many systems} \\
\bottomrule
\end{tabular}
\caption{\textmd{Used topologies.}}
\vspaceSQ{-1.5em}
\label{tab:parameters}
\end{table}

\fi

%
%
We use four size classes: small ($N \approx 1,000$), medium ($N \approx
10,000$), large ($N \approx 100,000$), and huge ($N \approx 1,000,000$).  We
set $p = \frac{k'}d$ 
\iftr
(in Section~\ref{sec:eval} we show that, assuming random uniform traffic, $p =
\frac{k'}{d}$ maximizes throughput while minimizing congestion and network
cost).
\fi
Third, we select network radix~$k'$ and router
count~$N_r$ so that, for a fixed $N$, the compared topologies use similar
amounts of hardware for similar construction costs.



%
Jellyfish -- unlike other topologies -- is ``fully flexible'': There is a JF
instance for each combination of $N_r$ and $k'$. Thus, to fully test JF,
{for each other network}~X, {we use an equivalent} JF (denoted as
X-JF) with identical $N_r, k'$.

\begin{figure*}[t]
\vspaceSQ{-1.6em}
\centering
\includegraphics[width=1.0\textwidth]{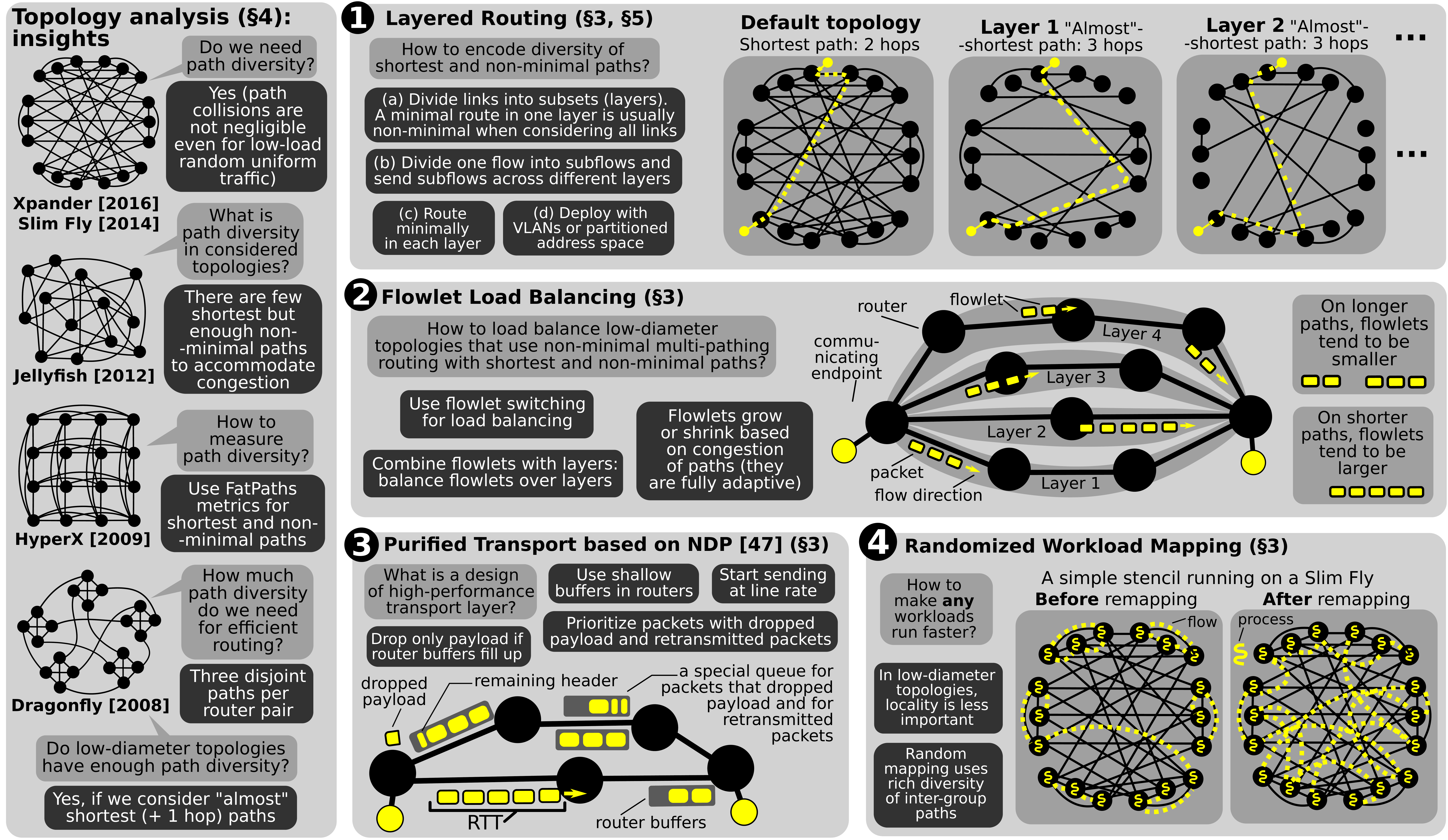}
\ifsq
\vspace{-1.5em}
\else
\vspace{-1em}
\fi
\caption{\textmd{\textbf{Overview of FatPaths}. Numbers (\ding{182} --
\ding{185}) refer to parts of Section~\ref{sec:overview}. We present selected
considered research questions with the obtained answers. ``Topology analysis''
summarizes our work on the path diversity of low-diameter topologies. ``Path
diversity'' intuitively means the number of edge-disjoint paths between router
pairs (details in~\cref{sec:paths}).}}
\vspaceSQ{-1.5em}
\label{fig:functional}
\end{figure*}

\subsection{{Traffic Patterns}}
We analyze recent works~\cite{besta2014slim,
prisacari2013fast, yuan2014lfti, yuan2013new, prisacari2014efficient,
prisacari2014randomizing, kathareios2015cost, prisacari2015performance,
chen2016evaluation, prisacari2013bandwidth, karacali2018assessing,
sehery2017flow, kassing2017beyond}
to select traffic patterns that represent important HPC and datacenter workloads.
Denote a set of endpoint IDs $\{1, ..., N\}$ as $V_e$.  Formally, a traffic
pattern is a mapping from source endpoint IDs $s \in V_e$ to destination
endpoints $t(s) \in V_e$.
First, we select \textbf{random uniform} ($t(s) \in V_e \text{\ u.a.r.,}$) and
\textbf{random permutation} ($t(s)$ $ = \pi_N(s)$, where $\pi_N$ is a
permutation selected $\text{\ u.a.r.,}$) that represent \textbf{irregular
workloads} such as graph computations, sparse linear algebra solvers, and
adaptive mesh refinement methods~\cite{Yuan:2013:NRS:2503210.2503229}.  Second,
we pick \textbf{off-diagonals} ($t(s) = (s+c)\mod N$, for fixed $c$) and
\textbf{shuffle} ($t(s) = \rotl_i(s) \mod N$, where the bitwise left rotation
on $i$ bits is denoted as $\rotl_i$ and $2^i < N < 2^{i+1}$). They represent
\textbf{collective operations} such as MPI-all-to-all or
MPI-all-gather~\cite{besta2014slim, Yuan:2013:NRS:2503210.2503229}.
We also use \textbf{stencils, realistic traffic patterns common in HPC}.  We
model 2D stencils as four off-diagonals at fixed offsets $c \in \{\pm 1,$ $\pm
1, \pm 42,$ $\pm 42\}$.  For large simulations ($N > 10,000$) we also use
offsets $c \in \{\pm 1,$ $\pm 1, \pm 1337,$ $\pm 1337\}$ to reduce counts of
communicating endpoint pairs that sit on the same switches.
Finally, we use \textbf{adversarial} and \textbf{worst-case} traffic patterns.
In the former, we use a skewed off-diagonal with large offsets (we make sure
that it has many colliding paths). For the latter, we use a pattern (detailed
in~\cref{sec:theory}) that maximizes stress on the interconnect
\emph{individually for each topology}.

\ifall

\subsection{Flow Sizes}

We model {flow sizes} with a Poisson distributed {flow arrival rate}. 

\fi

\ifall
We model \emph{traffic} using a \emph{traffic matrix} defined on endpoint pairs along
with a flow size distribution model with Poisson-distributed flow arrival rate.
%
%
\emph{Traffic pattern} is a mapping from source endpoint IDs $s \in
\mathcal{E}$ to destination endpoints $t(s) \in \mathcal{E}$; $\mathcal{E}$ is a set with IDs of all endpoints. We consider
random uniform:
$$
t(s) = T \in \mathcal{E} \text{\
u.a.r.,}
$$

\noindent
random permutations 
$$
t(s) = \pi_N(s),\ \pi_N \text{\ u.a.r.,}
$$

\noindent
off-diagonals as they appear in stencils 

$$
t(s) = (s+c) \mod N,\ \text{for fixed $c$,}
$$

\noindent
and shuffle 

$$
t(s) = \rotl_i(s) \mod N\text{,}
$$

\noindent
with the bitwise left rotation on $i$
bits $\rotl_i$ and $2^i < N < 2^{i+1}$.
Unless stated otherwise, we consider \emph{random assignments}, where the
endpoints $s$ and $t(s)$ are located at router~$\lfloor\pi_{N}(s)/p\rfloor \in
V$ and $\lfloor\pi_{N}(t(s))/p\rfloor \in V$, with $\pi_N$
chosen uniformly \emph{at random}.
\fi

\ifsq\enlargethispage{\baselineskip}\fi

\section{FatPaths Architecture: Overview}
\label{sec:overview}

We show FatPaths' architecture in Figure~\ref{fig:functional}.
%
%
The key part, layered routing, is summarized here and detailed in
Section~\ref{sec:routing}. For higher performance, FatPaths uses simple and
robust flowlet load balancing, ``purified'' high-performance transport, and
randomized workload mapping.
\iftr
Combined, they effectively use the ``fat'' diversity of minimal \emph{and} non-minimal
paths.
\fi
%

\subsection{{Layered Routing} {\large\ding{182}}}
To encode minimal \emph{and non-minimal} paths with commodity
HW, FatPaths divides all links into (not necessarily disjoint)
subsets called \emph{layers}\footnote{\scriptsize In FatPaths, a ``layer'' is
formally a \emph{subset} of links. We use the term ``layer'' as our concept is
similar to ``virtual layers'' known from works on
deadlock-freedom~\cite{skeie2002layered}}. Routing within each layer uses
shortest paths; these paths are usually \emph{not} shortest when considering
all network links.  Different layers encode different paths between each
endpoint pair.
%
%
%
\iftr
This enables taking advantage of the diversity of non-minimal paths in
low-diameter topologies.
\fi
The number of layers is minimized to reduce hardware resources needed to deploy
layers. Layers can easily be implemented with commodity schemes, e.g., VLANs or
a simple partitioning of the address space.


\ifsq\enlargethispage{\baselineskip}\fi

\subsection{{Simple and Effective Load Balancing} {\large\ding{183}}}
%

For simple but robust load balancing, we use flowlet
switching~\cite{sinha2004burstiness, kandula2007dynamic}, a technique
traditionally used to alleviate packet reordering in TCP. A flowlet is a
sequence 
%
%
of packets within one flow, separated from other flowlets by sufficient time
gaps, which prevents packet reordering at the receiver. Flowlet switching can
also provide a \emph{very} simple load balancing: a router simply picks a
random path for each flowlet, without \emph{any} probing for
congestion~\cite{vanini2017letflow}.
\iftr
This scheme
was used for Clos networks~\cite{vanini2017letflow}.
\fi  
Such load balancing is powerful because flowlets are \emph{elastic}: their size
changes \emph{automatically} based on conditions in the network. On
high-latency and low-bandwidth paths, flowlets are usually smaller in size
because time gaps large enough to separate two flowlets are more frequent.
Contrarily, low-latency and high-bandwidth paths host longer flowlets as such
time gaps appear less often.
Now, we propose to use flowlets in \emph{low-diameter} networks, to load
balance FatPaths. We combine flowlets with layered routing:
flowlets are sent using \emph{different layers}.
The key observation is that elasticity of flowlets \emph{automatically} ensures
that such load balancing considers both static 
and dynamic network properties (e.g., longer vs.~shorter paths as well as more
vs.~less congestion).  
Consider a pair of communicating routers. As we show
in~\cref{sec:paths}, virtually all router pairs in a low-diameter network are
connected with exactly one shortest part but multiple non-minimal paths,
possibly of different lengths. A shortest path often
experiences smallest congestion while longer paths are more likely to
be congested. Here, flowlet elasticity ensures that
larger flowlets are sent over shorter and less congested paths. Shorter
flowlets are then transmitted over longer and usually more congested paths.

\begin{figure*}[t]
\vspaceSQ{-1.25em}
\centering
\begin{minipage}{0.61\textwidth}
\centering
\includegraphics[width=1.0\textwidth]{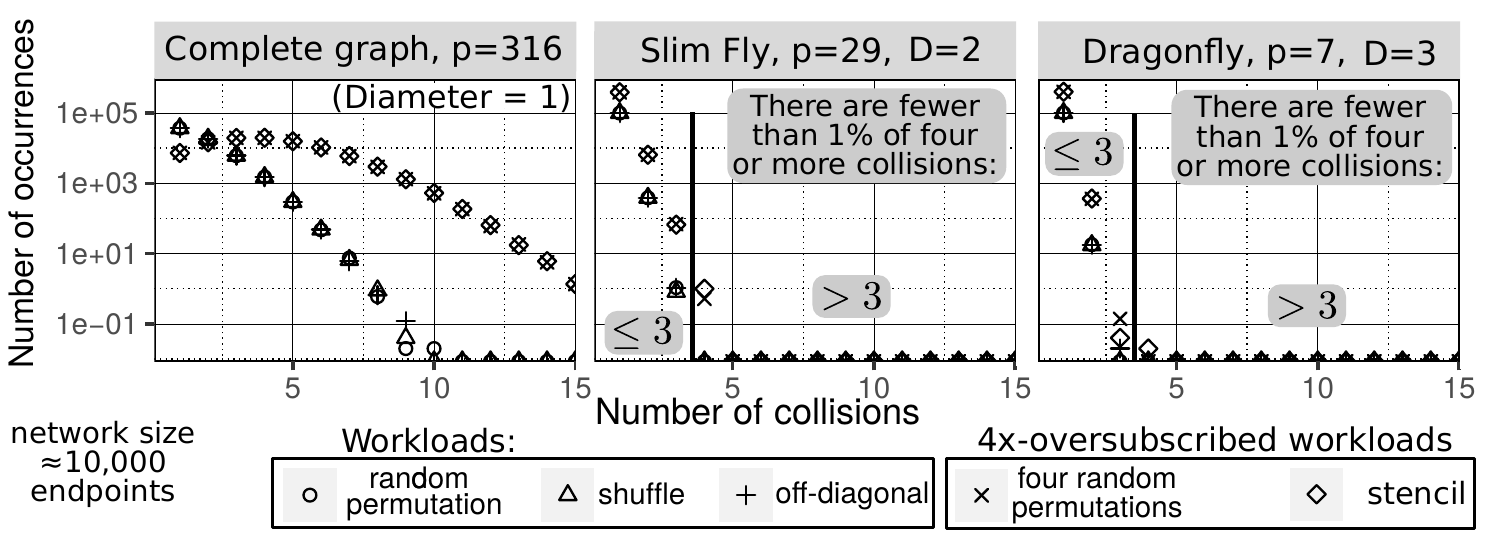}
\vspaceSQ{-1.5em}
\caption{\textmd{Histogram of colliding paths (per router pair) with
$p=\frac{k'}D$.  We plot data for SF, DF, and a complete graph, for $N \approx 100,000$. 
Other sizes $N$
and other topologies result in identical shares of colliding paths.}} 
\label{fig:path_collisions}
\end{minipage}\hfill
\begin{minipage}{0.35\textwidth}
\centering
\includegraphics[width=1.0\textwidth]{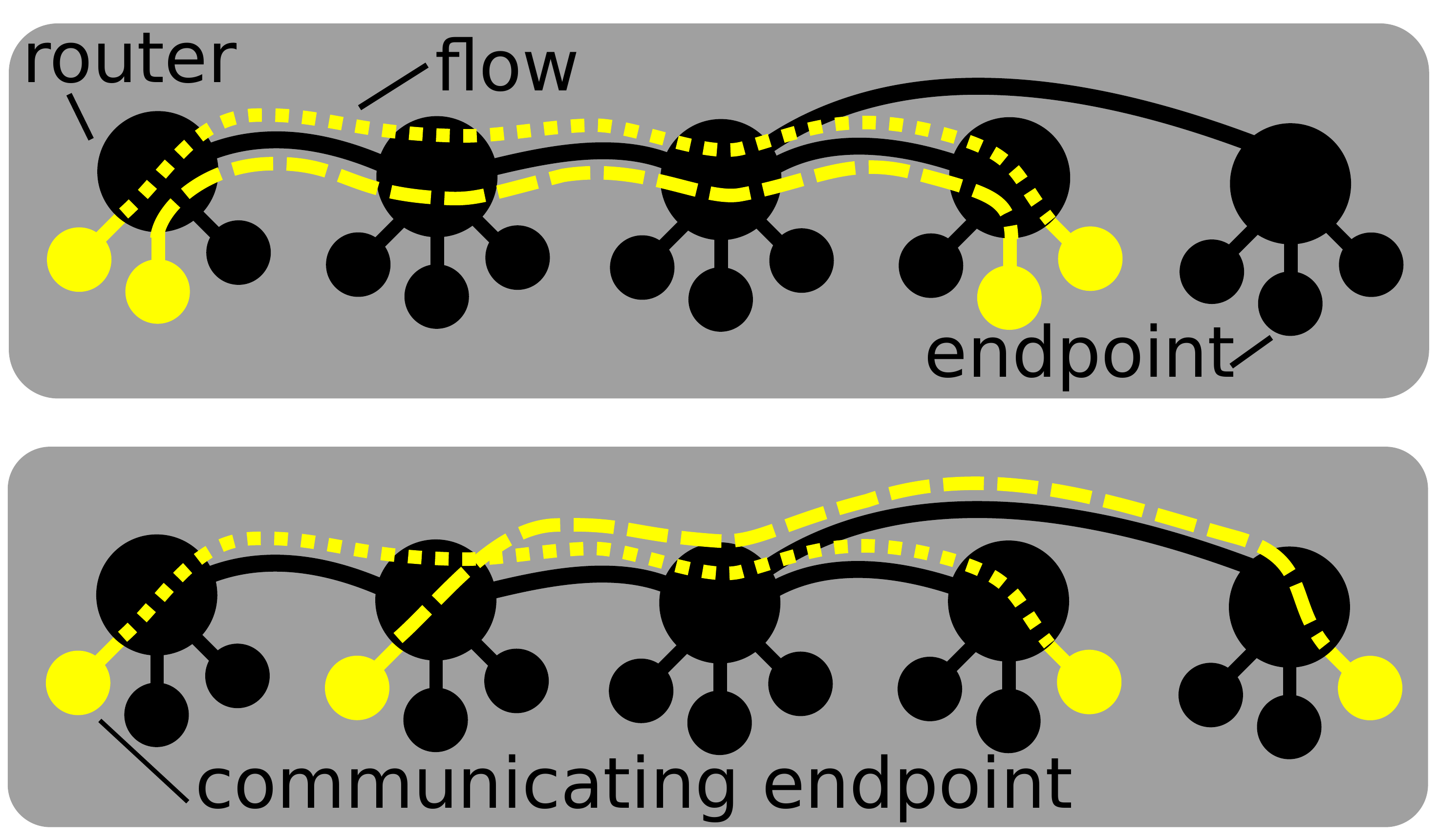}
\vspaceSQ{-0.5em}
\caption{\textmd{Example collision (top) and overlap (bottom) of paths and flows.}} 
\label{fig:paths-problems}
\end{minipage}
\vspaceSQ{-0.75em}
\end{figure*}

\subsection{{Purified Transport} {\large\ding{184}}}
The FatPaths transport layer is inspired by the NDP~\cite{handley2017re} Clos design,
in that it removes virtually all TCP/Ethernet performance issues. First, if router queues
fill up, \emph{only packet payload is dropped}; headers with all the
metadata are preserved and the receiver has full information on the congestion in
the network and can pull the data from the sender at a rate dictated by the
evolving network conditions. Specifically, the receiver can request to change
a layer, if packets transmitted over this layer 
arrive without payload, indicating congestion.
Second, routers can prioritize (1) headers of packets that lost 
payload, and (2) retransmitted packets. Thus, congested flows
finish quickly and head-of-line-blocking is reduced. Third, senders transmit the
first RTT at line rate (no probing for available bandwidth). Finally,
router queues are shallow. The resulting transport layer has low latency and
high throughput, it meets demands of various traffic patterns,
and it can be implemented with existing network technology.


\subsection{{Randomized Workload Mapping} {\large\ding{185}}}
We {optionally} assign communicating endpoints to routers \emph{randomly}.
\iftr
This is often done in HPC.  We stress that this scheme is \emph{even more
beneficial in FatPaths} due to the {low diameter} of targeted networks.
Specifically, we place communicating endpoints at routers chosen u.a.r., to
distribute load more evenly over the network.
The motivation is as follows.  First, in future large-scale data centers and
supercomputers, locality cannot be effectively utilized: Paths of 0--1 hops
asymptotically disappear as $N$ grows (in the random uniform traffic).
Moreover, one often cannot rely on locality due to schedulers or
virtualization.
For example, Cray machines often host processes from one job in different
machine parts to increase utilization~\cite{bhatele2013there}.
Furthermore, many workloads, such as distributed graph processing, have little
or no locality~\cite{lumsdaine2007challenges}.
%
%
Finally, the low diameter of used topologies, especially $D=2$, mostly
eliminates the need for locality-aware software.
We predict that this will be a future trend as \emph{reducing cost and power
consumption with simultaneous increase in scale {is inherently associated with
reducing diameter}}~\cite{besta2014slim}.
Still, to cover applications tuned for locality, we show that {FatPaths also
ensures speedups of non-randomized workloads}.
\fi

\section{Path Diversity in Low-Diameter Topologies}
\label{sec:paths}

\iftr
FatPaths \emph{enables} using the diversity of paths in low-diameter topologies
for high-performance routing. 
\fi
To develop FatPaths, we first need to understand the ``nature'' of path
diversity that FatPaths benefits from. 
For this, we first show that {low-diameter topologies exhibit congestion due to
conflicting flows even in mild traffic scenarios}, and we derive the minimum
number of disjoint paths that eliminate flow conflicts (\cref{sec:collisions}).
We then formalize the ``path diversity'' notion
(\cref{sec:measuring-diversity}) to show that {all low-diameter topologies have
few shortest but enough non-minimal paths to accommodate flow collisions, an
important type of flow conflicts} (\cref{sec:collisions-details}). 
\ifall
In evaluation (\cref{sec:eval}), we show that another type of flow conflicts,
flow overlaps, is also alleviated by FatPaths.
\fi
To the best of our knowledge, {we provide the most extensive analysis of path
diversity in low-diameter networks} (considering the number of path
diversity metrics and topologies), cf.~Related Work.
%
%

\iftr
The key motivation for this analysis is that modern low-diameter topologies
have good \emph{expansion
properties}~\cite{hoory2006expander, valadarsky2015}. Intuitively, such a (sparse) network should
\emph{not} offer much path diversity. We illustrate that, {contrarily to this
intuition}, modern low-diameter topologies {do indeed} offer path diversity
that is enough to ensure high performance of specified traffic patterns without
congestion, when considering non-minimal paths. Finally, our goal is {to also
provide deeper understanding of how to analyze path diversity}.
\fi

\subsection{How Much Path Diversity Do We Need?}
\label{sec:collisions}


%

\iftr
FatPaths uses path diversity {to avoid congestion} due to \emph{conflicting
flows}, i.e., flows that share at least one inter-router link.  Such conflicts
depend on both how communicating endpoints are mapped to routers and the
topology details.
Consider two communicating pairs of endpoints. The generated flows
\emph{conflict} when their paths \textbf{collide} (flows use an identical path)
{or} \textbf{overlap} (flows share some links), see
Figure~\ref{fig:paths-problems}. 
Collisions are caused by \emph{workload mapping}, when communicating endpoint
pairs occupy the same router pairs (i.e., sources and destinations are attached
to the same router pair). Thus, collisions \emph{only} depend on concentration
$p$ and the number of routers $N_r$.
Contrarily, overlaps depend on \emph{topology details} (i.e., connections
\emph{between routers}). Thus, overlaps capture {how well a topology can
sustain a given workload}. 
\fi

\ifsq
FatPaths uses path diversity {to avoid congestion} due to \emph{conflicting
flows}. Consider two communicating pairs of endpoints. The generated flows
\emph{conflict} when their paths \textbf{collide} (flows use an identical path)
{or} \textbf{overlap} (flows share some links), see
Figure~\ref{fig:paths-problems}. 
Collisions are caused by \emph{workload mapping}, when communicating endpoint
pairs occupy the same router pairs. Thus, collisions \emph{only} depend on
concentration $p$ and \#routers $N_r$.
Contrarily, overlaps depend on \emph{topology details} (i.e., connections
\emph{between routers}). Thus, overlaps capture {how well a topology can
sustain a given workload}. 
\fi

\ifsq\enlargethispage{\baselineskip}\fi

To understand how much path diversity is needed to alleviate flow conflicts, we
analyze the impact of topology properties ($D$, $p$, $N$) and a traffic pattern
on the number of colliding paths, see Figure~\ref{fig:path_collisions}. 
\iftr
We consider five traffic
patterns: a random permutation, a randomly-mapped off-diagonal, a randomly
mapped shuffle, four random permutations in parallel, and a randomly mapped
4-point stencil composed of four off-diagonals. The last two patterns are
4$\times$ oversubscribed and thus \emph{expected to generate even more
collisions}.
\fi
For $D > 1$, the number of collisions is at most \emph{three} in most 
cases, especially when lowering $D$ (while increasing $p$). Importantly,
this holds for the adversarial 4$\times$ oversubscribed patterns that
stress the interconnect. 
%
%
For $D=1$, at least nine collisions occur for more than 1\% of router pairs,
even in mild traffic patterns. 
\ifall
Here, one could significantly reduce $p$ in order to eliminate congestion. 
\fi
While we do not consider $D=1$ in practical applications, we indicate that
global DF links form a complete graph, demanding high path diversity at least
with respect to the global links.

\noindent\macb{\ul{Takeaway} }
We need {at least three disjoint paths per router pair to handle colliding
paths} in considered workloads, assuming random mapping.
Now, we observe that there are at least as many overlapping paths as colliding
paths (as seen from a simple counting argument: for each pair of colliding
flows~$x$ and $y$, any other flow in the network may potentially overlap with
$x$ and $y$). Thus, {the same holds for overlaps}.

\subsection{How Should We Measure Path Diversity?}
\label{sec:measuring-diversity}

To analyze if low-diameter topologies provide at least three disjoint paths per
router pair, we need to first {formalize} the notion of ``disjoint paths'' and
''path diversity'' in general.  For example, we must be able to distinguish
between {partially or fully disjoint} paths that may have {different lengths}.
Thus, we first define \emph{the count of disjoint paths} (CDP), minimal and
non-minimal, between routers (\cref{sec:ec}).  These measures address path
{collisions}.
Moreover, to analyze path {overlaps}, we define two further measures:
\emph{path interference} ({PI}, \cref{sec:pi}) and \emph{total network load}
({TNL}, \cref{sec:tnl}). 
\iftr
Intuitively, they can be seen as tools for estimating path overlap using local
and global topology properties, respectively.
\fi
We summarize each measure and we provide all formal details for
reproducibility; these details can be omitted by readers only interested in
intuition.
We use several measures because any single measure that we tested {cannot fully
capture the rich concept of path diversity}.

%


\iftr
\begin{figure*}[b]
\vspaceSQ{-1em}
\centering
\includegraphics[width=0.7\textwidth]{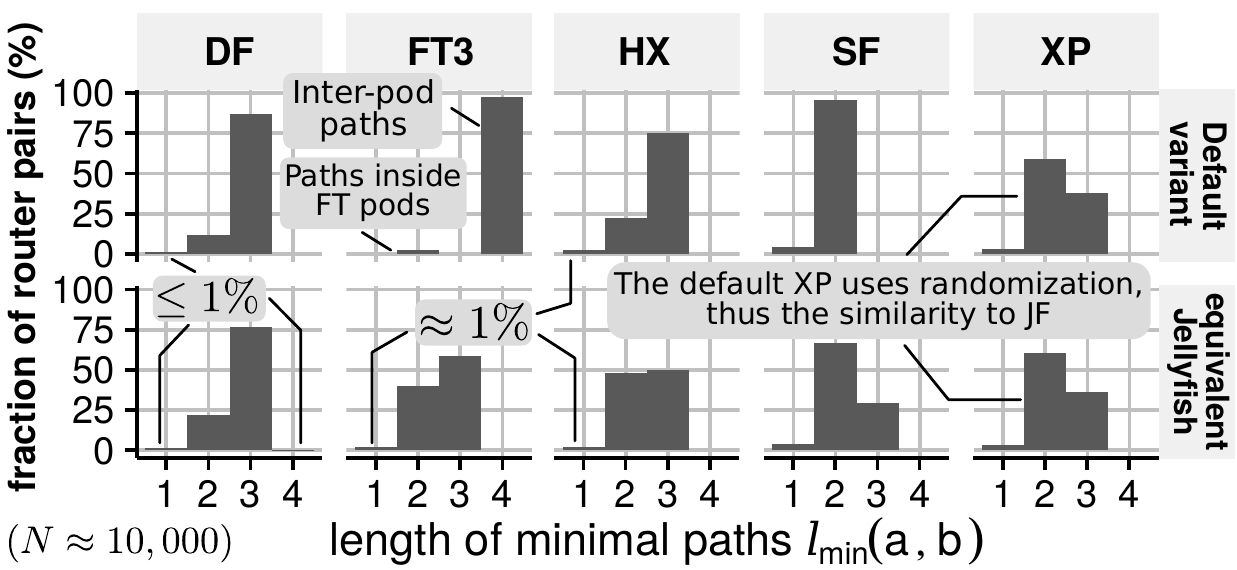}\\%
\includegraphics[width=0.7\textwidth]{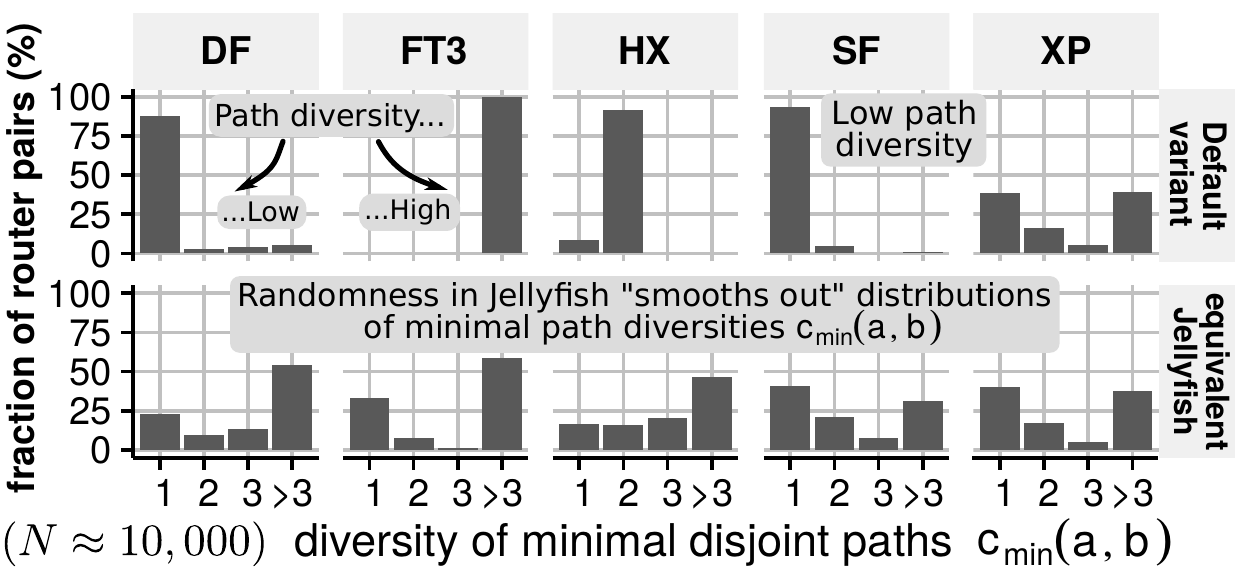}%
\vspaceSQ{-0.5em}
\caption{\textmd{Distributions of
lengths and counts of shortest paths.}}
\vspaceSQ{-1em}
\label{fig:shortest_path_length}
\label{fig:shortest_path_multiplicity}
\end{figure*}
\fi

\subsubsection{\textbf{Count of Disjoint Paths (CDP)}}
\label{sec:ec}

\ifsq\enlargethispage{\baselineskip}\fi
 
We define the count of \emph{disjoint} paths (CDP) between router sets $A, B
\subseteq V$ at length $l$ as the {smallest} number $c_l(A,B)$ of edges
that must be removed so that {no path of length at most $l$ exists from
any router in $A$ to any router in $B$}.
%
%
To compute $c_l(A,B)$, first define the $l$-step neighborhood $h^l(A)$ of a
router set $A$ as ``a set of routers at $l$ hops away from $A$'':

\vspaceSQ{-1em}
\small
\begin{align}
h(A) &= \{t \in V: \exists_{s \in A}\ \{s,t\} \in E\}\quad \text{(``routers attached to $A$'')} \nonumber\\
h^l(A) &= \underbrace{h(\cdots h(}_{l \text{\ times}} A)\cdots)\quad \text{(``$l$-step neighborhood of $A$'')}.\nonumber
\end{align}
\normalsize
\vspaceSQ{-1em}

%
Now, the condition that no path of length at most $l$ exists between any router
in $A$ to any router in $B$ is $h^l(A) \cap B = \emptyset$. To derive the
values of $c_l(A,B)$, we use a variant of the Ford-Fulkerson
algorithm~\cite{ford1956maximal} (with various pruning heuristics) that removes
edges in paths between designated routers in $A$ and $B$ (at various distances
$l$) and verifies whether $h^l(A) \cap B = \emptyset$. We are most often
interested in pairs of designated routers $s$ and $t$, and we use $A = \{s\}, B
= \{t\}$.




\textbf{Minimal paths} are vital in congestion reduction as they use fewest
resources for each flow. We derive the {distribution} of minimal path
\emph{lengths} $l_\text{min}$ and \emph{counts} $c_\text{min}$.  Intuitively,
$l_\text{min}$ describes (statistically) \emph{distances} between any router
pairs while $c_\text{min}$ provides their respective \emph{diversities}.
%
%
%
We have:

\vspaceSQ{-1em}
\small
\begin{align}
l_\text{min}(s,t) &= \argmin_{i\in \mathbb{N}}\{t \in h^i(\{s\})\}\quad \text{(``minimal path lengths'')}\nonumber\\
c_\text{min}(s,t) &= c_{l}(\{s\},\{t\}) \text{ with } l=l_\text{min}(s,t)\quad \text{(``minimal path counts'')}\nonumber
\end{align}
\normalsize

\noindent
Note that the diameter~$D$ equals $\max_{s,t}\ l_\text{min}(s,t)$.




To analyze \textbf{non-minimal paths}, we reuse the count of disjoint paths CDP
$c_{l}(A,B)$ of random router pairs $s \in A, t \in B$, but with path lengths
$l$ {larger than} $l_\text{min}(s,t)$ ($l > l_\text{min}(s,t)$).  Here, we are
interested in distributions of {counts} of non-minimal paths for fixed
non-minimal distances~$l$.

\subsubsection{\textbf{Path Interference (PI)}}
\label{sec:pi}

With Path Interference (PI), we want to quantify path overlaps.  This is
challenging because overlaps depend on the details of the structure of each
topology as well as on workload mappings.  Thus, a PI definition must be
\emph{local} in that it should consider \emph{all} router pairs that may
possibly communicate.  Consider two router pairs $a,b$ and $c,d$ where $a$
communicates with $b$ and $c$ communicates with $d$.  Now, paths between these
two pairs \emph{interfere} if their total count of disjoint paths (at a given
distance~$l$), $c_{l}(\{a,c\},\{b,d\})$, is lower than the sum of individual
counts of disjoint paths (at $l$): $c_{l}(\{a\},\{b\}) + c_{l}(\{c\},\{d\})$. 
We denote path interference with $I^l_{ac,bd}$ and define it as

\ifsq
\vspace{-0.25em}
\vspace{-0.75em}
\fi
$$
I^l_{ac,bd} = c_{l}(\{a, c\},\{b\}) + c_{l}(\{a, c\},\{d\}) - c_{l}(\{a,c\},\{b,d\})
$$
\ifsq
\vspace{-0.75em}
\vspace{-0.25em}
\fi

Path interference captures and quantifies the fact that, if $a$ and $b$
communicate \emph{and} $c$ and $d$ communicate \emph{and} the flows between
these two pairs use paths that are {not} fully disjoint (due to, e.g., not
ideal routing), then the available bandwidth between any of these two pairs of
routers is reduced.

\ifall\m{fix}
If $I^l_{ac,bd}$ is small,
the bandwidth loss can be made little with proper routing
(because there is little inherent path overlap between $\{a,c\}$
and $\{b,d\}$). Contrarily, for large $I^l_{ac,bd}$

This is because
some paths $c_l(\{a\}, \{b\}) > 0$ and $c_l(\{c\}, \{d\}) > 0$

a value proportional to $I^l_{ac,bd}$.

the available bandwidth between $c$ and $d$ is reduced, because the
number of disjoint paths (at~$l$) $c_l(\{c\}, \{d\})$ is smaller
precisely by the
%
%
\fi

\vspaceSQ{-0.45em}
\subsubsection{\textbf{Total Network Load (TNL)}}
\label{sec:tnl}


TNL is a simple {upper bound on the number of flows that a network can maintain
without congestion} (i.e., without flow conflicts). 
There are $k' N_r$ links in a topology.  Now, a flow occupying a path of length
$l$ ``consumes'' $l$ links.  Thus, with the average path length of~$d$, TNL is
defined as
%
%
$\frac{k' N_r}{d}$, because $\#flows \le \frac{k' N_r}{d}$.  
Thus, TNL constitutes the maximum \emph{supply of path diversity} offered by a
specific topology.

\noindent
\macb{\ul{Takeaway} }
We suggest to use several measures to analyze the rich nature of path
diversity, e.g., the count of minimal and non-minimal paths (for collisions),
and path interference as well as the total network load (for overlaps).

\begin{table*}[h]
\vspaceSQ{-1.25em}
\begin{minipage}[b]{0.65\textwidth}
\centering
\includegraphics[width=1.0\textwidth]{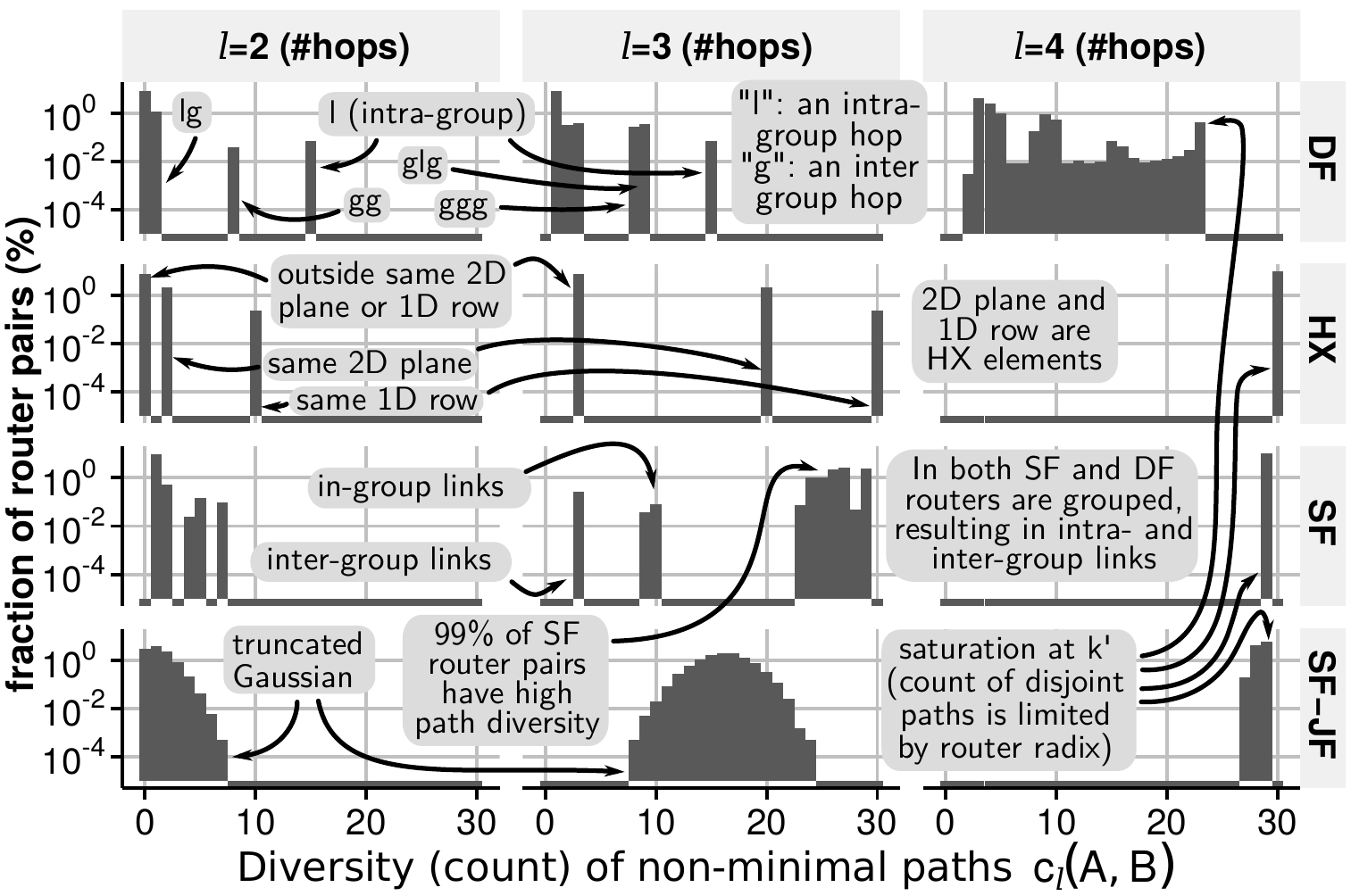}
\vspaceSQ{-1.5em}
\captionof{figure}{\textmd{Distribution of the number of non-minimal edge-disjoint paths
with up to $l$ hops ($c_l(A,B)$) between random router pairs ($N \approx
10,000$).}}
\vspaceSQ{-0.5em}
\label{fig:edge_disjoint_paths}
\end{minipage}
\hfill
\begin{varwidth}[b]{0.32\textwidth}
\centering
\centering
\setlength{\tabcolsep}{1pt}
\renewcommand{\arraystretch}{1.1}
\ssmall
\sffamily
\begin{tabular}{lcrrrr}
\toprule
\multicolumn{6}{c}{\textbf{Topology parameters}} \\ 
\midrule
& $d'$ & $D$ &$k'$ & $N_r$  & $N$ \\
\midrule
clique   & 2   & 1 & 100 & 101  & 10100 \\
SF       & 3   & 2 & 29  & 722  & 10108 \\
XP       & 3   & 3 & 32  & 1056 & 16896 \\
HX      & 3   & 3 & 30  & 1331 & 13310 \\
DF       & 4   & 3 & 23  & 2064 & 16512 \\
FT3      & 4   & 4 & 18  & 1620 & 11664 \\
\midrule
\multicolumn{6}{c}{\textbf{Default topology variant}} \\
& $d'$ & \textbf{CDP} (mean) & \textbf{CDP} (1\% tail) & \textbf{PI} (mean) & \textbf{PI} (99.9\% tail) \\ 
\midrule
clique     & 2 & 100\% & 100\% & 2\% & 2\% \\
SF         & 3 & 89\%  & 10\%       & 26\%  & 79\%  \\
XP         & 3 & 49\%  & 34\%    & 20\%  & 41\%    \\
HX        & 3  & 25\%  & 10\%       & 9\%   & 67\% \\ 
DF        & 4 & 25\%  & 13\%       & 8\%   & 74\%  \\
FT3       & 4 & 100\% & 100\%    & 0     & 0     \\ 
\midrule
\multicolumn{6}{c}{\textbf{Equivalent Jellyfish}} \\
& $d'$ & \textbf{CDP} (mean) & \textbf{CDP} (1\% tail) & \textbf{PI} (mean) & \textbf{PI} (99.9\% tail) \\
\midrule
SF-JF         & 3 & 56\%  & 38\%   & 23\%  & 45\% \\
XP-JF         & 3   & 51\%  & 34\%   & 21\%  & 41\% \\
HX-JF        & 3   & 50\%  & 23\%   & 17\%  & 37\% \\
DF-JF         & 4   & 87\%  & 78\%   & 13\%  & 26\% \\
FT3-JF        & 4   & 96\%  & 90\%   & 5\%   & 14\% \\
\bottomrule
\end{tabular}
\vspace{-0.5em}
%
%
\caption{\textmd{Counts of disjoint non-minimal paths 
CDP $\left(c_{d'}(A,B)\right)$ and path interference PI $\left(I^{d'}_{ac,bd}\right)$ at distance
$d'$; $d'$ is chosen such that the tail $c_{d'}(A,B)
\ge 3$. 
}}
\vspaceSQ{-3em}
\label{tab:measures}
\end{varwidth}
\end{table*}

\subsection{Do We Have Enough Path Diversity?}
\label{sec:collisions-details}


We now use our measures for path diversity analysis.

\ifsq\enlargethispage{\baselineskip}\fi

\subsubsection{{Analysis of Minimal Paths}}

Selected results for {minimal paths} are in
Figure~\ref{fig:shortest_path_length}. In DF and SF, most routers are connected
with one minimal path. In XP, more than 30\% of routers are connected with
\emph{one} minimal path. 
Only in HX, most router pairs have two minimal paths. 
In JF, the results are more leveled out, but pairs of routers with one shortest
part in-between still form large fractions. 
FT3 and HX show the highest diversity, with very few unique minimal paths,
while the matching JFs have lower diversities. 
The results match the structure of each topology (e.g., one can distinguish
intra- and inter-pod paths in FT3).
%
%
%
%
%
%
%

\noindent\macb{\ul{Takeaway} }
In all the considered low-diameter topologies, {shortest paths
fall short:} at least a large fraction of router pairs are connected with
\emph{only one} shortest path.

%

\subsubsection{{Analysis of Non-Minimal Paths}}

\iftr
Based on the collision multiplicity found in~\cref{sec:collisions}, we aim to
obtain at least three disjoint paths per router pair. 
\fi
For {non-minimal paths}, we first summarize the results in
Table~\ref{tab:measures}.  We report counts of disjoint paths as
\emph{fractions of router radix $k'$} to make these counts radix-invariant.
For example, the mean CDP of 89\% in SF means that 89\% of router links host
disjoint paths.
In general, all deterministic topologies provide higher disjoint path diversity
than their corresponding JFs, but there are specific router pairs with lower
diversity that lead to undesired tail behavior. JFs have more predictable tail
behavior due to the Gaussian distribution of $c_{l}(A,B)$.
A closer analysis of this distribution (Figure~\ref{fig:edge_disjoint_paths})
reveals details about each topology. For example, 
for HX, router pairs can clearly be separated into classes sharing zero, one,
or two coordinate values, corresponding to the HX array structure. Another
example is 
SF, where lower $c_{l}(A,B)$ are related to pairs connected with an edge while
higher $c_{l}(A,B)$ in DF are related to pairs in the same group or pairs
connected with specific sequences of local and global links.

\noindent\macb{\ul{Takeaway} }
\iftr
\emph{\textbf{Considered topologies provide three disjoint
``almost''-minimal (one hop longer) paths per router pair}}.
\fi

\iftr
Note that $d$ depends on used routing and may be larger or smaller than $D$:
with minimal routing, $d \leq D$. With non-minimal routing, it can be larger:
Valiant's algorithm doubles $d$ compared to minimal routing.  This cannot
always be avoided (e.g., it affects cliques). While $D=1$ promises a
high possible $p$, this configuration leads to a high number of path collisions
which need to be handled using length~2 paths. Those increase $d$, which in
turn requires $p$ to be reduced to avoid congestion, according to the TNL
argument. Yet, for randomized workloads and $D \geq 2$, as we showed in
\cref{sec:collisions}, the total amount of path collisions is low, which means
non-minimal routing is feasible without decreasing $p$.
\fi



\ifsq\enlargethispage{\baselineskip}\fi

\subsubsection{{Analysis of Path Interference}}

Next, we sample router pairs u.a.r.~and derive full \textbf{path interference}
distributions; they all follow the Gaussian distribution.
\iftr
Selected results are in Figure~\ref{fig:interference-dist}
(we omit XP and XP-JF; both are nearly identical to SF-JF)
\fi
As the combination space is very large, most samples fall into a common case, where
PI is small. We thus provide full data in the report and focus on the extreme tail of the
distribution (we show both mean and tail), see Table~\ref{tab:measures}.
We use
radix-invariant PI values (as for CDP) at a distance $d'$ selected
to ensure that the 99.9\% tail of collisions $c_{d'}(A,B)$ is at least $3$.
Thus, we analyze PI in cases where demand from a workload outgrows 
the ``supply of path diversity'' from a network (three disjoint paths per router pair).
%
%
%
%
All topologies except for DF achieve negligible PI for $d'=4$, but
the diameter-2 topologies do experience PI at $d'=3$. SF 
shows the lowest PI in general, but has (few) high-interference outliers.
\iftr
{In general, random }JFs{ have higher average PI but less PI in tails.
Deterministic topologies perform better on average, but with worse tails;
randomization helps to improve tail behavior.
\fi

\ifall
\maciej{?}
Full results on the distribution of the number of non-minimal paths
between random router pairs, and on the distribution of path
interference, are presented in Figures~\ref{fig:app-edge_disjoint_paths_all}
and~\ref{fig:app-path-interference-full}.
\fi

\iftr

\begin{figure*}[b]
\vspaceSQ{-0.25em}
\centering%
\includegraphics[width=0.87\textwidth]{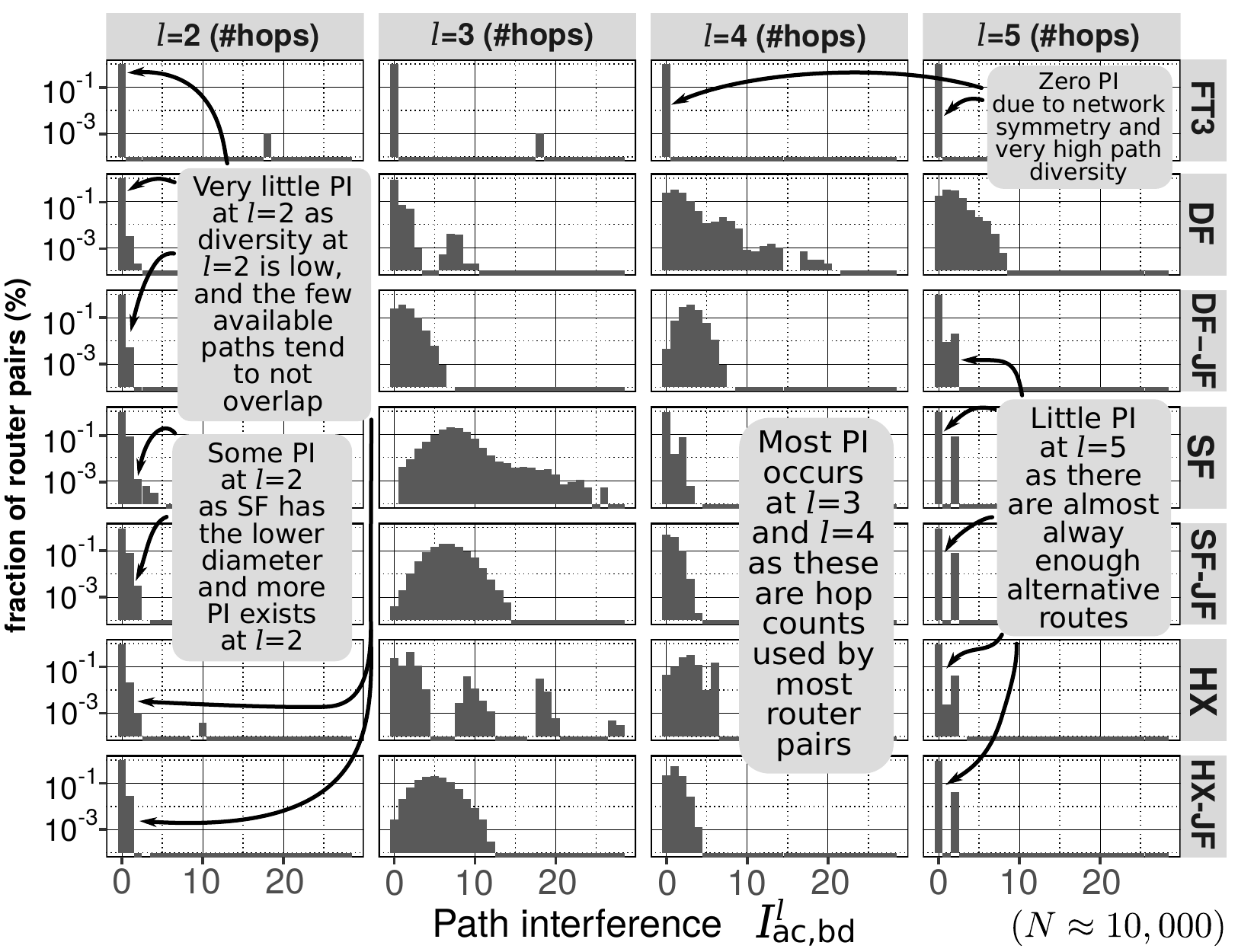}
\vspaceSQ{-1.5em}
\caption{\textmd{Distribution of path interference PI $\left(I^l_{ac,bd}\right)$ at 
various distances $l$.}}
\vspaceSQ{-1em}
\label{fig:interference-dist}
\end{figure*}

\fi

\ifall

\begin{table}
\centering
\setlength{\tabcolsep}{1pt}
\ssmall
\sffamily
\begin{tabular}{lccrrr|rrrr|rrrr}
\toprule
\multicolumn{6}{c|}{\textbf{Topology parameters}} & \multicolumn{4}{c|}{\textbf{Default topology variant}} & \multicolumn{4}{c}{\textbf{Equivalent Jellyfish}} \\
 & & & & & & CDP & CDP & PI & \multicolumn{1}{l|}{\makecell[c]{PI}} & CDP & CDP & PI & \multicolumn{1}{l}{\makecell[c]{PI}}\\
\midrule
& $D$ & $d'$ &$k'$ & $N_r$  & $N$ & mean & \makecell[c]{1\% tail} &  mean & \makecell[c]{99.9\% tail} &  mean & \makecell[c]{1\% tail} &  mean & \makecell[c]{99.9\% tail} \\
\midrule
clique   & 1   & 2 & 100 & 101  & 10100 & 100\% & 100\% & 2\% & 2\% & -- & -- & -- & -- \\
SF       & 2   & 3 & 29  & 722  & 10108 & 89\%  & 10\%       & 26\%  & 79\%     & 56\%  & 38\%   & 23\%  & 45\% \\
XP       & 3   & 3 & 32  & 1056 & 16896 & 49\%  & 34\%    & 20\%  & 41\%     & 51\%  & 34\%   & 21\%  & 41\% \\
HX      & 3   & 3 & 30  & 1331 & 13310 & 25\%  & 10\%       & 9\%   & 67\%     & 50\%  & 23\%   & 17\%  & 37\% \\
DF       & 3   & 4 & 23  & 2064 & 16512 & 25\%  & 13\%       & 8\%   & 74\%     & 87\%  & 78\%   & 13\%  & 26\% \\
FT3      & 4   & 4 & 18  & 1620 & 11664 & 100\% & 100\%    & 0     & 0        & 96\%  & 90\%   & 5\%   & 14\% \\
\bottomrule
\end{tabular}
\vspace{-0.5em}
%
%
\caption{\textmd{Counts of disjoint non-minimal paths 
CDP $\left(c_{d'}(A,B)\right)$ and path interference PI $\left(I^{d'}_{ac,bd}\right)$ at distance
$d'$; $d'$ is chosen such that the tail $c_{d'}(A,B)
\ge 3$. 
}}
\vspace{-3em}
\label{tab:measures}
\end{table}

\fi

\ifall
\begin{figure*}
\centering%
\includegraphics[width=\textwidth]{edge_disjoint_paths_all___mac.pdf}
\vspaceSQ{-2em}
\caption{\textmd{\textbf{Numbers of disjoint paths of various lengths} ($N \approx 10,000$).
Equivalent Jellyfish networks are denoted with ``-JF''.}}
\label{fig:app-edge_disjoint_paths_all}
\end{figure*}
\fi

\ifall
\begin{figure*}
\centering
\includegraphics[width=1\textwidth]{sf_connectivity_problems___mac.eps}
\vspace{-1.5em}
\caption{\maciej{?}\textbf{A map of the low-connectivity node pairs} in a very small SF instance
($N_r=98$). We show pairs with an edge-disjoint connectivity $c_{A,B} < 0.75
k'$ at a length of 3 (black symbols), as an adjacency matrix. For comparison, edges are also
marked (gray symbols). There are two classes of bad connectivity. Both are related to edges
between these pairs, but not all edges are affected by the problems.
}
\label{fig:app-sf_connectivity_problems}
\end{figure*}
\fi

\ifall
\begin{figure*}
\centering%
\includegraphics[width=\textwidth]{interference_topo_all___mac.pdf}
\vspace{-2em}
\caption{\textmd{\maciej{?}\textbf{Path interference at various distances} ($N \approx 10,000$).
All topologies except DF achieve negligible interference at $l = 4$, but the
diameter-2 topologies do experience interference problems at $l = 3$, where
they should already achieve high connectivity. Again, Slim Fly shows low
interference values in general, but some very high outliers. Note that these
results are not normalized by network radix; the interference value is limited
by the radix $k'$. The value does not saturate to 0 if the destinations are
connected by an edge, since this edge cannot contribute in the combined case.
Corresponding Jellyfish networks are denoted with ``-JF''.}
}
\label{fig:app-path-interference-full}
\end{figure*}
\fi

\vspaceSQ{-0.25em}
\subsection{Final Takeaways on Path Diversity}
\label{sec:paths-discussion}

%
We show a fundamental tradeoff between path length and diversity. High-diameter
topologies (e.g., FT) have high path diversity, even on minimal paths.
Yet, due to longer paths, they need more links for an equivalent $N$ and
performance. Low-diameter topologies {fall short of shortest paths}, but
{do provide enough diversity of non-minimal paths, requiring non-minimal
routing}. Yet, this may reduce the cost advantage of low-diameter networks
\emph{with adversarial workloads} that use many non-minimal paths,
consuming additional links. \emph{{Workload randomization} in FatPaths
suffices to avoid this effect}.
{\textbf{Overall, low-diameter topologies host enough path
diversity for alleviating flow conflicts. We now show how to
effectively use this diversity in FatPaths.}}

\ifsq\enlargethispage{\baselineskip}\fi

\section{FatPaths: Design and Implementation}
\label{sec:routing}

{FatPaths is a high-performance, simple, and robust \emph{routing
architecture} that uses rich path diversity in low-diameter topologies}
{to enhance Ethernet stacks in data centers and supercomputers}. FatPaths
aims to accelerate both datacenter and HPC workloads.
We outlined FatPaths in~\cref{sec:overview}. Here, we detail the layered
routing scheme that is capable of encoding the rich diversity of both minimal
and non-minimal paths, and can be implemented with commodity Ethernet hardware. 
%
%

%

\iftr
\subsection{Routing Model}

We assume simple destination-based routing, compatible with any
relevant technology, including source-based systems like NDP. To compute the
output port $j \in \{1,\dots,k'\}$ in a router~$s \in V$ for a packet addressed
to a router~$t \in V$, and simultaneously the ID of the next-hop router $s' \in
V$, a routing function $(j,s') = \sigma(s,t)$ is evaluated.  By iteratively
applying $\sigma$ with fixed $t$ we eventually reach $s' = t$ and finish. We
call the $l+1$ item sequence  $(s, s', \dots, t)$ a \emph{path} of length $l$.
The forwarding function $\sigma$ must be defined such that a path from any $s$
to any $t$ is \emph{loop-free}. 
For \emph{layered} routing in FatPaths, we use routing functions
$\sigma_1, \dots, \sigma_n$, where each router uses function $\sigma_i$ for a
packet with a \emph{layer tag} $i$ attached. The layer tags are chosen on the
endpoint by the adaptivity algorithm. 
\fi

\ifsq
\vspaceSQ{-0.25em}
\subsection{Routing Model}

We assume simple destination-based routing, compatible with any
relevant technology, including source-based systems like NDP. To compute the
output port $j \in \{1,\dots,k'\}$ in a router~$s \in V$ for a packet addressed
to a router~$t \in V$, and simultaneously the ID of the next-hop router $s' \in
V$, a routing function $(j,s') = \sigma(s,t)$ is evaluated.  By iteratively
applying $\sigma$ with fixed $t$ we eventually reach $s' = t$ and finish. 
%
%
The function $\sigma$ must ensure no loops on any path. 
\fi

\vspaceSQ{-0.25em}
\subsection{Layered Routing in FatPaths}
\label{sec:routing-layers}

We use $n$ routing functions
$\sigma_1, \dots, \sigma_n$ for $n$ layers. Each router uses $\sigma_i$ for a
packet with a \emph{layer tag} $i$ attached. The layer tags are chosen on the
endpoint by the adaptivity algorithm. 
\iftr
We use $n$ layers associated with $n$ routing functions.  Each router uses the
$i$-th routing function, denoted as $\sigma_i$, for a packet with a \emph{layer
tag} $i$ attached. 
\fi
All layers but one accommodate a fraction of links, maintaining non-minimal
paths.  One layer (associated with $\sigma_1$) uses all links to host minimal
paths.
%
%
The fraction of links in one layer is controlled by $\rho \in [0;1]$.
%
%
Now, the interplay between $\rho$ and $n$ is important.  More layers (higher
$n$) that are sparse (lower $\rho$) give more paths that are long, giving more
path diversity, but also more wasted bandwidth (as paths are long). More layers
that are dense reduce wasted bandwidth but also give fewer disjoint paths.
Still, this may be enough as we need three paths per router pair.  \emph{One
ideally needs more dense layers or fewer sparse layers}.
%
%
Thus, an important part of deploying FatPaths is selecting the best $\rho$ and $n$ for
a given network.
\iftr
Assuming enough minimal path diversity, we can obtain different routes even
with the same layer. Thus, we try to pick different next-hop choices for each
layer if there is enough minimal-path diversity in the layer. We evaluate
different fractions  of links per layer to find a good compromise, which can
entail $\rho =1$ for all layers if there is high minimal-path diversity in the
topology.
\fi
To facilitate implementation of FatPaths, the project repository contains layer
configurations ($\rho, n$) that ensure high-performance routing for used
topologies.  \emph{We analyze performance of different $\rho$ and $n$
in~\cref{sec:theory}} and~\cref{sec:eval}.

\iftr
\subsubsection{Interplay of Parameters: Discussion}

The layers are intentionally \emph{not} disjoint.
Specifically, we use minimum routing within the layers. Thus, to obtain enough
non-minimal paths, the layers need to by sparse. However, we also want to use
``not too many'' non-minimal paths (``just enough''), so having many edges per
layer is also important.

That is why $\rho$ is an adjustable parameter (that needs to be tuned at least
for a given topology and a given number of layers~$n$): $\rho$ needs to be low
enough to enable enough non-minimal paths at the given number of layers. Thus,
with many layers, one can use higher $\rho$ and still obtain enough path
diversity. The benefit is then having on average more shortest paths, and
therefore less total load (TNL), cf.~\cref{sec:tnl}.
This will be visible in Figure~\ref{fig:ndp_n_rho}: with lower $\rho$, we get
more longer paths, which can improve tail performance (thanks to better load
balancing) but has a cost in average throughput (due to the higher TNL). For
both very low and very high $\rho$, both tail performance and the average
throughput suffer. 

In practice, with the FatPaths layered routing, we obtain many non-disjoint paths
across multiple layers, where many of them are minimal. The FatPaths design
then leaves it up to the flowlet load balancing to decide which ones should be
picked for sending the next flowlet.
\fi

\ifsq\enlargethispage{\baselineskip}\fi


\iftr
\subsubsection{Layer Construction: Random Uniform Edge Sampling}
\fi

An overview of layer construction is in Listing~\ref{lst:layers}. We start
with one layer with all links, maintaining shortest paths.  We use $n-1$
random permutations of vertices to generate $n-1$ random layers.  Each such
layer is a subset $E' \subset E$ with $\lfloor\rho\cdot |E|\rfloor$ edges
sampled u.a.r.. $E'$ may disconnect the network, but for
the used values of $\rho$, this is unlikely and a small number of attempts
delivers a connected network.
{Note that the algorithm for constructing layers is general and can be used with
any topology; cf.~{\cref{sec:back_topos}} and Section~{\ref{sec:discussion}}.}

\ifall
\begin{lstlisting}[aboveskip=-0.95em,abovecaptionskip=-0.05em,belowskip=-2.5em,float=!h,label=lst:layers,caption=\textmd{
\hll{Overview of the algorithm for constructing routing layers.}}]
$L$ = $\{E\}$ //|Init a set of layers $L$; we start with $E$ that corresponds to $\sigma_1$|
$P$ = $\{\pi_1(V), ..., \pi_{n-1}(V)\}$ //|Generate $n-1$ random permutations of vertices|
foreach $\pi \in P$ { //|One iteration of the \underline{main loop} derives one layer|
  $E'$ = layer($\pi$, method); //|Generate a layer with a selected method|
|\vspace{0.25em}|  $L$ = $L \cup \{E'\}$ /*|Store a new layer|*/ } //|End of the \underline{main loop}|
layer($\pi$, method) { //|Derive a layer that corresponds to some $\sigma_i$|
  if(method == simple) //|rnd(0,1) (below) returns a random number $\in [0;1)$|
    foreach $(u,v) \in E$ //|Below, "$\pi(u)$<$\pi(v)$" ensures layer's acyclicity|
      if($\pi(u)$<$\pi(v)$ and rnd(0,1)>$\rho$) then $E'$=$E' \cup (u,v)$
  else {/*|Here, more sophisticated methods can be used|*/} }
\end{lstlisting}
\fi

\iftr
\begin{lstlisting}[aboveskip=0em,abovecaptionskip=0.0em,belowskip=0.0em,float=*,label=lst:layers,caption=\textmd{
Overview of the algorithm for constructing routing layers.}]
$L$ = $\{E\}$ //Init a set of layers $L$; we start with $E$ that corresponds to $\sigma_1$
$P$ = $\{\pi_1(V), ..., \pi_{n-1}(V)\}$ //Generate $n-1$ random permutations of vertices
foreach $\pi \in P$ do: //One iteration derives one layer associated with some $\sigma_i$ 
  $E'$ = $\{\}$; foreach $(u,v) \in E$ do:
    //Below, a condition "$\pi(u)$ <$ \pi(v)$" ensures layer's acyclicity, if needed 
    //Below, a call to rnd(0,1) returns a random number $\in [0;1)$
    if($\pi(u)$ < $\pi(v)$ and rnd(0,1) < $\rho$) then:
      $E'$ = $E' \cup (u,v)$ //Add a sampled edge to the layer
  $L$ = $L \cup \{E'\}$  
\end{lstlisting}
\fi

\iftr
\begin{lstlisting}[aboveskip=0em,abovecaptionskip=0em,belowskip=0em,float=*,label=lst:layersFp,caption=\textmd{
Constructing routing layers in FatPaths {(a variant that reduces Path Overlap)}.}]
$L$ = $\{E\}$ //Init a set of layers $L$; we start with $E$ that corresponds to $\sigma_1$.
$\Pi$ = $\{\pi_1(V), ..., \pi_{n-1}(V)\}$ //Generate $n-1$ random permutations of vertices.

//Init a matrix $W$ containing weights of edges $(u,v) \in E$.
$W = \{ [w_{uv}] \quad | \quad \forall u, v \in V: w_{uv} = 0 \}$

foreach $\pi \in \Pi$ do:  //One iteration of the main loop derives one layer.
  $E'$ = create_layer($\pi$, $E$, $W$, $L_{min}$, $L_{max}$) //Generate a layer.
  $L$ = $L \cup \{E'\}$ //Record the layer

//Derive a layer that corresponds to some $\sigma_i$.
create_layer($\pi$, $E$, $W$, $L_{min}$, $L_{max}$):
  //A condition "$\pi(u) < \pi(v)$" ensures layer's acyclicity.
  $\mathcal{V} = \{ (u,v) \in V \times V : \pi(u) < \pi(v) \}$
  //Init a priority queue $Q$. One queue element is a pair of vertices from $\mathcal{V}$.
  $Q.init(\mathcal{V})$
  $incidence_G = $incidence_matrix$(G)$ //Generate an incidence matrix of $G$.
  $p_{cnt} = 0$ //Init path count variable $p_{cnt}$.
  while ($\mathcal{V} \neq \emptyset$ and $p_{cnt} < M$) do:
    $(u,v) = Q.pop()$ //Get a pair of vertices with lowest priority (lowest number of added paths).
    $path = (v_1, v_2 \dots, v_d) \gets $find_path$ (u, v, W, incidence_G, L_{min}, L_{max})$ //Find path from $u$ to $v$, with path of $length \in [L_{min}, L_{max}]$, minimizing the sum of edge weights, such that for $i < j$: $\pi(v_i) < \pi(v_j)$ and the available edges are given by the $incidence_G$ matrix.
    if $path$ exists: //Check if $u$ and $v$ are connected, in the graph given by the $incidence_G$ matrix.
      $p_{cnt} = p_{cnt} + 1$
      foreach $link \in path$ do:
        $E' = E' \cup \{link\}$ //Add each edge from $path$ to the current layer.
      foreach $v_i, v_j \in path$ where $|i - j| > 1$ do:
        $incidence_G[v_i][v_j] = 0$ //Exclude all edges, which will force the traffic from $v_1$ to $v_d$ to use a different path than one that was found.
      foreach  $v_i, v_j \in path$ where  $j - i < L_{min}$ do:
        $\mathcal{V} = \mathcal{V} \setminus (v_i, v_j)$  // Even if an additional path from $v_i$ to $v_j$ of length at least $L_{min}$ will be added, there will still exist this shorter path between these vertices (the one contained by the path $(v_1, v_2, \dots v_d)$), therefore the pair $(v_i, v_j)$ should be further excluded from the $\mathcal{V}$ set.

find_path($src, dst, W, incidence_G, L_{min}, L_{max}$):
  best_path = null //Init variable containing the path with lowest cost.
  $Q = \{\}$ //Init an empty queue.
  $Q.push(src)$
  while $Q \neq \emptyset$ do:
    $path = Q.pop()$
    if $path.last() == dst$ then:
      if $path.length() > L_{min}$ then:
        path_weight = compute_weights($W, path$) //Sum the weights of all edges that $path$ consists of.
        if path_weight $<$  compute_weights($W$, best_path) then:
          best_path = $path$ //Look for path with lowest cost.
    else if $path.length() < L_{max}$ then:
      foreach $neighbour$ of $path.last()$ where $neighbour \notin path$:
        $Q.push(path \cup \{neighbour\})$
  //Change the weights $W$ of edges from the best path.
  foreach edge $(v_i, v_{i+1}) \in $ best_path do:
    $W[v_i][v_{i+1}] += i \cdot ($best_path$.length() - 1 - i)$
\end{lstlisting}
\fi

\iftr
\subsubsection{Layer Construction: Minimizing Path Overlap}
\fi

We also use a variant in which, instead of randomized edge
picking while creating paths within layers, a simple heuristic
minimizes path interference. For each router pair, we pick a set of paths with
minimized overlap with paths already placed in any of the layers.  Most
importantly, while computing paths, \emph{we prefer paths that are one hop
longer than minimal ones, using the insights from the path diversity analysis
(\cref{sec:paths})}.
\iftr
The algorithm is depicted in Listing~\ref{lst:layersFp}.

In general, the algorithm
%
%
selects pairs of vertices $(u, v) \in V \times V$ with lowest priority from the
priority queue $Q$ (pairs with lowest priority are assigned the smallest 
number of paths that have already been placed in the layers), finds the
shortest directed path from $u$ to $v$ of length at least equal to
$L_{min}$ (a functionality provided by ${find{\_}path}$), and adds the found
path to the layer which is currently being created.

At the moment of initialization, every layer generates $\mathcal{V}$ -- a set
containing all pairs of vertices $(u, v) \in V \times V$. This set will be 
used to supervise the process of vertices selection and to keep track of the
pairs that can still be sampled from $Q$. The $Q.pop()$
function returns a pair of vertices still present in $\mathcal{V}$, which has
the lowest priority in $Q$ (in case of equal priorities, ties are broken randomly).

After processing a pair of vertices $(u, v)$ chosen from $Q$, and adding 
a valid directed path $(v_1, v_2, \dots v_d)$ to a layer (such that $v_1 = u$,
$v_d = v$ and $d \geq L_{min} + 1$), pairs of vertices $(v_i, v_j)$, such that
$j - i < L_{min}$, may be deleted from $\mathcal{V}$, as even if an
additional path from $v_i$ to $v_j$ of length at least $L_{min}$ will be
added, there will still exist a shorter path between these vertices (the one
contained by the path $(v_1, v_2, \dots v_d)$). Thus, the new path will not be
used to forward the traffic between these two vertices in this layer.

Similarly, if the layer already includes a path $(v_1, v_2, \dots v_d)$, then
adding an edge between any two not directly connected vertices $v_1, \dots
v_d$, will force the traffic from $v_1$ to $v_d$ to be forwarded using such a shorter
path, using the added edge. This may not necessarily introduce a
decline of overall achievable throughput. However, it impedes the process of
tracking pairs of vertices, that have already various paths between them
provided. In order to avoid this, we exclude all edges $(v_i, v_j)$, such that
$|i - j| > 1 $. This is achieved through the definition and usage of
$incidence_G$ variable representing the incidence matrix of the graph $G$.

As the ${find\_path}$ function takes the given permutation $\pi$ into
consideration, while searching for the path (the function considers only these
edges $(v_i, v_j) \in E$, for which $\pi(v_i) < \pi(v_j)$), there is no need to
verify whether adding the path to the graph will create any cycle.
Additionally, the edges are being processed by the $find\_path$ function in the
order given by the $W$ matrix, to use the least exploited (so far) edges.

In order to enhance the efficiency of generating each of the layers a constant
$M$ has been introduced, which helps to limit the number of pairs of vertices,
for which paths can be added, per each layer. It also helps to distribute the
paths more evenly between the layers.
\fi

\ifall
The effectiveness of the routing depends strictly on the number and such a
choice of the described layers parameters, that allow to exploit the redundancy
in a given network topology - the parameter playing the key role in this
process is the $L_{min}$ value. In most of the cases the best results were
achieved after assigning it to be one hop longer than the shortest paths in the
given topology graph.
\fi

\iftr
\subsection{Populating Forwarding Entries}

The $\sigma_i$ functions are deployed with forwarding tables. An example scheme
for deriving and populating forwarding tables is illustrated in Appendix~\ref{sec:app_forwarding}.
To derive these tables, we compute minimum
paths between every two routers $s$, $t$ {within layer~$\sigma_i$}.  Then, for
each router~$s$, we populate the entry for $s$, $t$ in $\sigma_i$ with a port
that corresponds to the router $s_i$ that is the first step on a path from $s$
to $t$. We compute all such paths and choose a random first step port, if there
are multiple options.  For any hypothetical network size, constructing layers
is not a computational bottleneck, given the $\mathcal{O}(N^2 \log N)$
complexity of Dijkstra's shortest path algorithm for $N$
vertices~\cite{dijkstra1959note}.
\fi

\ifall
For any hypothetical network size, constructing layers is not a computational
bottleneck, given the $\mathcal{O}(n^2 \log n)$ complexity of Dijkstra's
shortest path algorithm for $n$ vertices~\cite{dijkstra1959note}.
\fi



\ifsq\enlargethispage{\baselineskip}\fi

\iftr
%


\subsection{Implementation of Layers}

We propose two schemes to deploy layers.
First, a simple way to achieve separation is \emph{partitioning of the address
space}. This requires no hardware support, except for sufficiently long
addresses. One inserts the layer tag anywhere in the address, the resulting
forwarding tables are then simply concatenated.  The software stack must
support multiple addresses per interface (deployed in Linux since v2.6.12,
2005).
%
%
Next, similarly to SPAIN~\cite{mudigonda2010spain} or
PAST~\cite{stephens2012past}, one can use \emph{VLANs}~\cite{frantz1999vlan}
that are a part of the L2 forwarding tuple and provide full separation. Still,
the number of available VLANs is hardware limited.
%



\subsection{Implementation of Forwarding Functions}

Forwarding functions can be implemented with simple lookup tables (flat
\emph{Ethernet exact matching} or hierarchical \emph{TCAM longest prefix
matching} tables).
In the former, one entry maps a single input tuple to a single next hop.  The
latter are usually much smaller but more powerful: one entry can provide the
next hop information for many input tuples.

As not all the considered topologies are hierarchical, we cannot use all the
properties of longest match tables. Still, we observe that all endpoints on one
router share the routes towards that router. We can thus use prefix-match
tables to \emph{reduce the required number of entries from $\mathcal{O}(N)$ to
$\mathcal{O}(N_r)$}. This only requires exact matching on a fixed address part.
%
%
As we target low-diameter topologies, space savings due to moving from
$\mathcal{O}(N)$ to $\mathcal{O}(N_r)$ \emph{can be large}. For example,
an SF with $N = 10,830$ has $N_r = 722$.
Such semi-hierarchical forwarding was proposed in
PortLand~\cite{niranjan2009portland} (hierarchical pseudo MAC addresses) or 
shadow MACs~\cite{agarwal2014shadow}.
As we use a simple, static forwarding function, it can also be implemented
on the endpoints themselves, using source routing~\cite{jyothi2015towards}.
 
%
%
\fi

\ifall
\noindent 
\macb{Addressing}
To integrate FatPaths with \emph{L2/Ethernet}, one can use exact match tables;
they should only support masking out a fixed field in the address before
lookup, which could be achieved with, for example, P4~\cite{bosshart2014p4}. 
%
%
Alternatively, one could also use a simple \emph{L3/IP} scheme.
First, every endpoint has an IP address of the form $10.i.s.h$ for each layer
($s$, $h$, and $i$ identify a router, an endpoint within the router, and the
layer ID).
Second, for the inter-router links, addresses from a disjoint range are used,
e.g,. $192.168.*.*$, with one $/30$ subnet per link.
Finally, each router~$s$ has one forwarding rule for each other router, of the
form $10.i.t.*/24 \textit{ via } 192.168.x.y$, where the inter-router link
address is chosen from the router's ports according to the forwarding function
$\sigma_i(s, t)$.

\fi

\iftr

\subsection{Adaptive Load Balancing + Transport}
\label{sec:routing-flowlets}

We now provide more details on the interplay between FatPaths' load balancing
based on LetFlow~\cite{vanini2017letflow} and transport based on
NDP~\cite{handley2017re}, expanding on Section~\ref{sec:overview}.
For portability and performance, we integrate FatPaths with LetFlow, a simple
yet powerful congestion control scheme that uses the middle-ground granularity
of \emph{flowlets}, instead of packets or flows, in load balancing.
LetFlow originally operates within routers.  Instead, we select a layer for
each flowlet \emph{at the endpoint}. 
Now, NDP, that is a basis of the transport protocol in FatPaths, does not
include a proper adaptivity solution (it uses oblivious per-packet load
balancing). In FatPaths, we combine the endpoint-based flowlet adaptivity with
NDP: The destination endpoint, which handles flow control in NDP, can request
to change the layer, when it observes truncated packets signaling network
congestion. This provides the necessary flowlet length elasticity that
implements LetFlow adaptivity.
The layer tags are chosen on the endpoint by the adaptivity algorithm.
\fi



\subsection{Fault-Tolerance}

\iftr
Fault-tolerance in FatPaths is based on preprovisioning multiple paths within
different layers. 
For major (infrequent) topology updates, we recompute
layers~\cite{mudigonda2010spain}. Otherwise, when a failure in some layer is
detected, FatPaths redirects the affected flows to a different layer. 
Here, unlike schemes such as DRILL~\cite{ghorbani17drill}, FatPaths -- by
targeting low-diameter topologies -- is natively  resilient to
the {network asymmetry} also caused by failures. 
Traffic redirection relies on flowlets~\cite{vanini2017letflow}. As with
congestion, the elasticity of flowlets automatically prevents data from using
an unavailable path (i.e., the flowlet based load balancing randomly picks
paths for flowlets and lets their elasticity {automatically} load-balance the
traffic on available paths).
%
%
Thus, its design is inherently asymmetric and failures that for example affect
the availability of shortest paths are redirected at the endpoints by balancing
flowlets.

For failure \emph{detection}, we rely on established
mechanisms~\cite{mudigonda2010spain, jain2011viro, hu2016explicit,
vanini2017letflow, handley2017re}. 


Besides flowlet elasticity, the layered FatPaths design enables
other fault-tolerance schemes.
%
%
For example, FatPaths could limit each layer to be a spanning tree and use
mechanisms such as Cisco's proprietary Per-VLAN Spanning Tree (PVST) or IEEE
802.1s MST to fall back to the secondary backup ports offered by these schemes.
%
%
%
%
Finally, assuming L3/IP forwarding and addressing, one could rely on resilience
schemes such as VIRO's~\cite{jain2011viro}.
\fi

\ifall\maciej{OK data}
\begin{figure*}[t]
\vspace{-1em}
\centering
\includegraphics[width=0.7\textwidth]{theory/{MAWP0_55___mac_e}.pdf}
\vspace{-0.5em}
\caption{\textmd{\textbf{Theoretical analysis of FatPaths performance:} Maximum achievable throughput in FatPaths and other layered routing mechanisms
(traffic intensity: 0.55) for the adversarial traffic pattern that maximizes stress on the interconnect}.}
\vspace{-1.5em}
\label{fig:theory}
\end{figure*}
\fi

\begin{figure*}[t]
\vspaceSQ{-1.25em}
\centering
\begin{minipage}{0.67\textwidth}
\vspaceSQ{-1.8em}
\centering
\includegraphics[width=1.0\textwidth]{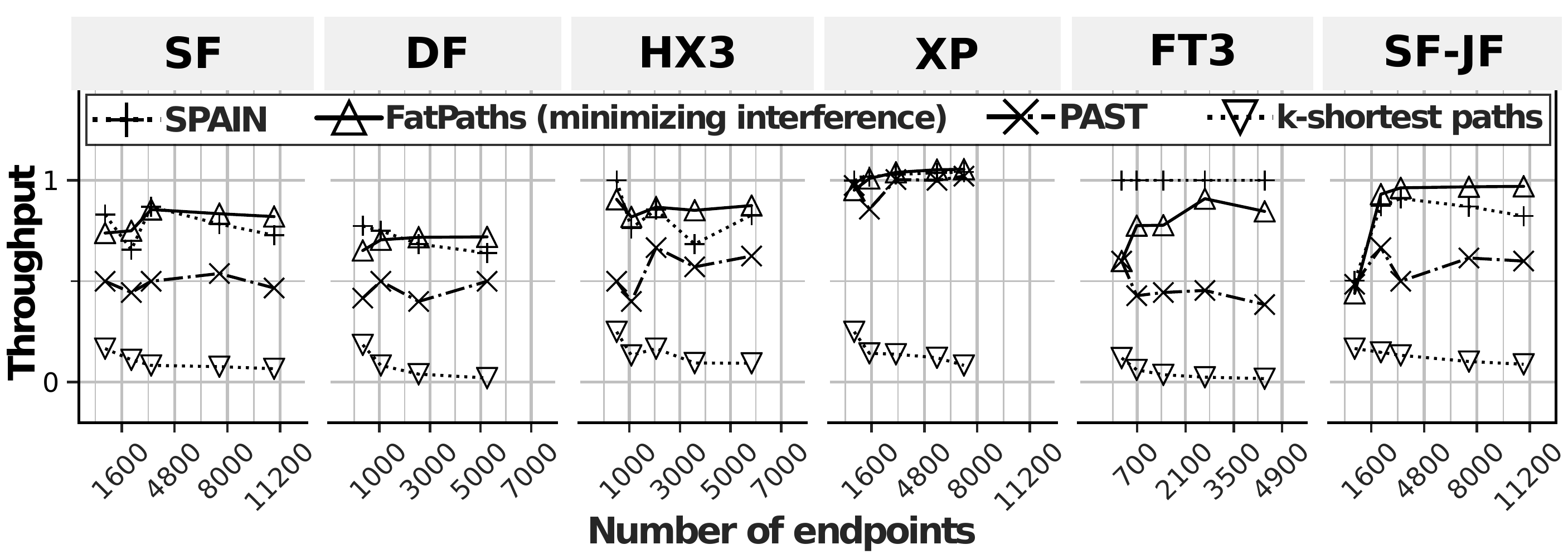}
\vspaceSQ{-2.0em}
\caption{\textmd{\textbf{Theoretical analysis of FatPaths performance:} Maximum
achievable throughput in FatPaths and other layered routing mechanisms (traffic
intensity: 0.55) for the adversarial traffic pattern that maximizes stress on
the interconnect}.}
\vspaceSQ{-1.5em}
\label{fig:theory}
\end{minipage}\hfill
\begin{minipage}{0.3\textwidth}
\centering
\includegraphics[width=1.0\textwidth]{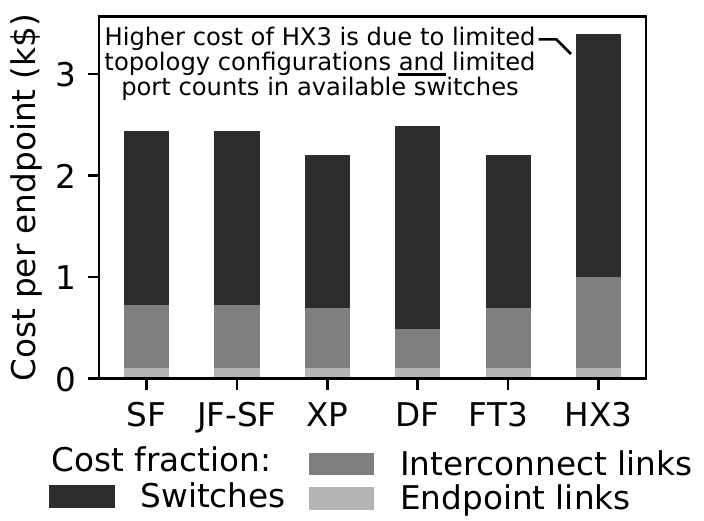}
\vspaceSQ{-1em}
\caption{\textmd{{An example cost model for topologies with
\mbox{$N \approx 10,000$}, assuming the 100GbE equipment.}}}
\label{fig:cost-model}
\end{minipage}
\vspaceSQ{-1.5em}
\end{figure*}


\section{Theoretical Analysis}
\label{sec:theory}

We first conduct a theoretical analysis. The main goal is to illustrate that
{layered routing in FatPaths enables higher throughput than}
SPAIN~\cite{mudigonda2010spain}, PAST~\cite{stephens2012past, hu2016explicit},
and $k$-shortest paths~\cite{singla2012jellyfish}, three recent schemes that
support (1) multi-pathing and (2) disjoint paths (as identified in
Table~\ref{tab:intro}).
%
%
%
%
%
%
SPAIN uses a set of spanning trees, using greedy coloring to minimize their
number; one tree is one layer. Then, paths between endpoints are mapped to the
trees, maximizing path disjointness. PAST uses one spanning tree per host,
aiming at distributing the trees uniformly over available physical links.
Finally, $k$-shortest paths~\cite{singla2012jellyfish} spreads traffic over
multiple shortest paths (if available) between endpoints. 
\iftr
In Appendix~\ref{sec:app_Layers}, we
provide detailed descriptions of SPAIN, PAST, and $k$-shortest paths, and of
how we integrate them in the FatPaths layered routing.
\fi

\ifall
\begin{figure*}[t]
\centering
\includegraphics[width=0.65\textwidth]{theory/{MAWP0_55___mac_e}.pdf}
\vspace{-1em}
\caption{\textmd{\textbf{Theoretical analysis of FatPaths performance:} Maximum
achievable throughput in FatPaths and other layered routing mechanisms (traffic
intensity: 0.55) for the adversarial traffic pattern}.}
\vspace{-1em}
\label{fig:theory}
\end{figure*}
\fi

\iftr
\begin{figure}[h!]
\small
\begin{align}
&\sum_{i=1}^k f_{iuv} \leq c(u,v), \enskip u, v \in V \label{eq:general-capacity}\\
&\sum_{v \in V} f_{iuv} - \sum_{v \in V} f_{ivu} = 0, \enskip i = 1, \dots, k, \enskip u \in V \setminus \{s_i, t_i\} \label{eq:general-flow}\\
&\sum_{v \in V} f_{is_iv} - \sum_{v \in V} f_{ivs_i} = T(s_i,t_i), \enskip i = 1, \dots, k \label{eq:general-demands}\\
&f_{iuv} \geq 0, \enskip u, v \in V, \enskip i = 1,  \dots, k \label{eq:general-nonzero}
\end{align}
\normalsize
\end{figure}
\fi

\iftr
\begin{figure*}[t]
\begin{align}
& -\sum_{v \in V} \sum_{l = 1, \dots n} f_{is_{i}vl} \cdot \delta_{v, \sigma_l(s_i, t_i)} + T(s_i,t_i) \cdot \mathcal{T} \leq 0, \enskip i = 1, 2, \dots, k \label{type0}\\
& \sum_{i =1, \dots k} \sum_{l = 1, \dots n} f_{iuvl} \cdot \delta_{v, \sigma_l(u, t_i)}  \leq c(u, v), \enskip \forall (u, v) \in E \label{type1}\\
& \sum_{v \in V}  f_{iuvl} \cdot \delta_{v, \sigma_l(u, t_i)}  - \sum_{v \in V} f_{ivul} \cdot \delta_{u, \sigma_l(v, t_i)}  = 0, \enskip i = 1, \dots, k, \enskip l = 1, \dots, n, \forall u \in V \setminus \{s_i, t_i\} \label{type2}\\
& \sum_{v \in V} \sum_{l = 1, \dots n} f_{is_{i}vl} \cdot \delta_{v, \sigma_l(s_i, t_i)} \leq \mathcal{T}_{upper bound} \cdot T(s_i,t_i), \enskip i = 1, \dots, k \label{type3}\\
& \sum_{v \in V} \sum_{l = 1, \dots n} f_{ivs_{i}l} \cdot \delta_{s_i, \sigma_l(v, t_i)}  = 0, \enskip i = 1, \dots, k \label{type4}
\end{align}
\end{figure*}
\fi

\iftr
\subsection{Definitions and Formulations}

We now present the associated mathematical concepts. Two important
notions are the \textbf{maximum achievable throughput (MAT)}
and the \textbf{multi-commodity flow (MCF) problem}.

\subsubsection{Multi-Commodity Flow Problem (MCF): Intuition}

Assume a network with predefined (1) flow demands between source and
destination endpoints, and (2) constraints on the capacities of links. Now, the
problem of assigning an amount of flow to be supplied between the endpoint
pairs, which is feasible in the network (i.e., the sums of flows between
endpoint pairs do not exceed the capacities of respective links) is called the
Multi-Commodity Flow (MCF) problem~\cite{even1975complexity}. 

\subsubsection{Maximum Achievable Throughput (MAT): Intuition}

Intuitively, the maximum achievable throughput may be defined as the ratio of
the flow that can be forwarded between each pair of network endpoints to the
demanded amount of flow for this pair (for a given traffic pattern).
In our analysis, to derive a single MAT value for all communicating endpoints
(for a given traffic matrix, a given topology, and a given routing scheme), we
maximize the minimum flow across all flow demands, assuming each pair of
endpoints is able to send and receive flows concurrently. The derivation is
done using linear programming (LP).

\subsubsection{MCF Linear Programming Formulations for MAT}
\label{sec:mcf-fp}

We now presented the MCF LP formulations, both for
the general problem of maximizing throughput under general
routing, and when a layered routing is assumed.

%
MAT is formally defined as the maximum value $\mathcal{T}$ for which there
exists a feasible multi-commodity flow that routes a flow $T(s, t) \cdot
\mathcal{T}$ between all router pairs $s$ and $t$, satisfying link capacity and
flow conservation constraints. $T(s,t)$ specifies traffic demand; it is the
amount of requested flow from $s$ to $t$ (more details about MAT and $T(s,t)$ are
in~\cite{jyothi2016measuring}).


Let us denote the directed graph representing the network topology by  $G = (V,
E)$, where $V$ is the set of vertices (switches) and $E$ is the set of edges.
Each edge $(u,v) \in E$ represents a network link with non-negative capacity
$c(u,v)$.

\paragraph{General MCF Formulation}

Given a graph $G$ defined as above, let us denote $k$ different commodities by
$K_1, K_2, \dots, K_k$, where $K_i = (s_i, t_i, T(s_i,t_i))$. A single commodity
indicates a single flow demand between some communicating pair of endpoints. Here,
$s_i$ is the source node of commodity~$i$ while $t_i$ is the destination node. The
variable $f_{iuv}$ defines the fraction of demanded flow of commodity $i$ from
vertex $u$ to vertex $v$. $f_i$ is a real-valued function satisfying the flow
conservation and capacity constraints, where $f_i \in [0, 1]$. The MCF problem
may be formulated as the problem of finding an assignment of all flow variables
satisfying the following constraints:

Eq.~(\ref{eq:general-capacity}) defines the capacity constraints.
Eq.~(\ref{eq:general-flow}) defines the flow conservation on non-source and
non-destination nodes. Eq.~(\ref{eq:general-demands}) defines the flow
conservation at the source node. Finally, Eq.~(\ref{eq:general-nonzero})
defines the non-negativity constraints.

\paragraph{MCF Formulations for Layered Routing in FatPaths}

In contrast to the original multi-commodity flow problem, whose aim was finding
an assignment of all flow variables satisfying the constraint rules, the main
goal of the linear program used for the purpose of our evaluation is maximizing
the possible achieved throughput $\mathcal{T}$. Thus, new rules had to be
defined for each of the layered routing schemes.

The following rules were defined in order to formulate the LP programs, used
further to benchmark the FatPaths and SPAIN routing schemes, assuming $n$
layers and $k$ commodities. The variable $f_{iuvl}$ defines
the fraction of the flow of commodity $i$, carried by a link~$(u,v)$ in a layer~$l$.
Moreover, the $\sigma_i$ function represents a routing function for layer $i$. For
routers $s, t \in V$, the value of function $\sigma_i(s,t)$ denotes the next
hop router in the path from $s$ to $t$ within layer $i$. Thus, denoting the
Kronecker delta by $\delta_{x,y}$, we can easily define function $\delta_{v,
\sigma_i(u,t)}$, whose output is equal to $1$ if a certain flow $i$ to a
destination node $t$ is allowed to pass through a certain link $(u,v)$ and $0$
otherwise.
We use $\delta_{v,\sigma_i(u,t)}$ to ensure that flows within specific
layers remain in these layers only.

Now, the main goal of Eq.~(\ref{type0}) and~(\ref{type3}) is to ensure that all
flows allocated for a certain commodity will exit the source node and that it
will not exceed some defined upper bound~$\mathcal{T}_{upperbound}$ (defined
for optimization purposes).
Specifically, the sum in Eq.~(\ref{type0}), for a given commodity~$i$, is equal
to the sum of all traffic fractions traversing from the source node $s_i$ to
the next hop node in any layer.
The aim of Eq.~(\ref{type1}) is to ensure that the total summed flow will not
exceed the link capacity. We assume that the layers provide only a logical
representation of graphs built upon existing physical network wiring. Thus, the
summed flows from all layers should not exceed the capacity on the physical
link (on the left side of Eq.~(\ref{type1}), we sum all the traffic (from each
commodity and in each layer) that traverses the link from node $u$ to $v$).
Eq.~(\ref{type4}) is introduced in order to avoid the flow traversing back to
the source node through any of the links. On the left side of the equation, we
sum all the fractions of the flow in each layer, which traverse from any node $v$
(such that source is the next hop on the path from $v$ to $t_i$ within this layer)
to source $s_i$.
Finally, Eq.~(\ref{type2}) is the most
important constraint to enable benchmarking of the layered routing -- its main
goal is to ensure that \emph{no flow leaks between the layers}.
Here, the first sum in the equation is equal to the total traffic, which is a
fraction of commodity $i$ traffic, that traverses from $u$ to $v$ within layer $l$.
The second sum is the total traffic belonging to commodity $i$, which traverses
the link from $v$ to $u$ in layer $l$. These two sums have to be equal for every
layer, as no traffic is allowed to leak between layers.

\fi

\subsection{{Analysis of Number of Layers}}

\iftr
Both SPAIN and PAST use trees as layers, contrarily to FatPaths that allows for
arbitrary connected subgraphs. 
\fi
This brings many drawbacks, as each SPAIN layer can use at most $N_r-1$ links,
while the topology contains $\frac{N_r k'}2$ links. Thus, at least
$\mathcal{O}(k')$ layers are required to cover all minimal paths, and SPAIN
requires even $\mathcal{O}(N_r)$ on many topologies.  Moreover, PAST always
needs $O(N)$ trees by its design.
By using layers that are arbitrary (but connected) subgraphs and contain a
large, constant fraction of links, \emph{FatPaths provides sufficient path
diversity with a low, $O(1)$ number of layers}.
%

\subsection{{Analysis of Throughput}}

\ifsq
\fi

\iftr
We now derive the MAT~$\mathcal{T}$.
\fi
We test {all considered topologies, topology sizes, traffic patterns
and intensities (fraction of communicating endpoint pairs)}.
We consider two FatPaths variants from~\cref{sec:routing-layers}.
\iftr
We use TopoBench, a throughput evaluation tool~\cite{jyothi2016measuring} that
uses linear programming (LP) to derive $\mathcal{T}$. We extended TopoBench's
LP formulation of MCF to include layered routing according to our
formulation from~\cref{sec:mcf-fp}.
\fi

We use a recently proposed worst-case traffic pattern which maximizes stress on
the network while hampering effective routing~\cite{jyothi2016measuring},
see Figure~\ref{fig:theory}. This pattern {is generated \emph{individually}
for each topology}; it uses maximum weighted matchings to find a
  pairing of endpoints that maximizes average flow path length, with both
  elephant and small flows.
As expected, SPAIN -- a scheme developed specifically for Clos -- delivers more
performance on fat trees. Yet, it uses up to $O(N_r)$ layers.  The layered
routing that minimizes path interference generally outperforms SPAIN on other
networks (we tuned SPAIN to perform as well as possible on low-diameter
topologies). Finally, also as expected, our heuristic that minimizes path
overlap delivers more speedup than simple random edge picking (we only plot the
former for more clarity).

%
\ifall
{Routing schemes use equally many layers to fix the amount of hardware resources.
Using more layers accelerates all comparison targets but also increases
counts of forwarding entries in routing tables. With more layers, SPAIN and PAST
become competitive with FatPaths on fat trees, but they use up to $O(N_r)$.
Still, \emph{they do \textbf{not} outperform FatPaths}.
They use up to $O(N_r)$.}
\fi
Tested schemes use equally many layers ($n$) to fix the amount of HW resources.
Increasing $n$ accelerates all comparison targets but also increases counts of
forwarding entries in routing tables. Here, SPAIN and PAST become faster on fat
trees and approach FatPaths, but they use up to $O(N_r)$ layers.  FatPaths
maintains its advantages for different traffic intensities.
\ifall
\maciej{fix}
i.e.,
different fractions of communicating endpoint pairs
(we tried values in the whole range $(0; 1]$) or any other parameters.
\fi
As expected, our heuristic that minimizes path overlap
outperforms a simple random edge picking.

\noindent
\macb{\ul{Takeaway} }
{FatPaths layered routing outperforms {competitive} schemes
in the used count of layers (and thus the amount of {needed
hardware resources}) and {achieved throughput}}.


%



\section{Simulations}
\label{sec:eval}

\ifsq\enlargethispage{\baselineskip}\fi

We now illustrate how low-diameter topologies equipped with FatPaths outperform
novel high-performance fat tree designs.
\iftr
Various additional analyses are
in Appendix~\ref{sec:app-simulations-full-data}.
\fi


\subsection{{Methodology, Parameters, and Baselines}}


We first discuss parameters, methodology, and baselines.
%


%
\subsubsection{{Topologies and Traffic Patterns}}

%
We use all topologies specified in~\cref{sec:back_topos}: SF, XP, JF, HX, DF,
and FT, in their most beneficial variants (e.g., the ``balanced''
Dragonfly~\cite{kim2008technology}). We fix the network size~$N$ ($N$ varies by
up to $\approx$10\% as there are limited numbers of configurations of each
network). SF represents a recent family of diameter-2 topologies such as
Multi-Layer Full-Mesh~\cite{kathareios2015cost} and Two-Level Orthogonal
Fat-Trees~\cite{valerio1994recursively, valiant1982scheme}. To achieve similar
costs and \mbox{$N$} we use \mbox{2$\times$} oversubscribed fat trees.
\iftr
Here, HyperX (HX) is a \emph{special case}: It has \emph{exceedingly
high-radix routers} but little design flexibility. Thus, it achieves high
performance, but at high cost and power consumption~\cite{kim2007flattened,
besta2018slim}. We still evaluate it to consider most recent low-diameter
networks.
\fi



%
We use the traffic patterns discussed in~\cref{sec:back}, in both 
{randomized} and {skewed non-randomized} variants.

\subsubsection{{Cost Model for Using Topologies of Comparable Cost}}
{We now show how to ensure comparable construction cost of considered
topologies. For this, we use the established cost models from past
works~\mbox{\cite{besta2014slim, kim2008technology, kim2007flattened}}.
Overall, for each ``network size category'' (\mbox{$N \in \{\approx1\text{k},
\approx10\text{k}, \approx100\text{k}, \approx1\text{M}\}$},
cf.~\mbox{\cref{sec:back_topos}}), we search for specific topology
configurations with minimal differences in their sizes~\mbox{$N$}. Their total
cost is derived based on existing router and cable cost models based on linear
functions~\mbox{\cite{besta2014slim, kim2008technology, kim2007flattened}}.
parametrized with prices of modern equipment (e.g., Mellanox equipment listed
on ColfaxDirect {\texttt{http://www.colfaxdirect.com}}). The models distinguish
between fiber and copper cables; the former are used for longer
router-router links (e.g., inter-group links in DF) and the latter forming
short endpoint and router-router connections (e.g., intra-group links in DF or
SF).  As used topology configurations vary in~$N$ (there is always a limited
number of configurations of each used topology), for fairness, the final prices
are normalized per single endpoint. An example cost model, for
\mbox{$N \approx 10,000$}, with the 100Gb Ethernet equipment, is in
Figure~\mbox{\ref{fig:cost-model}}; the prices follow past
data~\mbox{\cite{besta2014slim}}. One can distinguish effects caused by topology details,
for example lower cable costs in DF due to relatively few
expensive global inter-group connections. Variations in final costs are caused
by a limited number of topology configurations combined with limited counts of
ports in available switches.}

\iftr
\begin{figure*}[t]
\vspaceSQ{-0.75em}
\centering
\includegraphics[width=0.65\textwidth]{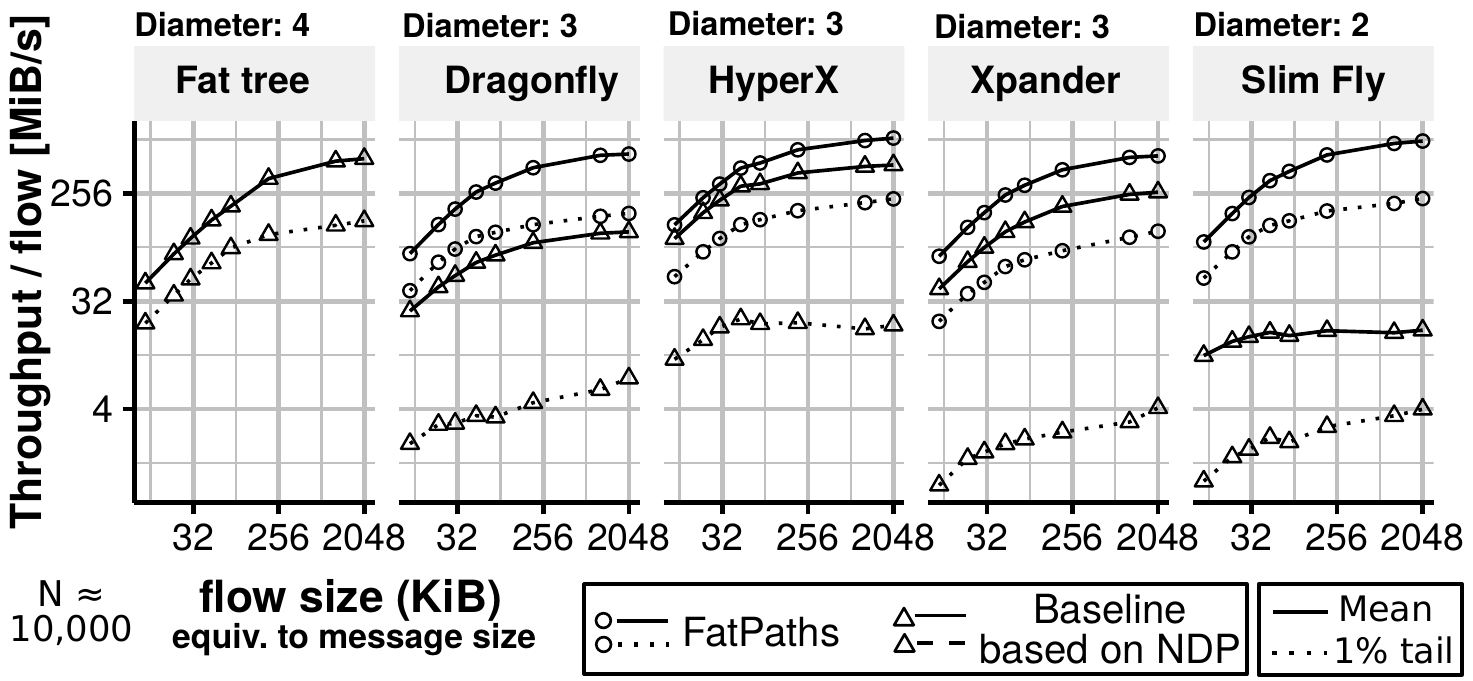}
\vspaceSQ{-1.5em}
\caption{\textmd{Performance analysis of a skewed adversarial traffic for similar-cost networks.
%
}}
\vspaceSQ{-0.75em}
\label{fig:ndp_results_skewed}
\end{figure*}
\fi

\subsubsection{{Routing and Transport Schemes}}

%
We use flow-based non-adaptive ECMP as the {baseline} (routing performance
lower bound). Low-diameter topologies use FatPaths while fat trees use NDP with
all optimizations~\cite{handley2017re}, additionally enhanced with
LetFlow~\cite{vanini2017letflow}, a recent scheme that uses flowlet switching
for load balancing in fat trees.
%
%
We also compare to a fat tree system using NDP with per-packet
congestion-oblivious load balancing as introduced by Handley et
al.~\cite{handley2017re}. 
%
%
%
%
%
For FatPaths, we vary $\rho$ and $n$ to account for \emph{different
layer configurations}, including $\rho = 1$ (minimal paths only).
%
%
%
We also consider simple TCP, MPTCP, and DCTCP with
ECN~\cite{alizadeh2011data, ramakrishnan2001addition, floyd1994tcp},
illustrating that FatPaths can accelerate not only bare Ethernet
systems but also cloud computing environments that usually use full TCP
stacks~\cite{isobe2014tcp, azodolmolky2013cloud}.

\subsubsection{{Flows and Messages}}

\ifsq\enlargethispage{\baselineskip}\fi

%
We vary flow sizes (and thus message sizes as a flow is equivalent to a message) from 32 KiB to 2 MiB.
We use a Poisson distributed {flow arrival rate} with a $\lambda$ parameter. 
\iftr
{Flow sizes}~$v$ are chosen according to the pFabric web search
distribution~\cite{alizadeh2013pfabric}, discretized to 20 flows, with an
average flow size of 1MB.
%
%
For the average $v \approx $1MB, we can handle $\approx$1 flow/ms
(the {injection rate}) on a 10G link ($\lambda = 1000$ [flows/endpoint/s]).
Yet, both our simulations show
that even at $\lambda=200$, there is typically more than one concurrent flow
per endpoint.
%
%
%
This only puts additional stress on flow control, which is not our main focus.
Thus, we choose $\lambda = 200$ (in TCP simulations) and $\lambda=300$ (in NDP
simulations with higher throughput and better flow control).
Details and data on flow behavior are in Appendix~\ref{sec:app-flow-omnet} and~\ref{sec:app-flow-htsim}.
\fi

\ifall
The flow sizes $v$ are chosen according to the pFabric web search
distribution~\cite{alizadeh2013pfabric}, discretized to 20 flows, with an
average flow size of 1MB. The number of flows per endpoint is chosen to
approximate the flow size distribution sufficiently well, while keeping the
total transferred volume low enough for packet-level simulation. The flow sizes
are not chosen randomly; rather, each endpoint uses the same fixed set of flows
(20 for INET, 40 for htsim), representing the distribution. All endpoints
communicate, which leads to a high uniform network load; this is a challenging
scenario for non-minimal routing due to the TNL constraint. For the skewed
test, we reduce TNL by having only 10\% of routers active, using a
non-randomized off-diagonal traffic pattern.  This causes p-way path collisions
on all flows, to show the full capability of non-minimal routing.
\fi

\subsubsection{{Metrics}}

%
We use (1) {flow completion time} (FCT), which also represents
(2) {throughput per flow} $\text{TPF} =
\frac{\text{flow size}}{\text{FCT}}$. 
We also consider (3) {total time to complete} a tested workload
\ifsq
\footnote{\scriptsize{When reporting some runtimes
(cf.~Figures~{\ref{fig:tcp_topo}}-{\ref{fig:tcp_topo_ttc}}), we use a relative
speedup over the plain ECMP baseline for clarity of presentation (as each plot
contains runtimes for flows of different sizes, some absolute runtime data
becomes hard to read).}}.
\fi
\iftr
\footnote{{When reporting some runtimes
(cf.~Figures~{\ref{fig:tcp_topo}}-{\ref{fig:tcp_topo_ttc}}), we use a relative
speedup over the plain ECMP baseline for clarity of presentation (as each plot
contains runtimes for flows of different sizes, some absolute runtime data
becomes hard to read).}}.
\fi

\iftr
\begin{figure*}[t]
\vspaceSQ{-0.5em}
\centering
\includegraphics[width=0.7\textwidth]{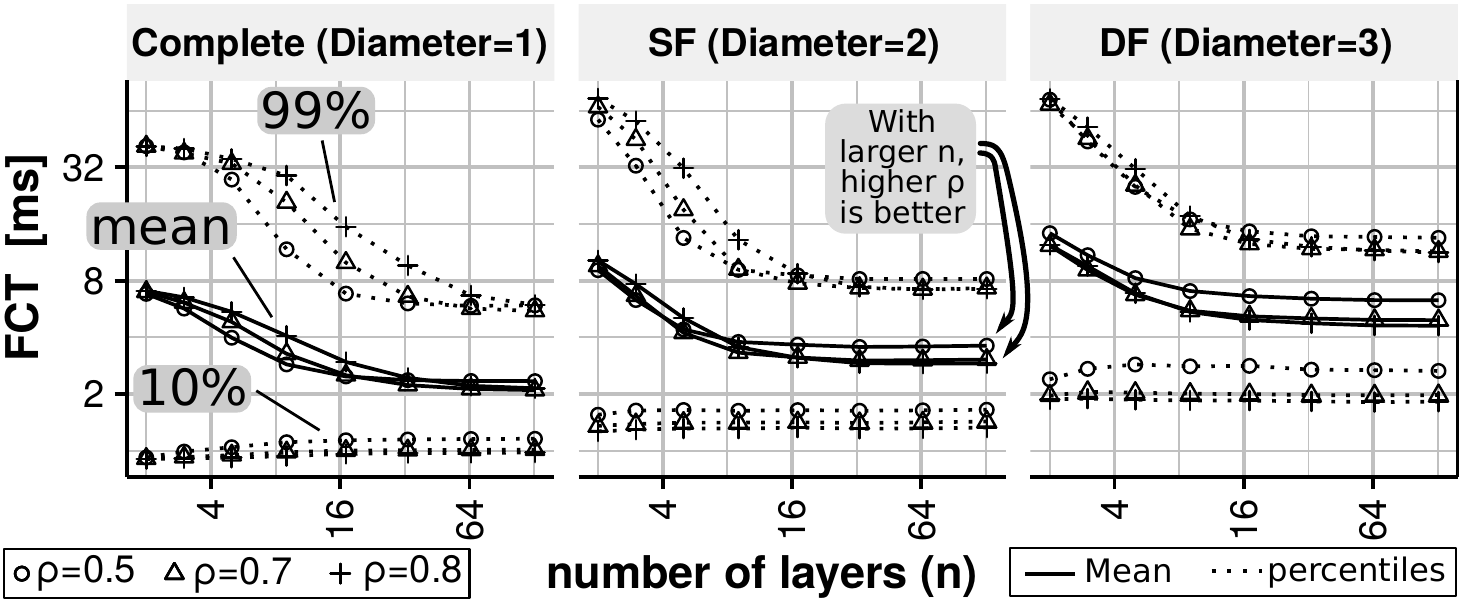}
\vspaceSQ{-1.5em}
\caption{\textmd{{Effects of \#layers $n$ and the amount of remaining edges $\rho$ on FatPaths}, on long flows (size $1\si{\mebi\byte}$); $N \approx 10,000$ (htsim).}}
\vspaceSQ{-1em}
\label{fig:ndp_n_rho}
\end{figure*}
\fi

\subsubsection{{Simulation Infrastructure and Methodology}}

%
We use the OMNeT++~\cite{varga2001omnet++,varga2008overview} parallel discrete
event simulator with the INET model package~\cite{inet} and the \emph{htsim}
packet-level simulator with the NDP reference
implementation~\cite{handley2017re}. OMNeT++ enables detailed
simulations of full Ethernet/TCP networking stack, with
all overheads coming from protocols such as ARP. We use htsim as its simplified
structure enables simulations of networks of much larger scales. {We extend
both simulators with any required schemes, such as flowlets,
ECMP, layered routing, workload randomization}. In LetFlow, we use precise
timestamps to detect flowlets, with a low gap time of $50\si{\micro\second}$ to
reflect the low-latency network. As INET does not model hardware or software
latency, we add a $1\si{\micro\second}$ fixed delay to each link. 
\iftr
We also extend the INET TCP stack with ECN (RFC~3168~\cite{rfc3168}), MPTCP
(RFC~6824~\cite{rfc6824}, RFC~6356~\cite{rfc6356}), and DCTCP. We extend the
default router model with ECMP (Fowler-Noll-Vo
hash~\cite{kornblum2006identifying}) and LetFlow.
In htsim, we use similar parameters; they match those used by Handley et al..
We extend htsim to support arbitrary topologies, FatPaths routing and
adaptivity, and our workload model.
\fi
\iftr
Routers use tail-dropping with a maximum queue size of 100 packets per port. 
ECN marks packets once a queue reaches more than 33 packets.  Fast
retransmissions use the default threshold of three segments.  We also model a
latency in the software stack (corresponding to interrupt throttling) to 100kHz
rate. 
For FatPaths, we use 9KB jumbo frames, an 8-packet congestion window, and a
queue length of 8 full-size packets. 
\fi
All our code is available online.

\iftr

\subsubsection{Scale of Simulations}
We fix various scalability issues in INET and OMNeT++ to allow parallel
simulation of large systems, \emph{with up to $\approx$1 million endpoints}.
%
%
%
To the best of our knowledge, \emph{we conduct the largest shared-memory simulations
(endpoint count) so far in the networking community}.
%
%

\fi

\subsubsection{{Gathering Results and Shown Data}}

%
We evaluate each combination of topology and routing method. 
%
As each such
simulation contains thousands of flows with randomized source, destination,
size, and start time, we only record per-flow quantities; this suffices 
for statistical significance.
%
\iftr
A run might be affected by a degenerate randomized topology instance, traffic
pattern assignment, or layer construction, but for the parameters that we
consider, such instances are unlikely and we never observed any during the
course of our experiments. 
\fi
%
%
We simulate a fixed number of flows starting in a fixed time window, and
drop the results from the first window half for warmup. 
%
%
We summarize the resulting distributions with arithmetic means of the
underlying time measurements, or percentiles of distributions.

\ifsq\enlargethispage{\baselineskip}\fi

\iftr
When some variants or parameters are omitted (e.g., we only show SF-JF to cover Jellyfish),
this indicates that \emph{the shown data is representative}.
\fi

\subsection{Performance Analysis: HPC Systems}
\label{sec:eval-ndp}

First, we analyze FatPaths with Ethernet \emph{but without
the TCP transport}. This setting represents HPC systems that use
Ethernet for its low cost, but avoid TCP due to its performance issues.
%
%
We use \emph{htsim} that can deliver such a setting.

\subsubsection{{{Low-Diameter Networks + FatPaths Beat Fat Trees}}}
%
%
We analyze Figure~\ref{fig:ndp_results_motiv} (page~2, randomized
workload) and Figure~\ref{fig:ndp_results_skewed} (skewed non-randomized
workload). In each case, low-diameter topologies outperform similar-cost fat trees, with up
to 2$\times$ and 4$\times$ \emph{improvement} in throughput for non-randomized
and randomized workload, respectively. Both fat tree and low-diameter networks
use similar load balancing based on flowlet switching and purified transport.
Thus, the advantage of low-diameter networks is their \emph{low diameter}
combined with the ability of FatPaths to \emph{effectively use the diversity of
``almost'' minimal paths}.
%
%
Answering one of two main questions from~\cref{sec:intro}, we conclude that
\emph{{FatPaths enables low-diameter topologies to outperform state-of-the-art
fat trees.}}

%


%
\subsubsection{{FatPaths Uses ``Fat'' Non-Minimal Path Diversity Well}}
%
%
%
We now focus on skewed non-randomized
workloads, see Figure~\ref{fig:ndp_results_skewed}. Non-minimal balanced routing over FatPaths layers, in each
low-diameter topology, \emph{leads to an up to $30\times$ FCT
improvement} over minimal routing (i.e., ``circles on topology~X outperform
triangles on~X'').
The exception is HyperX, due to its higher minimal path diversity (cf.~Figure~\ref{fig:shortest_path_multiplicity}).
\iftr
The traffic causes $p$-way collisions. In HX, its low $p$ explains the
relatively good performance of minimal routing on HX. 
\fi
Thus, \emph{FatPaths effectively leverages the non-minimal path diversity}.

\subsubsection{{What Layer Setup Is Best?}}

We also study the impact of the number~$n$ and the sparsity~$\rho$ of layers in
FatPaths on performance and collision resolution; see
Figure~\ref{fig:ndp_n_rho} (layers are computed with random edge sampling,
cf.~Listing~\ref{lst:layers}).
%
%
%
Nine layers (one complete and eight sparsified) suffice for three disjoint
paths per router pair, resolving most collisions for SF and DF (other networks
behave similarly).
%
%
To understand which $n$ resolves collisions on global channels in DF, we use a
complete graph. Here, more layers are needed, since higher-multiplicity path
collisions appear (cf.~the 99\% tail).
%
Moreover, when more layers \emph{can} be used, a higher $\rho$ \emph{is better}
(cf.~FCT for $n=64$).  This reduces the maximum achievable path diversity, but
also keeps more links available for alternative routes \emph{within each
layer}, increasing chances of choosing disjoint paths. It also increases the
count of minimal paths in use across all entries, reducing total network load.

\ifsq\enlargethispage{\baselineskip}\fi

\iftr

\subsubsection{{FatPaths Scales to Large Networks}}
We also simulate large-scale SF, DF, and JF (other topologies lead to excessive
memory use in the simulator).
We start with SF, SF-JF, and DF ($N \approx\num[group-separator={,}]{80000}$)
in Figure~\ref{fig:ndp_big}. A slight mean throughput decrease compared to the
smaller instances is noticeable, but latency and tail FCTs remain tightly
bounded. The comparatively bad tail performance of DF is due to path overlap on
the global links, where the adaptivity mechanism needs to handle high
multiplicities of overlapping flows. 
We also conduct runs with $N \approx 1,000,000$ endpoints. Here, we illustrate
the distribution of the FCT of flows for SF and SF-JF.  
Flows on SF tend to finish later that on SF-JF.

\begin{figure}[h]
\vspaceSQ{-0.75em}
\centering
\includegraphics[width=0.525\columnwidth]{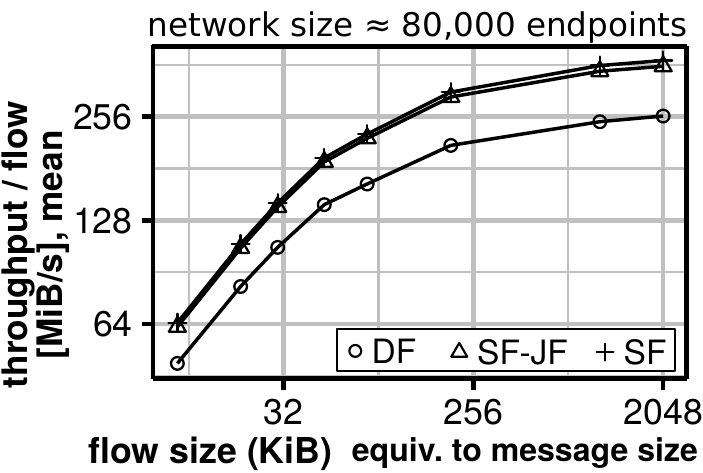}%
\includegraphics[width=0.475\columnwidth]{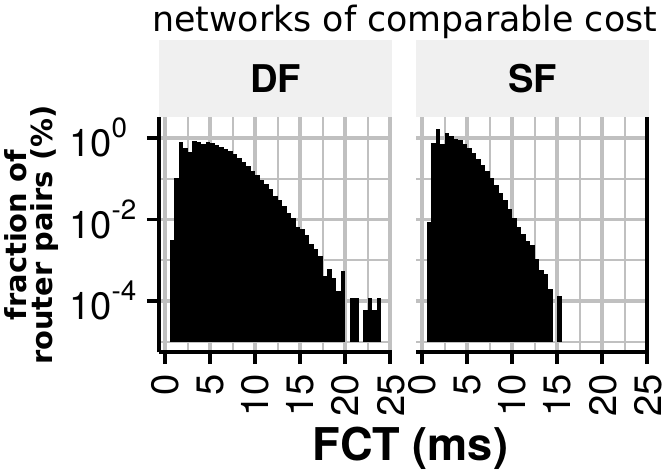}\\
\vspaceSQ{0.5em}
\includegraphics[width=1\columnwidth]{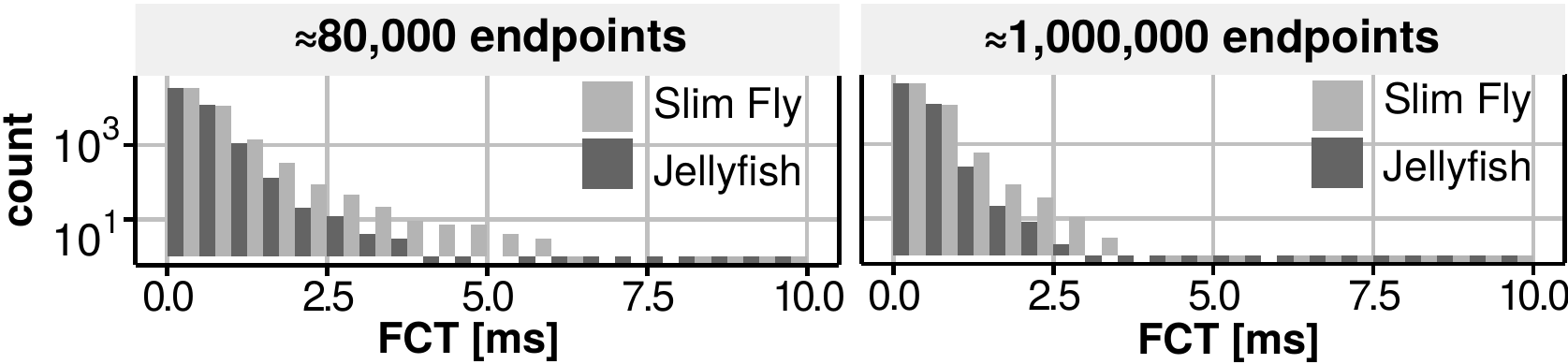}%
\vspaceSQ{-0.5em}
\caption{\textmd{FatPaths on large networks; FCT histograms for flow size $1\si{\mebi\byte}$ (htsim).}}
%
%
%
\vspaceSQ{-1em}
\label{fig:ndp_big}
\end{figure}

\fi





%

\subsection{Performance Analysis: Cloud Systems}
\label{sec:clouds}
 
We also analyze FatPaths on networks with Ethernet \emph{and full TCP
stack}. This represents TCP data centers often used as cloud
infrastructure~\cite{isobe2014tcp}.
Here, we use OMNeT++/INET.

We compare FatPaths to ECMP (traditional static load balancing) and
LetFlow (recent adaptive load balancing), see Figure~\ref{fig:tcp_topo}. The
number of layers was limited to $n=4$ to keep routing tables small; as they are
precomputed for all routers and loaded into the simulation in a configuration
file (this turned out to be a major performance and memory concern).
Most observations follow those from~\cref{sec:eval-ndp}, we only summarize
TCP-related insights.

\begin{figure}[t]
\vspaceSQ{-1em}
\centering
\includegraphics[width=0.5\textwidth]{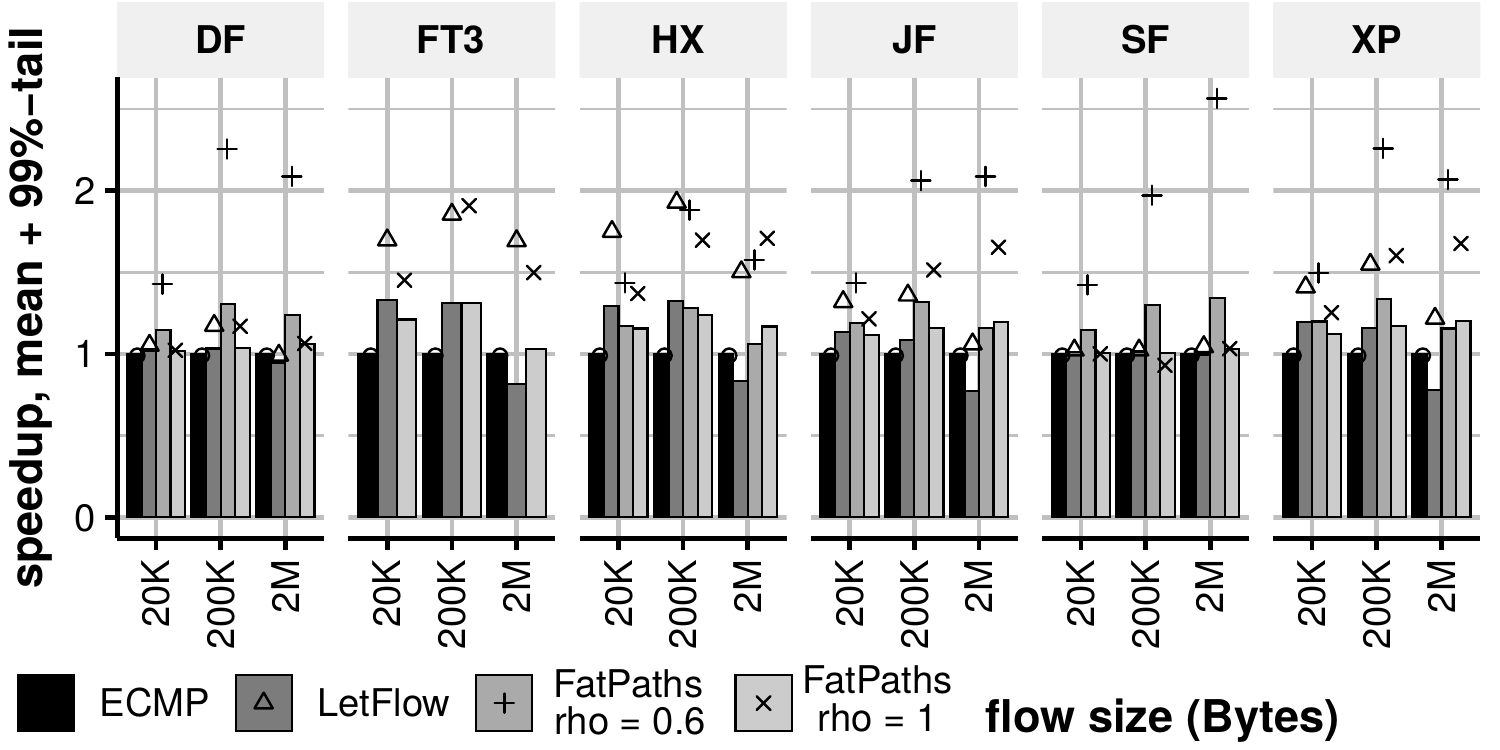}%
\vspaceSQ{-0.5em}
%
%
\caption{\textmd{{FatPaths+TCP} compared to 
ECMP and LetFlow} (mean and 99\% tail values).
Some flows on SF finish more than 2.5$\times$ faster with FatPaths than ECMP or LetFlow.}
\vspaceSQ{-0.5em}
\label{fig:tcp_topo}
\end{figure}

\iftr
\begin{figure*}[t]
\vspaceSQ{-0.5em}
\centering
\includegraphics[width=0.7\textwidth]{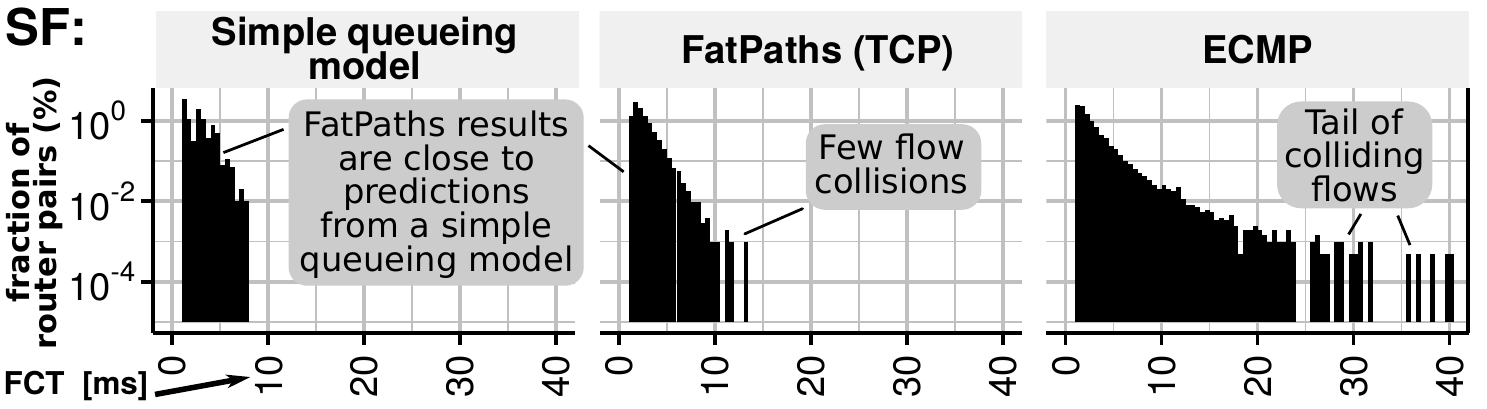}
\vspaceSQ{-1em}
\caption{\textmd{\textbf{FCT distribution of long flows ($1\si{\mebi\byte}$) on SF}. 
FatPaths on TCP with non-minimal routing approaches predictions from a simple queuing
model (model details omitted due to space constrains); ECMP has a long tail of
colliding flows.}}
\vspaceSQ{-1em}
\label{fig:tcp_histo}
\end{figure*}
\fi

LetFlow improves tail and short flow FCTs at the cost of long flow throughput,
compared to ECMP. Both are ineffective on SF and DF which have little
minimal-path diversity.  \emph{Non-minimal routing in FatPaths and $\rho=0.6$
fixes it}, even with only $n=4$ layers. On other topologies, even with minimal
paths ($\rho=1$), \emph{FatPaths adaptivity outperforms ECMP and LetFlow.}
\iftr
Specifically, \emph{with non-minimal routing, we observe large improvements in the tail FCT
as well as mean throughput in DF and SF}. An analysis of the FCT
distributions in Fig.~\ref{fig:tcp_histo} shows that with minimal routing and
low minimal-path diversity, there are many flows with low performance due to
path conflicts, although they do not vastly affect the mean
throughput. \emph{FatPaths fully resolves this issue and produces an FCT
distribution much closer to the queuing model prediction}. Short-flow FCTs are
dominated by TCP flow control effects, which are mostly unaffected by routing
changes.
\fi

\ifsq\enlargethispage{\baselineskip}\fi

We also analyze in detail performance effects in flows of different sizes
vs.~different layer configurations. The findings match those in the ``bare
Ethernet'' simulations in~\cref{sec:eval-ndp}.
For example, for large flows (1MiB), with $n=4$, the higher $\rho$ is, the
faster flows finish.  The largest impact of non-minimal routing is for DF and
SF, with a 2$\times$ improvement in tail FCT; small improvements on tail FCT
are seen in all topologies. 
%

\iftr
JFs, which provide some minimal-path diversity, can also benefit from
non-minimal routing in tail and short flow FCT. However, there is a cost in
long flow throughput due to the higher total network load with non-minimal
paths. To understand this effect better, Figure~\ref{fig:tcp_rho} shows the
impact of the fraction of remaining edges $\rho$ in each layer, and therefore
the amount of non-minimal paths, on FCT for long flows. The optimum choice of
$\rho$ matches the findings from the Ethernet simulations in~\cref{sec:eval-ndp} for SF and DF. 
For the
HyperX topologies, which provide two and three disjoint minimal paths for most
router pairs, a different effect can be observed: since there is a higher minimal path
diversity, they are typically selected more often over non-minimal ones. This is no
better than simply choosing random minimal paths, as it happens at $\rho=1$.
%
%
Using longer paths just
increase the network load and benefit only a few flows in the extreme tail. 
%

\begin{figure*}[b]
\vspaceSQ{-1em}
\centering
\includegraphics[width=0.65\textwidth]{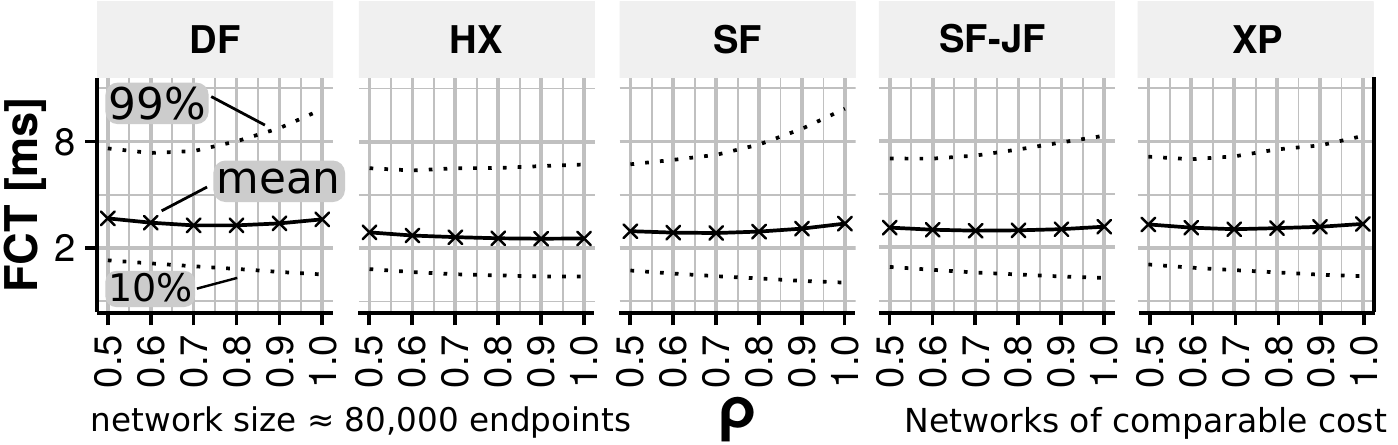}
\vspaceSQ{-0.5em}
\caption{\textmd{\textbf{Impact of $\rho$ on long-flow (1MiB) FCT with FatPaths on TCP},
$n=4$. The largest impact of non-minimal routing is for DF and SF, with
a 2$\times$ improvement in tail FCT; small improvements on tail FCT are seen in all
topologies, but there are no throughput improvements on networks with higher
minimal-path diversity.}}
\vspaceSQ{-0.5em}
\label{fig:tcp_rho}
\end{figure*}

\fi

\ifall

\begin{figure}[h!]
\centering
\includegraphics[width=0.96\columnwidth]{tcp_rho.pdf}
\vspaceSQ{-0.5em}
\caption{\textmd{(\cref{sec:eval-dctcp}) \textbf{Impact of $\rho$ on long-flow (1MiB) FCT with FatPaths on DCTCP} using
$n=4$ layers. The largest impact of non-minimal routing is for DF and SF, with
a 2$\times$ improvement in tail FCT; small improvements on tail FCT are seen in all
topologies, but there are no throughput improvements on networks with higher
minimal-path diversity.}}
\vspaceSQ{-2em}
\label{fig:tcp_rho}
\end{figure}

\fi

Besides FCT means/tails, we also consider a full completion time of a stencil
workload that is representative of an HPC application, in which processes
conduct local computation, communicate, and synchronize with a barrier; see
Figure~\ref{fig:tcp_topo_ttc}.
Results follow the same performance patterns as others. An interesting outcome
is JF: high values for LetFlow are caused by packet loss and do \emph{not}
affect the mean/99\% tail (cf.~Figure~\ref{fig:tcp_topo}), only the total
completion runtime.
{Overall, FatPaths ensures high speedups of completion times, e.g., more than
{2.5$\times$} and nearly {2$\times$} faster completion times on SF and XP, for
flows of the sizes of 200K and 2M bytes, respectively.}
  
\ifall\maciej{FIX}
As usual, mind the log scale, rho=0.6 LetFlow is often a factor of 2 faster.
Also LetFlow on a FatTree is amazing, sort-of as expected.  The

discussion is pretty much the same as for the 99\% results we had before: low
rho gives better tails at mean throughput cost, LetFlow can have a similar
effect, minimal routing (incl. rho=1 FatPaths) is doomed to suffer due to
path collisions.
\fi

{FatPaths also enables influencing communication latency: Specifically,
whenever lowest latency is prioritized, one can solely use a layer that
provides all shortest paths.  This ensures low latencies matching those
achieved with shortest-path routing in respective
networks~\mbox{\cite{besta2014slim}}.  For more throughput, one can use any
layer configuration offering diversity of almost-minimal paths.  Here, any
(marginal) latency overheads from the additional router-router hop are caused
by the properties of the underlying topology, \emph{not} the routing protocol.}


\begin{figure}[t]
\vspaceSQ{-1.25em}
\centering
\includegraphics[width=0.5\textwidth]{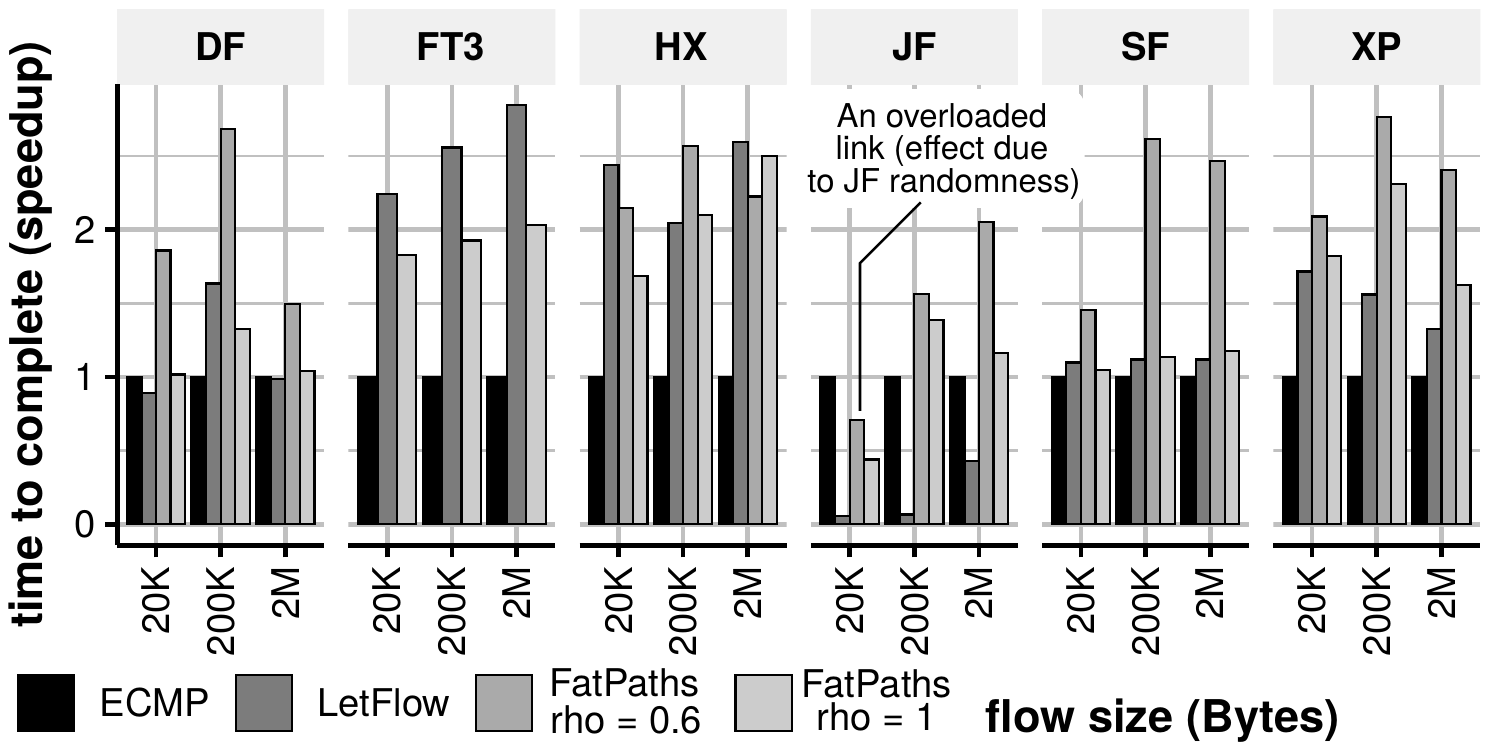}%
\vspaceSQ{-0.5em}
%
%
\caption{\textmd{FatPaths on TCP compared to 
ECMP and LetFlow (stencil + barrier).}}
\vspaceSQ{-1em}
\label{fig:tcp_topo_ttc}
\end{figure}

\subsection{Performance Analysis: Routing vs.~Topology}

How much performance gains in FatPaths come from its routing vs.~from simply the
benefits of low diameter~\cite{besta2014slim}?
Here, we extensively analyzed various design choices in FatPaths; full
description is in the extended report.  The takeaway is that simple past
routing schemes make low-diameter topologies worse ($\approx$2$\times$ and
more in FCT) than recent fat tree designs. This is because low diameter
\emph{must} be enhanced with effective tacking of flow conflicts and other
detrimental effects, which is addressed by multipathing in FatPaths.

\subsection{Performance Analysis: Impact from Partial Design Choices}

We also analyze speedups from \emph{specific parts of FatPaths}, e.g., only the
purified transport, flowlet load balancing, layered routing, or non-minimal
paths. While many of these elements can solely accelerate workloads in
low-diameter networks, \emph{it is the combination of effective non-minimal
multipath routing, load balancing, and transport that gives superior
performance}.
For example, Figure~\ref{fig:ndp_results_skewed} shows that fat trees with NDP
outperform low-diameter networks that do \emph{not} use multipathing based on
non-minimal paths (the ``NDP'' baseline).
\iftr
Next, different layer configurations ($\rho, n$) for various~$D$ are
investigated in Figure~\ref{fig:ndp_n_rho} and in~\cref{sec:eval-ndp} (bare
Ethernet systems), in Figure~\ref{fig:tcp_rho}, and
in~\cref{sec:clouds} (TCP systems). Differences (in FCT) across layer
configurations are up to 4$\times$; increasing both $n$ and $\rho$ maximizes
performance. Third, the comparison of adaptive load balancing (``LetFlow'')
based on flowlet switching vs.~static load balancing (``ECMP'') is in
Figure~\ref{fig:tcp_topo} and in~\cref{sec:clouds}; adaptivity improves tail
and short flow FCTs at the cost of long flow throughput. Fourth, in the
comparison of FatPaths with and without Purified Transport, performance with no
Purified Transport is always significantly worse. We also analyze performance
with and without layered routing (Figure~\ref{fig:tcp_topo}, ``ECMP'' and
``LetFlow'' use no layers at all); not using layers is detrimental for
performance on topologies that have no minimal-path diversity (e.g., SF or DF).
Next, we also study the impact of using \emph{only} the shortest paths in
FatPaths (Figure~\ref{fig:tcp_topo}, baseline~``$\rho=1$''); it is almost
always disadvantageous. Finally, the effect from workload randomization is
illustrated in Figures~\ref{fig:ndp_results_motiv} (randomization)
and~\ref{fig:ndp_results_skewed} (no randomization); randomization increases
throughput by $\approx$2$\times$.
\fi




\ifsq\enlargethispage{\baselineskip}\fi

\macb{\ul{Final Performance Takeaway}}
A high-performance routing architecture for low-diameter
networks should expose and use diversity of \emph{almost minimal} paths
(because they are numerous, unlike minimal paths).
\iftr
\emph{FatPaths enables this, giving speedups on both HPC systems such as
supercomputers or tightly coupled clusters, or cloud infrastructure such as
data centers}.
\fi

\iftr
\begin{figure*}[t]
\vspace{-0.5em}
\centering%
\includegraphics[width=0.68\textwidth]{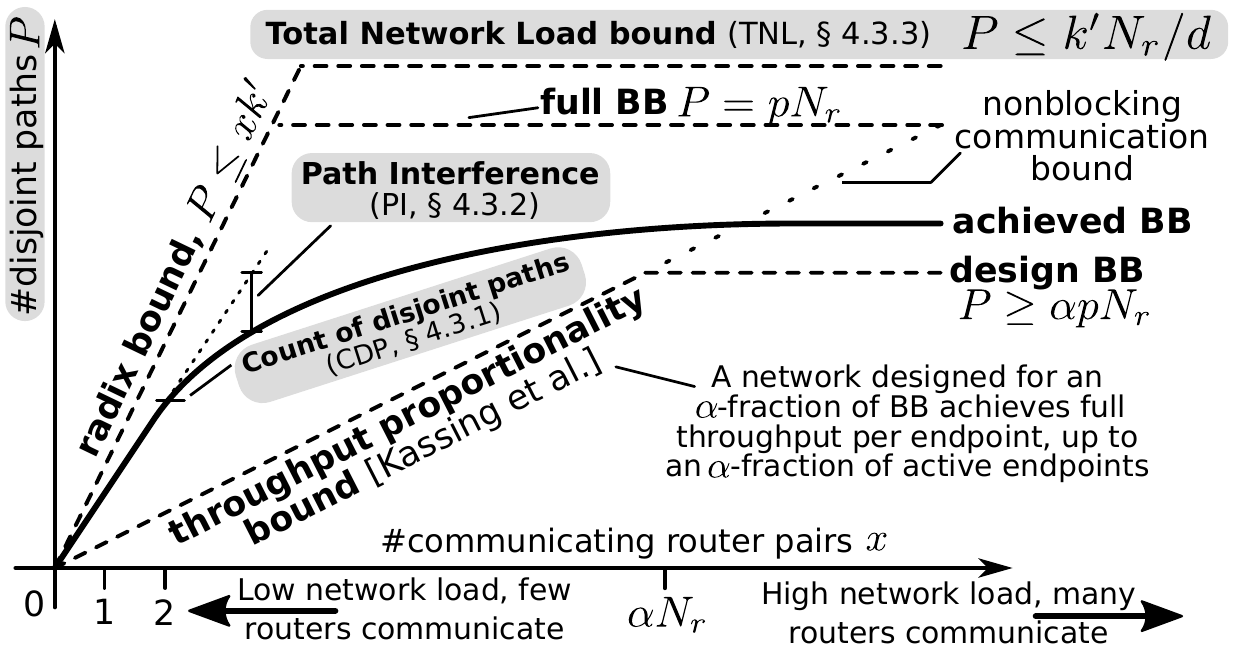}
\vspaceSQ{-0.5em}
\caption{\textmd{ 
Relations between connectivity- and BB-related measures.
Shade indicates metrics formalized and/or proposed as a part of FatPaths.
}}
\vspaceSQ{-1.5em}
\label{fig:paths-discussion}
\end{figure*}
\fi


\section{Discussion}
\label{sec:discussion}

\subsection{{Integration with Other Protocols for Wide Applicability}}

We also discuss integrating FatPaths with Data Center TCP (DCTCP)~\cite{alizadeh2011data},
RDMA~\cite{fompi-paper} (iWARP~\cite{iwarp},
RoCE~\cite{infiniband2014rocev2}), Infiniband~\cite{pfister2001introduction},
MPTCP~\cite{raiciu2011improving}, and non-minimal ECMP on FatPaths.
%
%

\iftr
\subsubsection{DCTCP}
\ 
We rely on DCTCP's packet loss to detect and avoid highly congested paths, but
also to detect failed links and routes missing in specific layers.  This
requires only minor changes to the TCP stack: On segment reception, the layer
tag in the address is masked out to not interfere with normal operation, and on
transmission, the respective bits are populated with the layer tag stored for
the flow. The layer tag is initialized to 1, to use a minimal path by default,
and randomly modified whenever the congestion window is reduced by DCTCP (which
leads to a flowlet boundary), or when an ICMP error is received. By responding
to ICMP errors, layered flowlet routing also provides basic fault tolerance.

\subsubsection{MPTCP}
\ 
We also provide a variant of FatPaths that uses MPTCP for congestion control,
as it already provides basic infrastructure and authentication mechanisms for
setting up multiple data streams. Our design uses ECN as a measure of
congestion instead of packet loss. If an incoming ACK packet does not have the
ECN field set, we increase the window analogously to the traditional TCP.
Otherwise, (every roundtrip time (RTT)) we update the congestion window size
(win size) accordingly.

\subsubsection{Standard TCP}
\
We also considered standard TCP.  As a ECN-enabled TCP variant, we extended the
standard TCP implementation with ECN echo according to the IETF DCTCP draft and
added the matching congestion control algorithm as a module derived from
TCPReno which can be used with standard TCP. Since we defined queue limits in
packets, we limit the congestion window size to integer multiples of the
maximum segment size, to make sure full size segments are transmitted. 
%

\subsubsection{ECMP on Non-Minimal Paths}
Our layered design could effectively enable \emph{ECMP on non-minimal paths}.
Specifically, minimal paths in a sparsified layer are effectively non-minimal
when considering all the layers (i.e., the whole network). Such layers could
then be used by ECMP for multipathing. 

\ifall
Specifically, to deploy ECMP on non-minimal paths, one would use polarization,
an ECMP-related issue~\cite{broadcom_smarthash, rezaaifar2008equivalent}.
Namely, for a given TCP flow, successive switches select identical output ports
for each packet, which may leave some links unused.  This effect is called
  \emph{polarization}.

Now, layered routing can even be implemented using ECMP hardware.
In ECMP, a switch selects the $i$-th forwarding rule out of $n$ equal-cost
options based on a pseudo-random hash value computed from the forwarding
parameters, including the pseudo-random TCP source port. In normal ECMP
deployments, these hash functions are randomized, to make sure the switches
select different combinations of $i$'s for varying source ports. Instead, we
can deliberately introduce \emph{polarization} by selecting \emph{the same hash
function for each switch}. Now, each forwarding tuple will get the same $i$
assigned over the entire network, and we can use $i$ as a layer tag, populating
each slot in the routing table with with a different forwarding function
$\sigma_i$. The paths in each slot might not be equal-cost, and the resulting
forwarding table not correct according to ECMP rules, but thanks to
polarization, the routing will exactly represent our layered routing and
therefore also be loop-free. A minor disadvantage of this approach is that we
can't systematically choose which layer $i$ a flow will use, since $i$ is still
determined by a hash function; however this is not a problem with the proposed
flowlet routing, which randomly chooses layers and will not be affected by
another randomization step.
\fi

\subsubsection{Integration with RDMA}
As FatPaths fully preserves the semantics of TCP, one could seamlessly
use iWARP~\cite{iwarp} on top of FatPaths.
%
%
FatPaths could also be used together with RoCE~\cite{infiniband2014rocev2}.
RoCE has traditionally relied on Ethernet with Priority Flow
Control~\cite{pfc_paper} (PFC) for lossless data transfer.  However, numerous
works illustrate that PFC introduces inherent issues such as head-of-line
blocking~\cite{mittal2018revisiting, le2018rogue, zhu2015congestion,
guo2016rdma}.  Now, the design of FatPaths reduces counts of dropped packets to
almost zero ($\le 0.01$\%) due to flowlet load balancing.  With its
packet-oriented design and a thin protocol layer over simple Ethernet,
\emph{FatPaths could become the basis for RoCE.}
Moreover, many modern RDMA schemes (e.g., work by Lu et al.~\cite{lu2018multi})
are similar to NDP in that they, e.g., also use packet spraying. Thus, many of our
results may be representative for such RDMA environments. For example, using RDMA
on top of FatPaths could provide similar advantages on low-diameter topologies
as presented in Figure~\ref{fig:ndp_results_motiv}
and~\ref{fig:ndp_results_skewed}. We leave this for future work.

\subsubsection{Enhancing Infiniband}
Although we focus on Ethernet, most of the schemes in FatPaths \emph{do not
assume anything Ethernet-specific} and they could \emph{straightforwardly
enhance IB routing architecture}.  For example, all the insights from
path diversity analysis, layered routing for multi-pathing, or flowlet load
balancing, \emph{could also be used with IB.} We leave these directions for
future work.
\fi

\iftr

\subsection{Integration with Performance Measures and Bounds}

For deeper understanding, we \emph{intuitively} connect our path diversity
measures to established network performance measures and bounds (e.g.,
bisection bandwidth (BB) or throughput
proportionality~\cite{kassing2017beyond}). Figure~\ref{fig:paths-discussion}
shows how various measures vary when increasing the network load expressed by
count of communicating router pairs~$x$. Values of measures are expressed with
numbers of disjoint paths~$P$.  In this expression, bandwidth measures are
counts of disjoint paths between two router sets; these numbers must match
corresponding counts in the original measure definitions (e.g., path count
associated with BB must equal the BB cut size).

\fi


\ifall

\subsection{FatPaths Limitations}
%
%
To facilitate applicability of our work in real-world installations, we discuss
FatPaths' limitations. First, as FatPaths addresses low-diameter topologies, it
is less advantageous on high-diameter older interconnects such as torus.  This
is mostly because such networks provide multiple (almost or completely
disjoint) shortest paths between most router pairs.
Second, FatPaths inherits some of NDP's limitations, namely interrupt
throttling. Similarly to NDP, we alleviate this by dedicating a single CPU core
to polling for incoming packets.
Finally, even if FatPaths delivers decent performance for non-randomized
workloads (as shown in~\cref{sec:eval-ndp} and in
Figure~\ref{fig:ndp_results_skewed}), it ensures much higher performance with
workload randomization. Yet, as discussed in~\cref{sec:overview}, this is (1) a
standard technique in HPC systems and (2) it is not detrimental for application
performance on low-diameter networks that -- by design -- have very low
latencies for all router pairs.

\fi

\subsection{{FatPaths Limitations}}
%
%
To facilitate applicability of our work in real-world installations, we discuss
its limitations. First, FatPaths addresses low-diameter topologies,
being less beneficial on high-diameter older interconnects such as torus,
because such networks provide multiple (almost or completely disjoint)
shortest paths between most router pairs.
FatPaths also inherits some NDP's limitations, namely interrupt
throttling. As in NDP, this is fixed by dedicating one CPU core
to polling for incoming packets.

\iftr

\subsubsection{Deriving Path Diversity Efficiently}
We developed algorithms to efficiently compute diversity measures, basing on
schemes by Cheung~et.~al~\cite{cheung2013},
Ford-Fulkerson~\cite{ford1956maximal}, Gomory-Hu~\cite{panigrahi2008gomory},
Gusfield~\cite{gusfield1990}.  As these algorithms are outside the core thread
of this work, they are detailed in the Appendix. 
\fi

\ifall
\maciej{not sure if we need it. For now, just a list of stuff below.}

\begin{itemize}
\item 
There is a tradeoff between path shortness and path diversity; we can use a topology with long shortest paths or non-minimal paths
to get diversity. 
\item 
To save resources, we need to reduce the demand for path diversity
(random mapping). Then, we get away with shorter paths than Clos.
\item
In low diameter topologies, basically all $d+1$ length paths are
disjoint (for $l=2$ paths can only be equal or disjoint and $l=3$ is not too
different).
\item
Anything that increases path length a bit is good, as long as paths
do not get excessively long (spanning trees, Valiant).
\item
Layers are a simple way to handle non-minimal paths.
\item
Things are highly asymmetrical, therefore we need adaptive load balancing
\item
Flowlets are simple and effective, and can be applied on top of
almost anything.
\item
Once we have adaptivity, fault tolerance is easy.
\item
To prevent evil tail latencies, we need sane flow control -- almost anything
that is not TCP will do.
\marcel{This is where the big gains can be made, also with flow scheduling,
but IMO not our problem. Though it might not hurt to mention it more.}
\end{itemize}

\fi

\ifsq
\enlargethispage{\baselineskip}
\fi

\section{Related Work}

FatPaths touches on various areas.  We now briefly discuss related works,
excluding the ones covered in past sections.
%

%
Our work targets modern \textbf{low-diameter topologies} such as Slim
Fly~\cite{besta2014slim, singla2012jellyfish, valadarsky2015, ahn2009hyperx, kim2008technology}.
%
%
\emph{FatPaths enables these networks to
achieve low latency and high throughput with various workloads, outperforming
similar-cost fat trees.}

%
We analyze \textbf{routing} schemes in Table~\ref{tab:intro} and
in Section~\ref{sec:theory}. \emph{FatPaths is the first to offer generic and adaptive
multi-pathing using both shortest and non-shortest disjoint paths.}
\iftr
Here, we discuss only layered designs similar to FatPaths.
%
%
Layered routing exists is various schemes. SPAIN~\cite{mudigonda2010spain} is
similar to FatPaths but limits layers to trees. 
This brings drawbacks, see Section~\ref{sec:theory}. 
Other protocols are similar to SPAIN, for example PAST~\cite{stephens2012past}
that uses one spanning tree per host.  However, PAST does not enable multiple
paths between two servers, because it uses one spanning tree per
destination~\cite{hu2016explicit}.
%
%


Some schemes complement FatPaths in their focus. For example,
XPath~\cite{hu2016explicit} and source routing~\cite{jyothi2015towards} could
be used together with FatPaths, for example, by providing effective means to
encode the rich path diversity exposed by FatPaths.
\fi

%
Adaptive \textbf{load balancing} can be implemented using
flows~\cite{curtis2011mahout, rasley2014planck, sen2013localflow,
tso2013longer, benson2011microte, zhou2014wcmp, al2010hedera,
kabbani2014flowbender, hopps2000analysis}, flowcells (fixed-sized packet
series)~\cite{he2015presto}, and packets~\cite{zats2012detail, handley2017re,
dixit2013impact, cao2013per, perry2015fastpass, zats2012detail,
raiciu2011improving, ghorbani2017drill}. We choose an intermediate level,
flowlets (variable-size packet series)~\cite{katta2016clove, alizadeh2014conga,
vanini2017letflow, katta2016hula, kandula2007dynamic}. 
\iftr
In contrast, HPC networks often use packet level adaptivity,
focusing on choosing good congestion signals, often with hardware
modifications~\cite{garcia2012ofar,garcia2013efficient}. We use flowlet
elasticity as introduced by Vanini et~al.~\cite{vanini2017letflow} (which does
not need any hardware support) as an adaptivity component to use non-minimal
paths efficiently. 
\fi
FatPaths is the first architecture to use load balancing based on flowlets
\emph{for low-diameter networks}.

%
We do not compete with \textbf{congestion or flow control} schemes; we use them
for more performance. \emph{FatPaths can use any such scheme in its
design}~\cite{mittal2015timely, cardwell2016bbr, alizadeh2011data, he2016acdc,
zhuo2016rackcc, handley2017re, raiciu2011improving, bai2014pias,
alizadeh2013pfabric, vamanan2012deadline, lu2016sed, hwang2014deadline,
montazeri2018homa, jiang2008explicit, banavalikar2016credit,
alasmar2018polyraptor}.
%
%
\iftr
In many HPC networks, flow control is implemented in hardware, since the
network is lossless. A backpressure mechanism propagates congestion
information through each link up to the endpoints. In contrast, Ethernet-based
networks are lossy and flow as well as congestion control needs to be
implemented end-to-end, in software, or with schemes such as
PFC~\cite{ieee802.1bb, pfc_paper}. Besides the classical TCP Reno and Vegas
algorithms, there are more recent schemes such as
TIMELY~\cite{mittal2015timely}, BRR~\cite{cardwell2016bbr} and
DCTCP~\cite{alizadeh2011data}, which explicitly require and use ECN.
AC/DC~\cite{he2016acdc} move congestion control to the hypervisor, and
RackCC~\cite{zhuo2016rackcc} even moved it into the network hardware itself.
Recently, Handley et al.~\cite{handley2017re} redesigned end-to-end
flow control in Clos topologies entirely, with the NDP protocol (that we
incorporate for low-diameter networks in FatPaths). They propose simple changes
to the network hardware to support effective per-endpoint flow control, which
resolves many well-known issues with TCP, such as the \emph{incast problem}.
Another direction is MPTCP~\cite{raiciu2011improving}, which, among other
aspects, allows moving network adaptivity into the endpoint software.  More
sophisticated flow scheduling was proposed by PIAS~\cite{bai2014pias} or
pFabric~\cite{alizadeh2013pfabric}. Such methods could be combined with our
routing scheme to improve real-world application performance.
\fi


%
Many works on \textbf{multi-pathing} exist~\cite{zhou2014wcmp, benet2018mp, cao2013per,
greenberg2009vl2, sen2013localflow, caesar2010dynamic, kassing2017beyond,
mudigonda2010spain, perry2015fastpass,
aggarwal2016performance, suurballe1984quick, huang2009performance,
sohn2006congestion, li2013openflow, bredel2014flow, van2011revisiting,
benet2018mp, suchara2011network}.
Our work differs from them all: \emph{it focuses on path diversity in
low-diameter topologies and it uses both minimal and non-minimal paths}.

%
Some works analyze various \textbf{properties of low-diameter topologies}, for
example path length, throughput, and bandwidth~\cite{valadarsky2015,
kathareios2015cost, jyothi2016measuring, singla2012jellyfish,
kassing2017beyond, li2018exascale, besta2018slim, kawano2018k,
harsh2018expander, kawano2016loren, truong2016layout, flajslik2018megafly,
kawano2017layout, azizi2016hhs, truong2016distributed, al2017new}.
\emph{FatPaths offers the most extensive analysis on path diversity so far}.

%
Some schemes complement FatPaths. For example, XPath~\cite{hu2016explicit} and
source routing~\cite{jyothi2015towards} deliver means to \textbf{encode
different paths}. They could be used \emph{together with FatPaths} by encoding
the rich path diversity \emph{exposed by FatPaths}.
Finally, FatPaths could be used to accelerate communication-efficient workloads
that benefit from low-diameter properties of Slim Fly and other modern
topologies, including deep learning~\cite{ben2019modular,
gronquist2019predicting, li2020taming, oyama2018accelerating, ben2018neural,
ben2019demystifying}, linear algebra computations~\cite{besta2020communication,
besta2017slimsell, solomonik2017scaling, kwasniewski2019red}, graph
processing~\cite{besta2015accelerating, besta2019substream, besta2019slim,
besta2019graph, besta2018slim, besta2017push, besta2018log, besta2018survey,
gianinazzi2018communication, besta2019practice, besta2019demystifying}, and
other distributed workloads~\cite{besta2015active, besta2014fault, fompi-paper,
ziogas2019data, ben2019stateful, di2019network} and
algorithms~\cite{schweizer2015evaluating, schmid2016high, sutton2018optimizing,
tate2014programming}. One could possibly use some elements of the FatPaths
routing for the associated problems in the on-chip networking~\cite{de2018transformations, besta2018slim}.



\section{Conclusion}

We introduce \textbf{\emph{FatPaths: a simple, high-performance, and robust
{routing architecture} for a modern family of low-diameter topologies}}.
FatPaths enables such networks to achieve unprecedented performance by exposing
the rich (``fat'') diversity of minimal \emph{and non-minimal} paths.
We formalize and extensively analyze this path diversity and show that, even
though the considered topologies \emph{fall short of shortest paths}, they can
accommodate three ``almost'' minimal disjoint paths, which is enough to avoid
congestion in many traffic scenarios.
%
%
Our path diversity metrics and methodology can be used to analyze other
properties of networks.


The key part of FatPaths, layered routing, enables harnessing diversity of both
shortest and non-minimal paths.  Supported with simple yet effective flowlet
load balancing, and high-performance transport in TCP settings, FatPaths
achieves low-latency and high-bandwidth, outperforming very recent fat tree
architectures~\cite{handley2017re} by 15\% in net throughput at 2$\times$ in
latency, for comparable cost.
Even though we focus on Ethernet, most of these schemes -- for example adaptive
flowlet load balancing and layers -- are generic and they could enhance
technologies such as RDMA (RoCE, iWARP) and Infiniband.

We deliver simulations with up to  \emph{one million} endpoints.  Our code is
online and can be used to foster novel research on next-generation large-scale
compute centers.

FatPaths uses Ethernet for maximum versatility.  We argue that it can
accelerate both HPC clusters or supercomputers as well as data centers and
other types of cloud infrastructure.  FatPaths will help to bring the areas of
HPC networks and cloud computing closer, fostering technology transfer and
facilitating exchange of ideas.

\ifsq\enlargethispage{\baselineskip}\fi

\vspace{1em}

\ifsq

\def\baselinestretch{0.9}\selectfont
\scriptsize 

\fi

{\small
\macb{Acknowledgements: }
We thank Mark Klein, Hussein Harake, Colin McMurtrie, Angelo Mangili, and the
whole CSCS team granting access to the Ault and Daint machines, and for their
excellent technical support. We thank Timo Schneider for his immense help with
computing infrastructure at SPCL. We thank Nils Blach, Alessandro Maisen, and Adam Latos for
useful comments that helped improve the quality of the paper. 
}

\ifsq

\normalsize
\def\baselinestretch{1.0}\selectfont

\fi

%% file: appendix.tex
\section*{APPENDIX}

We now provide full discussions, analyses, and results omitted in the
main paper body to maintain its clarity.




\captionsetup[table]{font={normalsize,sf},labelfont={normalsize,sf}}

\begin{table*}

\setlength{\tabcolsep}{1.7pt}
\centering
\scriptsize
\sffamily
\begin{tabular}{l|l|llllllccll}
\toprule
& \textbf{Topology} & \textbf{Hierarchy} & \textbf{Flexibility} & \textbf{Input} & $N_r$ & $N$ & $k'$ & $p$ & $D$ & \textbf{Remarks} & \textbf{Deployed?} \\
\midrule
\multirow{8}{*}{\begin{turn}{90}deterministic\end{turn}} 
& Slim Fly~\cite{besta2014slim} & group hierarchical  &  fixed  & $q$               & $2 q^2$       & $pN_r$ & & $\left\lceil\frac{k'}{2}\right\rceil$ & 2 & ``MMS'' variant~\cite{besta2014slim, mckay1998note} & unknown \\
& Dragonfly~\cite{kim2008technology} & group hierarchical & fixed  & $p$               & $4p^3+2p$  & $pN_r$ & $3p-1$   & $p$ & 3 & ``balanced'' variant~\cite{kim2008technology} (§3.1) & \makecell[l]{PERCS~\cite{arimilli2010percs},\\Cascade~\cite{faanes2012cray}} \\
& HyperX~\cite{ahn2009hyperx} & semi-hierarchical & fixed  & $S$               & $S^2$         & $pN_r$ & $2(S-1)$ & $\left\lceil\frac{k'}2\right\rceil$ & 2 & \makecell[l]{``regular'' variant, 2x-oversubscribed,\\forms a Flattened Butterfly~\cite{kim2007flattened}} & unknown \\
& HyperX~\cite{ahn2009hyperx} & semi-hierarchical & fixed  & $S$               & $S^3$         & $pN_r$ & $3(S-1)$ & $\left\lceil\frac{k'}3\right\rceil$ & 3 & \makecell[l]{``regular'' variant, 2x-oversubscribed,\\forms a cube} & unknown \\
& Fat tree~\cite{leiserson1996cm5} & semi-hierarchical & fixed & $k$       & $5\left\lfloor\frac{k^2}4\right\rfloor$& $p\left\lfloor\frac{k^2}2\right\rfloor$     & $\frac{k}2$      & $\left\lceil\frac{k}2\right\rceil$  & 4 & 2-stage variant (3 router layers) & Many installations \\
& Complete (clique) & flat & fixed & $k'$         & $k'+1$        & $pN_r$ & $k'$     & $k'$ & 1 & $D=1$ HyperX, 2x-oversubscribed & crossbar routers \\
\midrule
\multirow{2}{*}{\begin{turn}{90}rand.\end{turn}} 
& Jellyfish~\cite{singla2012jellyfish} & flat & flexible  & $k'$, $N_r$, $p$  & $N_r$         & $pN_r$ & $k'$     & $p$ & n/a & ``homogeneous'' variant~\cite{singla2012jellyfish} & unknown \\
& Xpander~\cite{valadarsky2015} & flat & semi-flexible  & $\ell$ & $\ell(k'+1)$  & $pN_r$ & $\ell$     & $\left\lceil\frac{k'}2\right\rceil$ & n/a & Restricted to $\ell=k'$, $D\approx 2$, $p=\left\lceil\frac{k'}2\right\rceil$ & unknown \\
\bottomrule
\end{tabular}
%
\caption{\textmd{The considered topologies. ``Input'' means input parameters
used to derive other network parameters.}
%
%
}
\vspaceSQ{-1em}
\label{tab:parameters-app}
\end{table*}

\section{Formal Description of Topologies}
\label{sec:app-topos-details}

We first extend the discussion of the considered topologies.
Table~\ref{tab:parameters-app} provides details.
Some selected networks are \emph{flexible} (parameters determining their
structure can have arbitrary values) while most are \emph{fixed} (parameters
must follow well-defined closed-form expressions).
Next, networks can be \emph{group hierarchical} (routers form \emph{groups}
connected with the same pattern of intra-group \emph{local} cables and then
groups are connected with \emph{global} inter-group links),
\emph{semi-hierarchical} (there is some structure but no such groups), or
\emph{flat} (no distinctive hierarchical structure at all).
Finally, topologies can be \emph{random} (based on randomized constructions) or
\emph{deterministic}.

The motivation for picking such networks is as follows.
{Slim Fly} (SF)~\cite{besta2014slim} is a state-of-the-art cost-effective
topology that 
optimizes its structure towards the Moore Bound~\cite{hoffman2003moore}.
It represents the recent family of diameter-2 networks.
{HyperX} (Hamming graph) (HX)~\cite{ahn2009hyperx} generalizes, among others,
{hypercubes} (HCs)~\cite{bondy1976graph} and {Flattened Butterflies}
(FBF)~\cite{kim2007flattened}.
We also consider {Dragonfly} (DF)~\cite{kim2008technology}, an established
hierarchical network.
%
%
{Jellyfish} (JF)~\cite{singla2012jellyfish} is a random regular graph with good
expansion properties~\cite{bondy1976graph}. 
{Xpander} (XP)~\cite{valadarsky2015} resembles JF but has a deterministic
construction variant. 
{Fat tree} (FT)~\cite{leiserson1996cm5}, a widely used interconnect,
is similar to the {Clos network}~\cite{clos1953study} with
disjoint inputs and outputs and unidirectional links. 
FT stands for designs that are widely used and feature excellent performance
properties such as full bisection bandwidth and non-blocking routing. 
%
%
We use three-stage FTs (FT3); fewer stages reduce scalability while more
stages lead to high diameters.
Finally, we also consider {fully-connected} graphs. They offer lower bounds on
various metrics such as minimal path length, and can be used for validation.

Now, each topology uses certain \emph{input parameters} that define the
structure of this topology. These parameters are as follows: $q$ (SF), $a,h$
(DF), $\ell$ (XP), and $L,S,K$ (HX).


\subsection{Slim Fly}
 
Slim Fly~\cite{besta2014slim} is a modern cost-effective
topology for large computing centers that uses mathematical optimization to
minimize diameter $D$ for a given radix $k$ while maximizing size $N$.  SF's
low diameter ($D = 2$) ensures the lowest latency for many traffic patterns and
it reduces the number of required network resources (packets traverse fewer
routers and cables), lowering cost, static, and dynamic power consumption. 
SF is based on graphs approaching the Moore Bound (MB): The upper bound on the
number of vertices in a graph with a given $D$ and $k'$. This ensures full
global bandwidth and high resilience to link failures due to good expansion
properties.
Next, SF is group hierarchical. A group is not necessarily complete but all the
groups are connected to one another (with the same number of global links) and
form a complete network of groups.
We select SF as it is a state-of-the-art design that
outperforms virtually all other targets in most metrics and represents
topologies with $D = 2$.

\noindent
\macb{Associated Parameters}
$N_r$ and $k'$ depend on a parameter $q$ that is a prime power with certain
properties (detailed in the original work~\cite{besta2014slim}). Some
flexibility is ensured by allowing changes to $p$ and with a large number of
suitable values of the parameter~$q$. We use the suggested value of $p = \left\lceil {k'}/{2}
\right\rceil$.

\subsection{Dragonfly}
 
Dragonfly~\cite{kim2008technology} is a group hierarchical network with $D
= 3$ and a layout that reduces the number of global wires.  Routers form
complete {groups}; groups are connected to one another to form a complete
network of groups with one link between any two groups.
DF comes with an intuitive design and represents deployed networks with $D =
3$.

\noindent
\macb{Associated Parameters}
Input is: the group size $a$, the number of channels from one
router to routers in other groups $h$, and concentration $p$. We use the
\emph{maximum capacity} DF (with the number of groups $g = ah+1$) that is
\emph{balanced}, i.e., the load on global links is balanced to avoid
bottlenecks ($a = 2p = 2h$).
In such a DF, a single parameter $p$ determines all others.

\subsection{Jellyfish}
 
Jellyfish~\cite{singla2012jellyfish} networks are random regular graphs
constructed by a simple greedy algorithm that adds randomly selected edges
until no additions as possible. The resulting construction has good expansion
properties~\cite{bondy1976graph}. Yet, all guarantees are probabilistic and
rare degenerate cases, although unlikely, do exist. Even if $D$ can be
arbitrarily high in degenerate cases, usually $D < 4$ with much lower $d$.
We select JF as it represents flexible topologies that use randomization and
offer very good performance properties.

\noindent
\macb{Associated Parameters}
JF is flexible. $N_r$ and $k'$ can be arbitrary; we use parameters matching
less flexible  topologies.
%
%
%
To compensate for the different amounts of hardware used in different
topologies, we include a Jellyfish network constructed from the same routers
for each topology; the performance differences observed between those networks
are due to the different hardware and need to be factored in when comparing
the deterministic topologies. 

\subsection{Xpander}

Xpander~\cite{valadarsky2015} networks resemble JF but have a
deterministic variant.  They are constructed by applying one or more so called
$\ell$-\emph{lifts} to a $k'$-clique  $G$.
The $\ell$-lift of $G$ consists of $\ell$ copies of $G$, where for each edge
$e$ in $G$, the copies of $e$ that connect vertices $s_1, \dots, s_\ell$ to
$t_1, \dots, t_\ell$, are replaced with a \emph{random matching} (can be
derandomized): $s_i$ is connected to $t_{\pi(i)}$ for a random
$\ell$-permutation $\pi$.
This construction yields a $k'$-regular graph with $N = \ell k'$ and good
expansion properties. The randomized $\ell$-lifts ensure good properties in the
expectation.
%
%
%
We select XP as it offers the advantages of JF in a deterministically
constructed topology.

\noindent
\macb{Associated Parameters}
We create XP with a single lift of arbitrary $\ell$. Such XP is
flexible although there are more constraints than in JF.  Thus, we cannot
create matching instances for each topology. We select $k' \in \{16, 32\}$ and
$\ell = k'$, which is comparable to diameter-2 topologies. We also consider
$\ell = 2$ with multiple lifts as this ensures good
properties~\cite{valadarsky2015}, but we do not notice any additional speedup.
We use $p = \frac{k'}2$, matching the diameter-2 topologies. 

\subsection{HyperX}
 
HyperX~\cite{ahn2009hyperx} is formed by arranging vertices in an
$L$-dimensional array and forming a clique along each 1-dimensional row.
Several topologies are special cases of HX, including complete graphs,
hypercubes (HCs)~\cite{bondy1976graph}, and Flattened Butterflies
(FBF)~\cite{kim2007flattened}. HX is a generic design that represents a wide
range of networks.

\noindent
\macb{Associated Parameters}
An HX is defined by a 4-tuple $(L, S, K, p)$. $L$ is the number of dimensions
and $D=L$, $S$ and $K$ are $L$-dimensional vectors (they respectively denote
the array size in each dimension and the relative capacity of links along each
dimension). Networks with uniform $K$ and $S$ (for all dimensions) are called
\emph{regular}. We only use regular $(L, S, 1, \cdot)$ networks with $L \in
\{2,3\}$.
HX with $L=2$ is about a factor of two away from the MB ($k' \approx 2
\sqrt{N_r}$) resulting in more edges than other topologies. Thus, we include
higher-diameter variants with $k'$ similar to that of other networks.
Now, for full bisection bandwidth (BB), one should set $p = \frac{k'}{2D}$.
Yet, since HX already has the highest $k'$ and $N_r$ (for a fixed $N$) among
the considered topologies, we use a higher $p = \frac{k'}D$ as with the other
topologies to reduce the amount of used hardware. As we do not consider
worst-case bisections, we still expect HX to perform well.

\subsection{Fat Tree}
 
%
Fat tree~\cite{leiserson1996cm5} is based on the Clos
network~\cite{clos1953study} with disjoint inputs and outputs and
unidirectional links. By ``folding'' inputs with outputs, a multistage fat tree
that connects any two ports with bidirectional links is constructed. We use
three-stage FTs with $D = 4$; fewer stages reduce scalability while more stages
lead to high $D$.
%
%
FT represents designs that are in widespread use and feature excellent
performance properties such as full BB and non-blocking routing.

\noindent
\macb{Associated Parameters}
A three-stage FT with full BB can be constructed from routers with uniform
radix $k$: It connects ${k^3}/4$ endpoints using five groups of ${k^2}/4$
routers. Two of these groups, ${k^2}/2$ routers, form an \emph{edge group}
with ${k}/2$ endpoints.  Another two groups form an \emph{aggregation layer}:
each of the edge groups forms a complete bipartite graph with one of the
aggregation groups using the remaining ${k}/2$ ports, which are called
\emph{upstream}. Finally, the remaining group is called the \emph{core}: each
of the two aggregation groups forms a fully connected bipartite graph with the
core, again using the remaining ${k}/2$ upstream ports. This also uses all $k$
ports of the core routers.
Now, for FT, it is not always possible to construct a matching JF as $N/N_r$
can be fractional. In this case, we select $p$ and $k'$ such that $k = p+k'$
and ${k'}/p \approx 4$, which potentially changes $N$. Note also that for FT,
$p$ is the number of endpoints per edge router, while in the other topologies,
all routers are edge routers.

\subsection{Fully-Connected Graphs}
 
We also consider fully-connected graphs. They represent interesting corner-cases, offer lower
bounds on various metrics such as minimal path length, and can be used for
validation. 

\noindent
\macb{Associated Parameters}
A clique is defined by a single parameter~$k'$, leading to $N_r = k'+1$. We use $p=k'$ with the
same rationale as for the HyperX topologies.

\section{Efficient Path Counting}
\label{sec:app-path-counting}

Some measures for path diversity are computationally hard to derive for
large graphs. Algorithms for all-pairs shortest paths analysis based on
adjacency matrices are well known, and we reintroduce one such method here
for the purpose of reproducibility. For
the disjoint-paths analysis however, all-pairs algorithms exist, but are not
commonly known. We introduce a method by Cheung~et.~al~\cite{cheung2013} and
\emph{we adapt for length-limited edge connectivity computation}.

\subsection{Matrix Multiplication for Path Counting}
\label{sec:app-mmm}

It is well known that for a graph represented as an adjacency matrix, matrix
multiplication (MM) can be used to obtain information about paths in that graph.
Variations of this include the Floyd-Warshall
algorithm~\cite{floyd1962algorithm} for transitive closure and all-pairs
shortest paths~\cite{seidel1995all}, which use different semirings to aggregate
the respective quantities.
To recapitulate how these algorithms work, consider standard
MM using $\cdot$ and $+$ operators on non-negative integers, which
computes the number of paths $n_i(s,t)$ between each pair of vertices.

\begin{theorem}
If $A$ is the adjacency matrix of a directed graph $G = (V, E)$, $A_{i,j} = 1$
iff $(i,j) \in E$ and $ A_{i,j} = 0$ iff $(i,j) \not\in E$, then each cell $i \in V,j \in V$ of $Q = A^l =
\underbrace{A\cdot \ldots \cdot A}_{l\ \text{times}}$ contains the number of
paths from $i$ to $j$ with exactly $l$ steps in $G$.
\end{theorem}

\begin{proof}
By induction on the path length $l$: For $l=1$, $A^l = A$ and the adjacency
matrix contains a $1$ in cell $i, j$ iff $(i,j) \in E$, else $0$. Since
length-1 paths consist of exactly one edge, this satisfies the theorem.
Now consider matrices $A^p$, $A^q$ for $p+q = l$ for which the theorem holds
since $p,q < l$. We now prove the theorem also holds for $A^l = A^p \cdot A^q$.
Matrix multiplication is defined as 

\begin{equation} (A^p \cdot A^q)_{i,j} =
\sum_k A^p_{i,k} \cdot A^q_{k,j}\,.
\end{equation}

\noindent
According to the theorem, $A^p_{i,k}$ is the number of length-$p$ paths from
$i$ to some vertex $k$, and $A^q_{k,j}$ is the number of length-$q$ paths from said
vertex $k$ to $j$. To reach $j$ from $i$ via $k$, we can choose any path from $i$
to $k$ and any from $k$ to $j$, giving $ A^p_{i,k} \cdot A^q_{k,j}$ options.
As we regard \emph{all} paths from $i$ to $j$, we consider
\emph{all} intermediate vertices $k$ and count the total number (sum) of paths.
This is exactly the count of length-$l$ paths demanded by the theorem, as 
each length-$l$ path can be uniquely split into a length-$p$ and a length-$q$
segment. 
\end{proof}

In the proof we ignored a few details caused by the adjacency matrix
representation: first, the adjacency matrix models a directed graph. We can
also use the representation for undirected graphs by making sure $A$ is
symmetrical (then also $A^l$ is symmetrical). Adjacency matrices
contain the entry $A_{i,j} = 0$ to indicate $(i,j) \notin E$ and
$A_{i,j} = 1$ for $(i,j) \in E$. By generalizing $A_{i,j}$ to be the number of
length-1 paths ($=$ number of edges) from $i$ to $j$ as in the theorem, we can
also represent multi-edges; the proof still holds.

Finally, the diagonal entries $A_{i,i}$ represent self-loops in the graph,
which need to be explicitly modeled. Note that also $i = j$ is allowed above
and the intermediate vertex $k$ can be equal to $i$ and/or $j$. Usually
self-loops should be avoided by setting $A_{i,i} = 0$. Then $A^l_{i,i}$ will be
the number of cycles of length $l$ passing through $i$, and the paths counted
in $A_{i,j}$ will include paths containing cycles. These cannot easily be
avoided in this scheme\footnote{Setting $A^l_{i,i} = 0$ before/after each step
does not prevent cycles, since a path from  $i$ to $k$ might pass $j$, causing
a cycle, and we cannot tell this is the case without actually recording the
path.}. For most measures, e.g., shortest paths or disjoint paths, this is not
a problem, since paths containing cycles will naturally never affect these
metrics.
%

On general graphs, the algorithms outlined here are not attractive since it
might take up to the maximum shortest path length $D$ iterations to reach a
fixed point, however since we are interested in low-diameter graphs, they are
practical and easier to reason about than the Floyd-Warshall aorithms.

\subsubsection{Matrix Multiplication for Routing Tables}

As another example, we will later use a variation of this algorithm to compute
next-hop tables that encode for each source $s$ and each destination $t$ which
out-edge of $s$ should be used to reach $t$. In this algorithm, the matrix
entries are sets of possible next hops. The initial adjacency matrix will
contain for each edge in $G$ a set with the out edge index of this edge,
otherwise empty sets. Instead of summing up path counts, we union the next-hop
sets, and instead of multiplying with zero or one for each additional step,
depending if there is an edge, we retain the set only if there is an edge for
the next step. Since this procedure is not associative, it cannot be used to
form longer paths from shorter segments, but it works as long as we always use
the original adjacency matrix on the right side of the multiplication. The
correctness proof is analogous to the path counting procedure.

\subsection{Counting Disjoint Paths}

The problem of counting all-pairs disjoint paths per pair is equivalent to the
all-pairs edge connectivity problem which is a special case of the all-pairs
max flow problem for uniform edge capacities. It can be solved using a spanning
tree (\emph{Gomory-Hu tree}~\cite{panigrahi2008gomory}) with minimum $s-t$-cut values for the respective
partitions on the edges. The minimum $s-t$ cut for each pair is then the
minimum edge weight on the path in this tree, which can be computed cheaply for
all pairs. The construction of the tree requires $\BigO(N_r)$ $s-t$-cuts, which
cost $\BigO(N_r^3)$ each (e.g., using the Push-Relabel scheme~\cite{cherkassky1997implementing}). 

Since we are more interested in the max flow values, rather than the min-cut
partitions, a simplified approach can be used: while the Gomory-Hu tree has max
flow values and min cut partitions equivalent to the original graph, a
\emph{equivalent flow tree}~\cite{gusfield1990} only preserves the max flow
values. While constructing it needs the same number of max-flow computations,
these can be performed on the original input graph rather than the contracted
graphs of Gomory-Hu, which makes the implementation much easier.

For length-restricted connectivity, common max-flow algorithms have to be
adapted to respect the path length constraint. The Gomory-Hu approach does not
work, since it is based on the principle that the distances in the original
graph do not need to be respected. We implemented an algorithm based on the
Ford-Fulkerson method~\cite{ford1956maximal}, using 
BFS~\cite{cormen2009introduction}, which is not suitable for an all-pairs
analysis, but can provide results for small sets of samples.

The spanning-tree based approaches only work for undirected graphs, and solve
the more general max-flow problem. There are also algorithms that only solve
the edge-connectivity problem, using completely different approaches. Cheung
et.~al~\cite{cheung2013} propose an algorithm based on linear algebra which
can compute all-pairs connectivity in $\BigO(|E|^\omega + |V|^2 k'^\omega)$;
$\omega \leq 3$ is the exponent for matrix-matrix multiplication. For our
case of $k' \approx \sqrt{N_r}$ and naive matrix inversion, this is
$\BigO(N_r^{4.5})$ with massive space use, but there are many options to use
sparse representations and iterative solvers, which might enable 
$\BigO(N_r^{3.5})$. Due to their construction, those algorithms also allow a
limitation of maximum path length (with a corresponding complexity reduction)
and the heavy computations are built on well-known primitives with low constant
overhead and good parallel scaling, compared to classical graph schemes.

\subsection{Deriving Edge Connectivity}

This scheme is based on the ideas of Cheung et.~al.~\cite{cheung2013}. First, we
adapt the algorithm for vertex connectivity, which allows lower space- and time
complexity than the original algorithm and simplifies its design.
The original edge-connectivity algorithm is obtained by applying it to a
transformed graph.\footnote{Vertex-connectivity, defined as the minimum size of
a cut set $c_{st} \subset V\setminus\{s,t\}$ of vertices that have to be removed to
make $s$ and $t$ disconnected, is not well defined for neighbors
in the graph. The edge-connectivity algorithm avoids this problem, but this
cannot be generalized for vertex-connectivity.
%
%
} 
We then introduce the path-length constraint by replacing the exact solution
obtained by matrix inversion with an approximated one based on iterations,
which correspond to incrementally adding steps.
The algorithm is randomized in the same way as the original is; we ignore
the probability analysis for now, as the randomization is only required to
avoid degenerate matrices in the process and allow the use of a finite domain.
The domain $\mathbb{F}$ is defined to be a finite field of sufficient size to
make the analysis work and allow a real-world implementation; we can assume
$\mathbb{F} = \mathbb{R^+}$.

First, we consider a \emph{connection matrix}, which is just the adjacency matrix with
random coefficients for the edges:

\begin{equation}
K_{i,j} = \begin{cases}
    x \in \mathbb{F}\ \text{u.a.r.} & \text{iff}\ (i,j) \in E \\
    0 & \text{else}\,.
\end{cases}
\end{equation}

In the edge-connectivity algorithm we use a much larger adjacency matrix of a
transformed graph here (empty rows and columns could be dropped, leaving an $|E|
\times |E|$ matrix, but our implementation does not do this since the empty
rows and columns are free in a sparse matrix representation):

\begin{equation}
K'_{(i,k),(k,j)} = \begin{cases}
    x \in \mathbb{F}\ \text{u.a.r.} & \text{iff}\ (i,k) \in E \wedge (k,j) \in E\\
    0 & \text{else}\,.
\end{cases}
\end{equation}

Now, we assign a vector $F_i \in \mathbb{F}^k$, where $k$ is the maximum vertex
degree, to each vertex $i$ and consider the system of equations defined by the
graph: the value of each vertex shall be the linear combination of its neighbors
weighted by the edge coefficients in $K$. To force a non-trivial solution, we
designate a source vertex $s$ and add pairwise orthogonal vectors to each of its
neighbors. For simplicity we use unit vectors in the respective columns of a
$k \times |V|$ matrix $P_s$ (same shape as $F$). So, we get the condition

\begin{equation}
\label{equ:def_f}
F = FK+P_s\,.
\end{equation}

This can be solved as
\begin{equation}
F = -P_s(\mathbb{I}-K)^{-1}\,.
\end{equation}

The work-intensive part is inverting $(\mathbb{I}-K)$, which can be done
explicitly and independently from $s$, to get a computationally inexpensive
all-pairs solution, or implicitly only for the vectors in $P_s$ for a
computationally inexpensive single-source solution.
To compute connectivity, we use the following theorem.
The scheme outlined in the following proof counts vertex-disjoint paths of any
length. 

\begin{theorem}
The size of the vertex cut set $c_{st}$ from $s$ to $t$ equals 
$\rank(FQ_t)$, where $F = -P_s(\mathbb{I}-K)^{-1}$ and $Q_t$ is a $|V| \times
k$ permutation matrix selecting $t$'s incoming neighbors. 
\end{theorem}

\begin{proof}
First, $c_{st} \leq \rank(FQ_t)$, because all non-zero vectors were injected
around $s$ and all vectors propagated through the cut set of $c_{st}$ vertices to
$t$, so there cannot be more than $c_{st}$ linearly independent vectors near
$t$.
Second, $c_{st} \geq \rank(FQ_t)$, because there are $c_{st}$ vertex-disjoint
paths from $s$ to $t$. Each passes through one of the $c_{st}$ outgoing
neighbors of $s$, which has one of the linearly independent vectors of $P_s$
assigned (combined with potentially other components). As there is a path
from $s$ to $t$ trough this vertex, on each edge of this path the component of
$P_s$ will be propagated to the next vertex, multiplied by the respective
coefficient in $K$. So, at $t$ each of the paths will contribute one orthogonal
component.
\end{proof}


To count length-limited paths instead, we simply use an iterative approximation
of the fixed point instead of the explicit solution. Since we are only
interested in the ranks of sub-matrices, it is also not necessary to actually
find a precise solution; rather, following the argument of the proof above,
we want to follow the propagation of linearly independent components through
the network. The first approach is simply iterating Equation~\ref{equ:def_f}
from some initial guess. For this guess we use zero vectors, due to $P_s$ in
there we still get nontrivial solutions but we can be certain to not introduce
additional linearly dependent vectors:

\begin{equation}\begin{aligned}
F_0 &= \big(0\big)\quad (k \times |V|)\\
F_l &= F_{l-1}K + P_s\,.
\end{aligned}\end{equation}

This iteration still depends on a specific source vertex $s$. For an all-pairs
solution, we can iterate for all source vertices in parallel by using more
dimensions in the vectors; we set $k = |V|$. Now we can assign every vertex a
pairwise orthogonal start vector, e.g., by factoring out $P_s$ and selecting
rows by multiplying with $P_s$ in the end. The intermediate products are now
$|V| \times |V|$ matrices, and we add the identity matrix after each step.
Putting all together gives 

\begin{equation}
c_{st} = \rank(P_s\underbrace{(((K+\mathbb{I})\cdot K + \mathbb{I})\cdot\ldots)}_{\text{$l$ times, precomputed}}Q_t)\,.
\end{equation}

%

The total complexity includes the $\mathcal{O}\left(|V|^3 l\right)$ operations to precompute the
matrix for a maximum path length of $l$ and $\mathcal{O}\left(|V|^2 k^3\right)$ operations for
the $\rank$ operations for all vertex pairs in the end, which will be the leading
term for the $k = \mathcal{O}\left((\sqrt{|V|}\right)$ (diameter 2) undirected graphs considered
here, for a total of $\mathcal{O}\left(|V|^{3.5}\right)$.

For the edge connectivity version, we use the edge incidence connection matrix
$K'$, and select rows and columns based on edge incidence, instead of vertex
adjacency. Apart from that, the algorithm stays identical, but the measured cut
set will now be a cut set of edges, yielding edge connectivity values. However,
the algorithm is more expensive in terms of space use and running time:
$\mathcal{O}\left(|E|^3 l\right)$ to precompute the propagation matrix.

\section{Details of Layered Routing}
\label{sec:app_Layers}

We extend the description of FatPaths' layered routing.

\subsection{Deriving Forwarding Entries}
\label{sec:app_forwarding}

An example scheme for deriving and populating forwarding tables is illustrated
in Listing~\ref{lst:layersEntries}.

\begin{lstlisting}[aboveskip=0em,abovecaptionskip=0em,belowskip=0em,float=*,label=lst:layersEntries,caption=\textmd{
Populating forwarding tables in FatPaths.}]
//|\textbf{Input:}| $E'_i,\ i \in \{1, ..., n\}$: the specification of each layer. $E'_i$ are router-router links in each layer $i$,
//produced by FatPaths layer construction algorithms specified in Listing |\ref{lst:layers}| or Listing |\ref{lst:layersFp}|.
//|\textbf{Output:}| $F_{i,s}$: a forwarding table in router $s \in V$ within layer $i\ (1 \le i \le n)$.
//$F_{i,s}[t] \in \{1, ..., k'\}$ is a port number (in router $s \in V$) that leads to router $t \in V$ (within layer $i\ (1 \le i \le n)$).

//Compute shortest paths between routers in each layer $i$. First, initialize the function $\sigma_i$.
//$\sigma_i(s,t)$ is a port number (in router $s$) that leads to router $t$ (in layer $i$); $d$ is an auxiliary structure.
foreach $i \in \{1, ..., n\}$ do: //For each layer...
  foreach $(s,t) \in V \times V$ do: //For each router pair...
    if $s$ == $t$ then:
      $d_{st} = 0$ //The distance between $u$ and $v$ is zero, if $u == v$.
      $\sigma_i(s,t) = s$
    else if $(s,t) \in E'_i$ then:
      $d_{st} = 1$ //The distance between directly connected routers is 1.
      $\sigma_i(s,t) = t$
    else:
      $d_{st} = +\infty$ //The initial distance between non-adjacent routers is infinity.
      $\sigma_i(s,t) = null$

  //For each layer, derive the routing functions:
  foreach $z \in |V|$ do:
    foreach $s \in |V|$ do:
      foreach $t \in |V|$ do:
        if $d_{st} > d_{sz} + d_{zt}$ then:
          $d_{st} = d_{sz} + d_{zt}$
          $\sigma_i(s,t) = \sigma_i(s,z)$

  //Once $\sigma_i$ is computed, populate a forwarding table. First, initialize the forwarding table.
  foreach $s \in V$ do: //For every router $s$:
    $F_s = \{\}$ //Initialize the forwarding table $F_s$.

  //Build a forwarding table within layer $i$, using the previously devised structures:
  foreach $s \in V$ do:
    foreach $t \in V, t \neq s$ do:
      $F_s[t] = \sigma_i(s,t)$
\end{lstlisting}

\subsection{Details of SPAIN Comparison Baseline}
\label{sec:app_SPAIN}

In FatPaths, we include SPAIN as a comparison baseline when generating
FatPaths routing layers, see Listing~\ref{lst:layersSpain}.
The implementation of SPAIN consists of two main stages: pre-computation of
paths exploiting redundancy in a given network topology and mapping the path
sets to VLANs. The time complexity of path computation and path layout
algorithm is together $O(|V|^2(|V| + |E|))$.

The first stage of the algorithm is based upon a per-destination VLAN
computation algorithm. At first, for every destination node, all paths
connecting it with other nodes are computed. A good path set should meet two
basic requirements -- it should exploit the available topological redundancy
and simultaneously contain as few paths as possible. The algorithm for path
computation greedily searches for $k \in \mathbb{N}$ shortest paths between
pairs of nodes, giving the preference to link-disjoint paths over alternative
paths, which could be possibly shorter but share links with at least one
already selected path. However, even though link-disjoint paths have the
advantage as a single link failure does not take down multiple paths, in some
circumstances it may be beneficial to consider a shorter path which is not
necessarily fully disjoint path but shorter~\cite{mudigonda2010spain}. The
paths are computed to be pairwise disjoint only for a certain destination node.

The next step in the per-destination VLAN computation algorithm uses graph
coloring techniques to assign the least possible number of VLANs to path sets
for a destination node. Afterwards, in order to minimize the number of VLANs,
they are being merged together using a greedy algorithm, as any two generated
subgraphs for different destinations can be combined into a single graph, if
and only if the resulting graph will not contain any loops. The input of the
algorithm is the set containing all computed subgraphs. For each graph, the
algorithm tries to merge it with as many other graphs as possible (checking
them in a randomized order), using a breadth-first search on every combined
graph to check whether it is acyclic or not. Unfortunately, the current design
of SPAIN lacks the flexibility to define an arbitrary number of layers.
Moreover, the greedy merging algorithm may result in creating layers, that will
be unbalanced in terms of size (number of edges per each graph).

\begin{lstlisting}[aboveskip=0em,abovecaptionskip=0em,belowskip=0em,float=*,label=lst:layersSpain,caption=\textmd{
Constructing routing layers in FatPaths when incorporating the SPAIN baseline.}]
$L$ = $\{E\}$ //Init a set of layers $L$, start with $E$.

foreach $u \in V $ do: //Compute VLANs per destination $u$.
  $P = \emptyset$ //Init a set of paths leading to $u$.

  foreach $v \in V$, $v \neq u$ do: //Compute $k$ paths from each vertex $v$ to $u$.
    $P_v = \emptyset$ //Init a set of paths from $v$ to $u$.
    $\forall e \in E$: $w(e) = 0$ //Init a matrix of edge weights.
 
    while $|P_v| < k$ do: //Add $k$ shortest paths to the set $P_v$.
      $p$ = shortest($E$, $u$, $v$, $w$) //Shortest undirected path from $v$ to $u$.
      if $p \in P_v$ then break //End the loop if no other path can be found.
      $P_v = P_v \cup \{p\}$ //Add the found path to the path set.
      foreach $edge \in path$ do: $w(e) += |E|$
    $P = P \cup P_v$

  $V_{color} = {v_1, \dots, v_{|P|}}$
  $E_{color} \neq \emptyset$
  foreach $p_i, p_j \in P$ do:
    if not vlan-compatible($p_i$, $p_j$) then: //the vlan-compatible function is defined below.
      $E_{color} = E_{color} \cup \{v_i, v_j\}$

  //Assign each vertex to a color - the assignment is given by mapping function:
  $(\#colors, mapping)$ = greedy-vertex-coloring($V_{color}, E_{color}$)
  for $k = 1, 2, \dots \#colors$ do:
    $E' = \{ e \in p_i : p_i \in P, mapping(v_i) = k \}$
    $L = L \cup E'$

//Merge greedily the graphs in $L$ while processing them in random order.
foreach $E_i, E_j \in L$, $E_i \neq E_j$:
  if acyclic($E_i \cup E_j$) then:
    $L = L \cup \{ E_i \cup E_j \} - {E_i} - {E_j}$

//The following function verifies if two paths can be merged within one graph,
//such that one can create a VLAN out of it (i.e., there are no cycles in the graph).
The details are in the SPAIN paper |\cite{mudigonda2010spain}|.
vlan-compatible($p_i= (u_1, u_2, \dots, u_l)$, $p_j= (v_1, v_2, \dots, v_m)$)
  foreach $u_i \in p_i$ do:
    if $u_i = v_j \in p_j$ then:
      if $u_{i+1} \neq v_{j+1}$ then:
        return false
  return true
\end{lstlisting}

\subsection{Details of PAST Comparison Baseline}
\label{sec:app_PAST}

We consider PAST~\cite{stephens2012past} for completeness as it does not
support (by default) multi-pathing between endpoint pairs for more performance.
Specifically, PAST uses per-address spanning trees to forward the traffic to
any endpoint in the network. Any given switch uses only one path for
forwarding the traffic to the given endpoint. It provides various approaches to
define the spanning trees, aiming at distributing the trees uniformly across
available physical links (i.e., using breadth-first search with random
tie-breaking for paths finding or implementing a non-minimal variant, inspired
by the Valiant load balancing).
A simple specification of generating routing layers in FatPaths when using
PAST can be found in Listing~\ref{lst:layersPast}.

Specifically, in contrast to SPAIN, which searches for link-disjoint paths to a
given destination node, PAST creates spanning trees for every address in the
network using one of three possible approaches: standard breadth-first search
algorithm, with the destination node being the root of the tree, breadth-first
search with random tie-breaking and breadth-first search, which weights its
random selection by considering how many endpoints each switch has as children,
summed across all trees built so far. The complexity of the bread-first search,
which affects directly the complexity of PAST layers creation algorithm, may be
expressed as $O (|V| + |E|)$.

Additionally a non-minimal, inspired by Valiant load balancing, variant of PAST
was designed. In this approach, before computing the spanning tree, the
algorithm selects randomly a switch, which plays the role of the root of
spanning tree, instead of the destination endpoint.

\begin{lstlisting}[aboveskip=0em,abovecaptionskip=0em,belowskip=0em,float=h,label=lst:layersPast,caption=\textmd{
Constructing routing layers in FatPaths when incorporating the PAST baseline (a non-minimal variant).}]
$L$ = $\{\}$ //Init a set of layers $L$.

//Create a spanning tree per each
//destination address:
foreach host $h$ do: 
  //Select a random intermediate switch.
  $s = \text{random}(V)$
  //Create spanning tree rooted at $s$ using BFS
  $E' = BFS(s, E)$
  $L = L \cup \{E'\}$
\end{lstlisting}

\subsection{Details of $k$-Shortest Paths Comparison Baseline}
\label{sec:app_k}

$k$-shortest paths~\cite{singla2012jellyfish} spreads traffic over multiple
shortest paths (if available) between endpoints. For the purpose of evaluation
we used Yen's algorithm with Dijkstra's algorithm to derive the $k$ shortest
loop-free paths. As the complexity of Dijkstra's algorithm may be expressed as
$O(|E|+|V|\log |V|)$, the complexity of computing $k$ paths between a pair of
endpoints using Yen's algorithm is equal to $O(k |V| (|E|+|V|\log |V|))$.

\section{Details of Evaluation Setup}
\label{sec:app-simulations-full-data}

We also provide details of the evaluation setup. 

\subsubsection{Topology Parameters}
\
We now extend the discussion on the {selection of key topology parameters}
(full details and formal descriptions are in the Appendix).
We select $N_r$ and $k'$ so that
considered topologies use similar amounts of hardware.
To analyze these amounts, we analyze the \emph{edge density}: a ratio between
the number of all the cables $\frac12 N_r k' + N_r p$ and endpoints $N = N_r
p$.
It turns out to be (asymptotically) constant for all topologies (the left plot
in Figure~\ref{fig:density}) and \emph{related to $D$}.
Next, higher-diameter networks such as DF require more cables. As explained
in past work~\cite{besta2014slim}: Packets traverse more cables on the way to their
destination.
We also illustrate $k$ as a function of $N$ (the right plot in
Figure~\ref{fig:density}). An interesting outlier is FT. It scales with $k$
similarly to networks with $D=2$, but with a much lower constant factor,
at the cost of a higher $D$ and thus more routers and cables. This implies that
FT is most attractive for small networks using routers with constrained $k$.
We can also observe the unique SF properties: For a fixed (low) number of
cables, the required $k$ is lower by a constant factor (than in, e.g., HX),
resulting in better $N$ scaling.

\begin{figure}[h]
\centering
\includegraphics[width=\columnwidth]{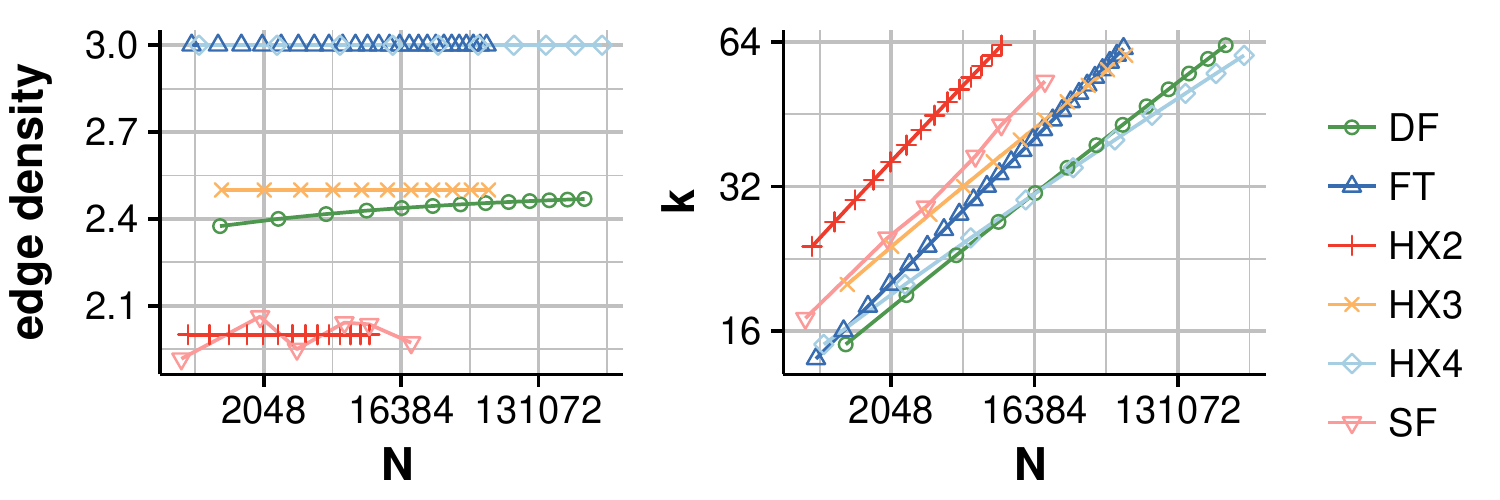}%
\vspaceSQ{-0.75em}
\caption{\textmd{Edge density
$(\allowbreak\text{number of edges})/(\allowbreak\text{number of endpoints}) =
(\frac{1}{2}N_r k' + N_r p)/N$ and radix~$k$ as a function of network size~$N$.
All cables, including endpoint links, are considered. This allows for a fair
comparison with fat trees. We use radix values $k \le 64$. }
%
%
}
\label{fig:density}
\end{figure}

\subsubsection{Baselines for Performance Validation}

Our evaluation is influenced by a plethora of parameters and effects, many of
which are not necessarily related to the core paper domain. Some of them may be
related to the incorporated protocols (e.g., TCP), others to the used traffic
patterns. Thus, we also establish baseline comparison targets and we fix
various parameters to ensure fair comparisons.
To characterize \textbf{TCP effects}, \textbf{one baseline is a star (ST)
topology} that contains a single crossbar switch and attached endpoints. It
should not exhibit any behavior that depends on the topology structure as it
does not contain any inter-switch links. We use the same flow distribution and
traffic pattern as in the large-scale simulation, as well as the same transport
protocols. This serves as an upper bound on performance. Compared to
measurements, we observe the lowest realistic latency and the maximum
achievable link throughput, as well as flow control effects that we did not
explicitly model, such as TCP slow start. There is no additional congestion
compared to measured data since we use randomized workloads.
Second, as a lower bound on \textbf{routing performance}, we show results for
flow-based \textbf{ECMP} as an example of a non-adaptive routing scheme, and
\textbf{LetFlow} as an example of an adaptive routing scheme. We also include
results of unmodified \textbf{NDP} (with oblivious load balancing) on FTs.

\subsection{Behavior of Flows in OMNet++}
\label{sec:app-flow-omnet}

First, we use
Figure~\ref{fig:app-star_lambda} to illustrate the behavior of long flows
(2MB) in the pFabric web search distribution in response to the flow arrival
rate $\lambda$ (flows per endpoint per second) on a 60-endpoint crossbar (tail
limited to 90\% due to the low sample size). The left plot shows how the
per-flow throughput decreases beyond $\lambda=250$, which is a sign of network
saturation; the right figure shows three samples of completion time
distributions for low, moderate, and high loads.

\begin{figure*}[t!]
\centering
\includegraphics[width=0.45\textwidth]{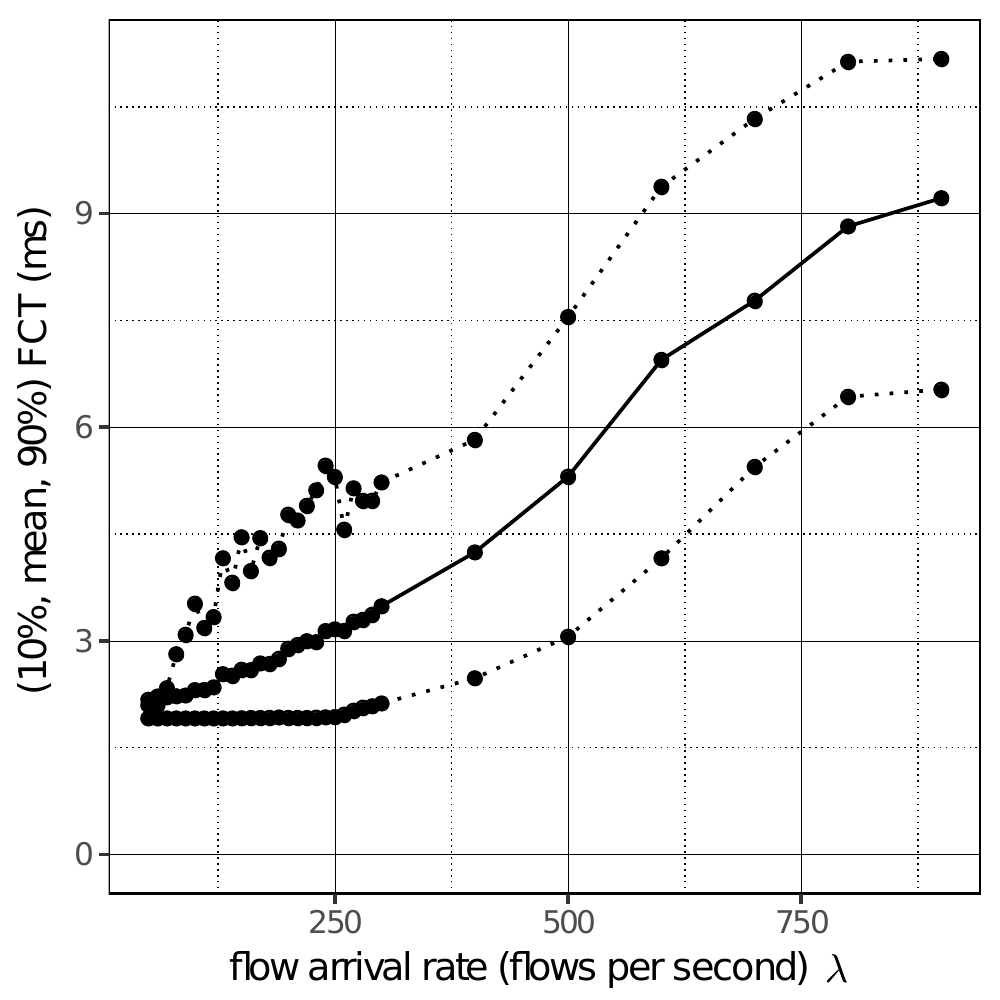}
\includegraphics[width=0.45\textwidth]{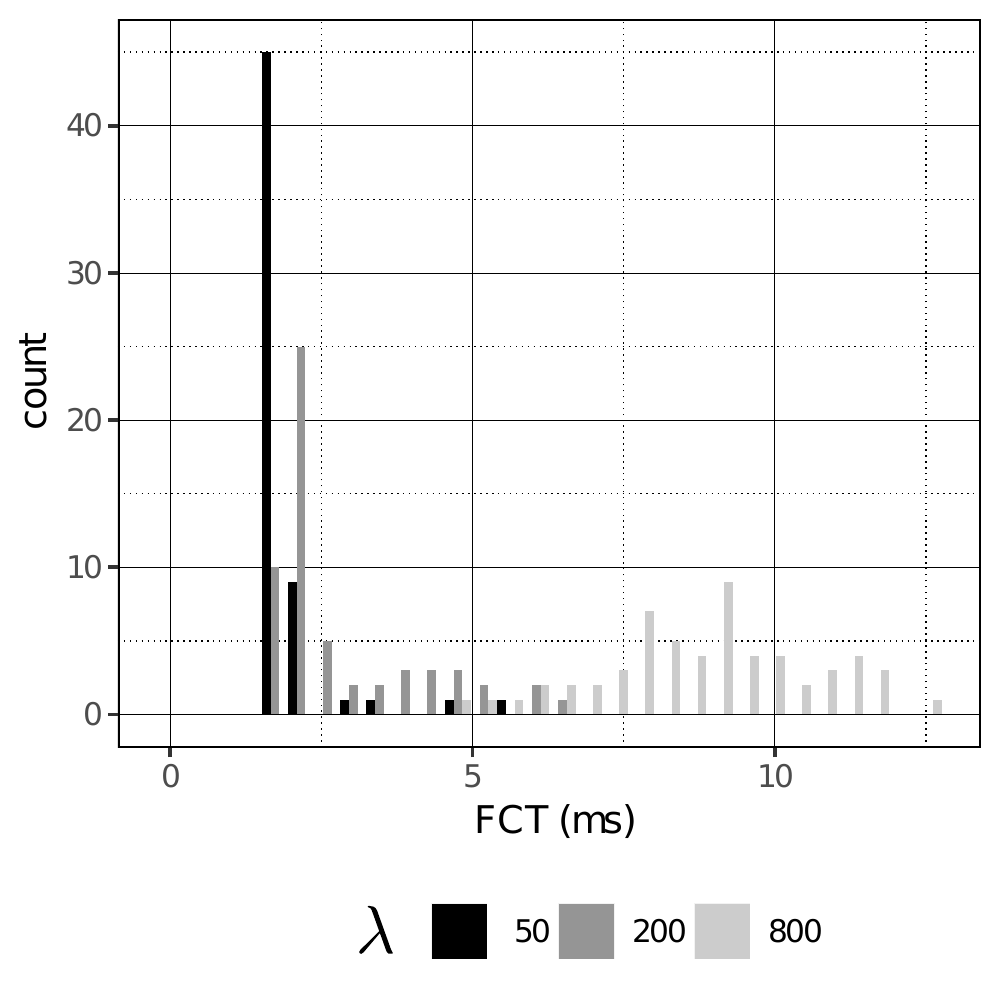}
\vspaceSQ{-1em}
\caption{\textbf{Behavior of long flows (2MB) in the pFabric web search distribution}
in response to the flow arrival rate $\lambda$ (flows per endpoint per second)
on a 60-endpoint crossbar (tail limited to 90\% due to the low sample size).}
\label{fig:app-star_lambda}
\end{figure*}

\subsection{Behavior of Flows in htsim}
\label{sec:app-flow-htsim}

In htsim simulations, on the star network as well as the
baseline $2x$-oversubscribed fat tree, we observe better performance compared
to the OMNet++ TCP simulations. This leads to a lower network load, which would
be misleading in topology comparisons. Therefore, we use $\lambda = 300$, where
the system starts to show congestion at the network level. At lower $\lambda$,
we only observe endpoint congestion (crossbar (Star) results equal fat tree
results), while at higher $\lambda$, the FCTs increase beyond the $2\times$ expected
from the oversubscription: the lower service rate leads to more concurrent
flows, decreasing the throughput per flow (see
Figure~\ref{fig:app-eta_lambda}).

\begin{figure*}[t]
\centering
\includegraphics[width=\textwidth]{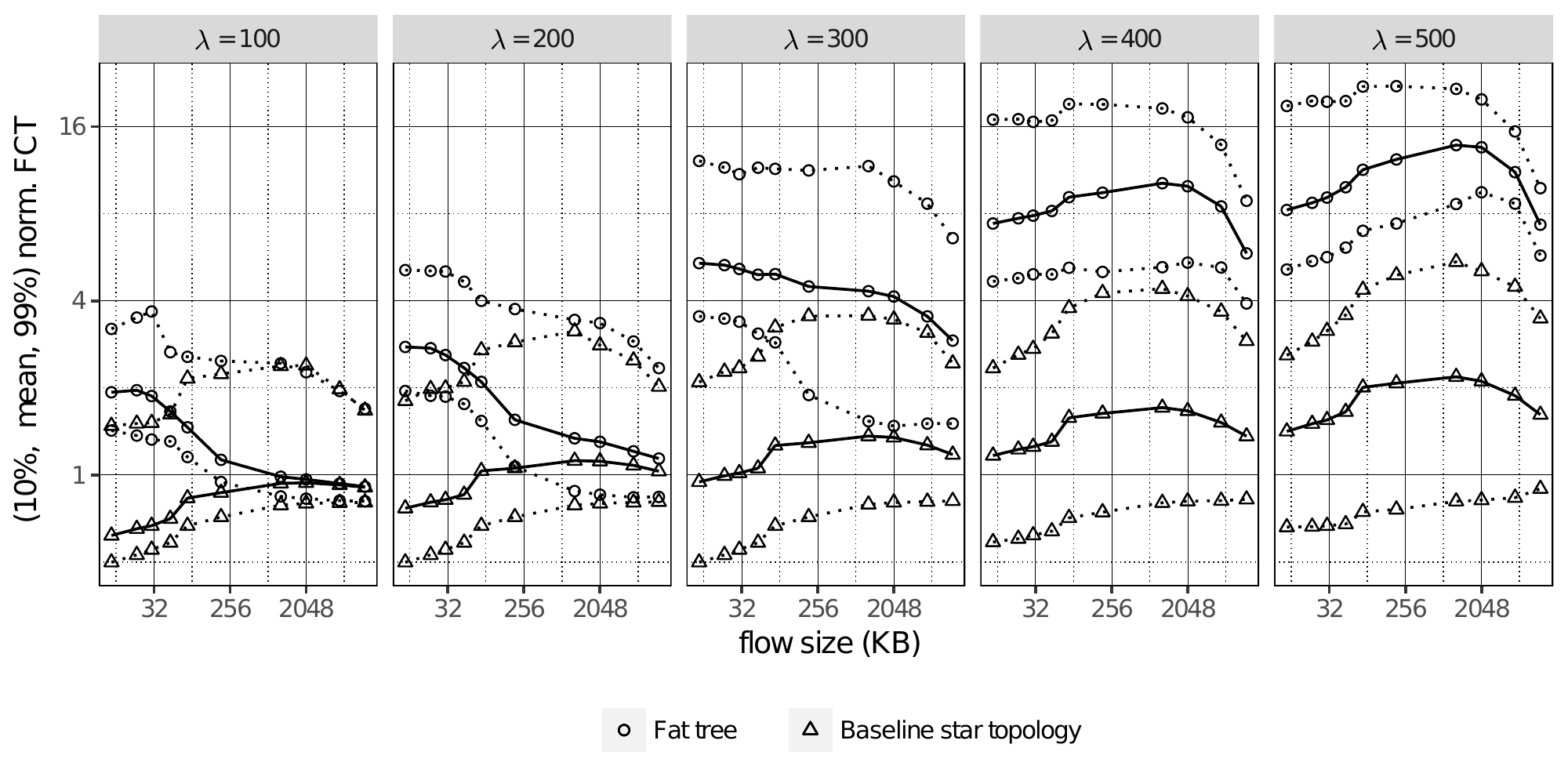}
\vspaceSQ{-3em}
\caption{\textbf{Influence of $\lambda$ on the baseline NDP simulations}. The FCT are
normalized with $\mu$=1GB/s and $\delta$=50us. At $\lambda = 100, 200$ the
oversubscription is not noticeable (similar long-flow behavior for crossbar
and fat tree), which indicates that the total network load is still low. The
difference in short-flow FCT is due to the drastically different network
diameters.} \label{fig:app-eta_lambda}
\end{figure*}


\subsection{Selection of TCP Parameters}

TCP retransmissions on packet loss are controlled by two separate timeouts: one
for the initial SYN and SYNACK handshake packets, and one that is used during
the flow. Both are guarded by an upper and lower limit. For the handshake
timeout, the lower limit is used in the beginning, increasing it up to the
upper limit on retries. For the normal limit, it is adapted in response to the
measured RTT but limited by the bounds, initially starting from a high value.

Since we usually do not see lost SYN packets, we did not optimize the handshake
timeouts. Most of the time, they simply have no effect at all. We did optimize
the normal retransmission timeouts, and observed that limiting the lower bound
can decrease performance at high load, while the upper bound does not have much
impact (again, this is because it is unlikely that a packet is lost before an
RTT estimate is produced, so this parameter is not used at all). The value of
200$\mu$s for the RTO lower bound is fairly high and can lead to performance
degradation, but it also models a limited timer granularity on the endpoints,
which makes low timeouts unrealistic. In the usual workload model considered in
this work, packet loss rates are low enough that the RTO does not have any
measurable impact, as long as the timeouts are not very high (with the INET
default value of 1s, a single flow experiencing a RTO can influence the mean
significantly).
The TCP retransmission parameters become more relevant if very sparse, and
therefore incomplete, layers are used, where packets are lost not only due to
congestion but also due to being non-routable. However, in this case we use
feedback via ICMP to trigger an immediate retransmission on a different layer,
therefore the RTO limits also have no siificant impact in this scenario.